\documentclass[11pt,a4paper]{article}
\usepackage{geometry}

\usepackage{color}
\usepackage{amsmath}
\usepackage{mathtools}
\usepackage{amssymb}
\usepackage{bbm}
\usepackage{amsthm}

\usepackage{csquotes}
\usepackage{authblk}

\usepackage{comment}

\usepackage{tikz}

\theoremstyle{definition} \newtheorem{definition}{Definition}[section]
\theoremstyle{plain} \newtheorem{theorem}[definition]{Theorem}
\theoremstyle{plain} \newtheorem{assumption}[definition]{Assumption}
\theoremstyle{plain} \newtheorem{proposition}[definition]{Proposition}
\theoremstyle{plain} \newtheorem{lemma}[definition]{Lemma}
\theoremstyle{plain} \newtheorem{corollary}[definition]{Corollary}
\theoremstyle{plain} \newtheorem{remark}[definition]{Remark}
\theoremstyle{definition} 

\numberwithin{equation}{section}
\DeclareMathOperator{\supp}{\mathrm{supp}}

\newcommand{\one}{\mathbbm{1}}

\newcommand{\dd}{\mathrm{d}}

\usepackage{amsmath}
\usepackage{amsfonts}
\usepackage{amssymb}
\usepackage{amsthm,amsfonts}
\usepackage{vmargin}
\usepackage{setspace}
\usepackage{xcolor}
\usepackage{amssymb}
\usepackage{mathtools}
\usepackage{esint}

\definecolor{vertfonce}{rgb}{0.20, 0.46, 0.25}
\definecolor{rougefonce}{rgb}{0.64, 0.09, 0.20}
\usepackage[breaklinks=true,
colorlinks=true,
linkcolor=rougefonce,
citecolor=vertfonce]{hyperref}
\usepackage{cleveref}

\usepackage{mathrsfs, enumerate, csquotes, color}

\usepackage[explicit]{titlesec}
\titleformat{\section}{\centering\Large\bfseries}{\thesection \ --}{0.7em}{\Large\bfseries #1}
\titleformat{\subsection}{\centering\large\bfseries}{\thesubsection \ --}{0.4em}{\large\bfseries #1}
\titleformat{\subsubsection}{\centering\bfseries}{\thesubsubsection \ --}{0.4em}{\bfseries #1}

\makeatletter

\@addtoreset{equation}{section}
\makeatother

\author{M. Olivieri\thanks{mo@math.ku.dk, olivieri.math@gmail.com}}

	\title{A Lee-Huang-Yang type expansion for the thermodynamic energy density of a dilute mixture of Bose gases}
	
	\affil{\small{Department of Mathematical Sciences, Universitetsparken 5\\ DK-2100 Copenhagen\\ Denmark}}

\usepackage{hyperref}
\begin{document}

	\maketitle

\begin{abstract}
We consider a dilute gas in 3D composed of two species of bosons interacting through positive inter-species and intra-species
pairwise potentials. We prove a second order expansion for the energy density in the thermodynamic limit. For the case of compactly supported, integrable potentials, we derive the correct second order of the expansion. If we make the further assumption of having soft potentials, we also derive the correct coefficient of the second order and the resulting formula is coherent with the physics literature. If we let the density and scattering length of one of the species go to zero, we obtain the Lee-Huang-Yang formula for one species of bosons. The paper also contains a proof of BEC for a mixture of bosons in a box with length scale larger than the Gross-Pitaevskii one.
\end{abstract}

\section{Introduction}

Mixtures of Bose gases offer a rich variety of phenomena which are of interest from both a physical and mathematical perspectives, due to the interplay between intra-species and inter-species interactions.
Such models can describe systems which have been studied in experiments involving, for example, Rubidium atoms $\prescript{87}{}{\text{Rb}}$ in different hyperfine states \cite{Rb87-3}, \cite{Rb87-4}, \cite{Rb87-2}, \cite{Rb87} and heteronuclear mixtures of Rubidium and Potassium $\prescript{41}{}{\text{K}}-\prescript{87}{}{\text{Rb}}$ \cite{KRb-1}, $\prescript{41}{}{\text{K}}-\prescript{85}{}{\text{Rb}}$ \cite{KRb-2}, $\prescript{37}{}{\text{K}}-\prescript{85}{}{\text{Rb}}$ \cite{KRb-3}, $\prescript{85}{}{\text{K}}-\prescript{87}{}{\text{Rb}}$ \cite{KRb-4}.
Furthermore, Guyer and Miller \cite{guyermiller}, \cite{miller} consider the study of the bosons-bosons mixture a natural starting point for the analysis of fermions-bosons mixtures, like the case of $\prescript{40}{}{\text{K}}-\prescript{87}{}{\text{Rb}}$ \cite{modugnomodugno}.

The model we use in this paper to describe these systems is the non-relativistic, many-body Hamiltonian for bosons. Specifically, we consider $N$ interacting bosons in a box $\Lambda_L := [-L/2,L/2]^3 \subseteq \mathbb{R}^3$. We assume there are $N_A$ and $N_B$ bosons of species $A$ and $B$, respectively, such that $N = N_A + N_B$ and the Hilbert space associated to two species of bosons in three dimension is 
\begin{equation}
\mathscr{H}_{N_A,N_B} := L^2_s(\Lambda_L^{N_A};\mathrm{d}x) \otimes L^2_s(\Lambda_L^{N_B};\mathrm{d}y) = \mathscr{H}_A \otimes \mathscr{H}_B,
\end{equation}
where $x=(x_1,\ldots,x_{N_A}) \in \Lambda_L^{N_A}, y=(y_1,\ldots, y_{N_B}) \in \Lambda_L^{N_B}$ are the position variables for the bosons of type $A$ and $B$, respectively. Observe that the wave-functions are symmetric separately in the $x$ and $y$ variables, but not in the interspecies exchange. We define the Hamiltonian  
\begin{align}
	 H_{N_A,N_B} =& \sum_{j=1}^{N_A} - \Delta_{x_j} + \sum_{1 \leq i < j \leq N_A} v_A(x_i-x_j)\nonumber \\
	  &+ \sum_{j=1}^{N_B} - \Delta_{y_j} + \sum_{1 \leq i < j \leq N_B} v_B(y_i-y_j) + \sum_{j=1}^{N_A} \sum_{k=1}^{N_B} v_{AB}(x_j-y_k),\label{def:HN}
\end{align}
acting on the space $\mathscr{H}_{N_A,N_B}$, where $v_A$ describes the potential internal at the particles of the subgroup $A$, $v_B$ the one internal at the type $B$, and $v_{AB}$ the inter-species potential between type $A$ and type $B$. The potentials are all assumed to be repulsive and we will denote by $a_A,a_B,a_{AB}$ the scattering lengths of $v_A,v_B,v_{AB}$, respectively (see Definition \ref{def:scattlength}).

We introduce the ground state energy of the system as the bottom of the spectrum of the Hamiltonian $H_{N_A,N_B}$
\begin{equation}
E_{N_A,N_B} := \inf \mathrm{Spec} (H_{N_A,N_B}),
\end{equation}
and the energy density in thermodynamic limit is given by
\begin{equation}\label{eq:gsedensity}
e_{3D}(\rho_A,\rho_B) := \lim_{\substack{N \rightarrow + \infty, \\ \rho_A = \frac{N_A}{L^{3}}, \rho_B = \frac{N_B}{L^{3}} = const}} \frac{1}{L^3} E_{N_A,N_B}.
\end{equation}
Letting $\rho := \rho_A +\rho_B$ denote the total density of the gas, in the dilute regime of small $\rho$, one finds in \cite{olessacha}, \cite{larsen}, \cite{petrov} the second order expansion for the energy density, which in our notation reads
\begin{align}\label{eq:expansMix}
e_{3D}(\rho_A,\rho_B) &\simeq  4\pi \Big((\rho_A^2 a_A +2 \rho_A \rho_B a_{AB} + \rho_B^2 a_B)\nonumber \\
&\quad +\frac{16 \sqrt{2}}{15\sqrt{\pi}}\sum_{\pm}\Big(\rho_A a_A+ \rho_B a_B \pm \sqrt{(\rho_A a_A - \rho_B a_B)^2 + 4 \rho_A\rho_B a_{AB}^2 } \,\Big)^{\frac{5}{2}}\Big),
\end{align}
where we assumed equal masses for the two species $m_A = m_B = \frac{1}{2}$ to simplify the formula. This shows how the expansion is, up to at least its second order of precision, universal, in the sense that it does not depend on the shape of the potentials, but only on their scattering lengths $a_A,a_B,a_{AB}$. Nevertheless, the aforementioned papers offer a derivation of the formula which lacks mathematical rigour.

The aim of this present paper is to provide a rigorous proof of the expansion \eqref{eq:expansMix} in the case of \textit{soft potentials} (see condition \eqref{def:softpotential}) and to obtain the second order (but with the wrong constant) in the case of general integrable potentials.

The expansion \eqref{eq:expansMix} is the analogous of the famous Lee-Huang-Yang (LHY) formula for the energy density in dilute regime of the single-species Bose gas:
\begin{equation}
e_{3D}(\rho) \simeq 4\pi \rho^2 a \Big( 1+\frac{128}{15\sqrt{\pi}}(\rho a^3)^{1/2}\Big).
\end{equation}
The proof of the Lee-Huang-Yang formula for the case of one species of bosons has required several decades since its first non-rigorous derivation introduced in \cite{LHY} using a pseudo-potentials method. The rigorous derivation of the main order has been obtained by Dyson in 1957 \cite{dyson} with an upper bound and in 1998 by Lieb and Yngvason \cite{LY} with a lower bound. The upper bound for the second order was derived in \cite{ESY}, where the authors used a quasi-free trial state. This strategy only gives the right LHY constant in the case of soft potential, and it inspired the proof of the upper bound in the present paper.

The right constant for the LHY expansion has been obtained in cornerstone papers via an upper bound in \cite{BCS}, \cite{YY} and via a lower bound in \cite{FS}, \cite{FS2}, the last paper including the case of the hard-core potential. To obtain an upper bound for the hard-core potential in 3D is still an open problem. The 2D case, though, has been solved in \cite{2DLHY}, where the authors proved a LHY-type expansion for the energy of the dilute Bose gas in 2D, giving upper and lower bounds including both the cases of integrable and hard-core potentials.  

Wu in \cite{Wu} showed how the universality of the expansion of the energy still holds at the third order, and calculated it. In a recent paper \cite{3rdSchlein}, the authors proved the upper bound for the third order expansion in thermodynamic limit. In \cite{wuGp} the third order expansion is proven in Gross-Pitaevskii regime. A rigorous proof for the lower bound in thermodynamic regime is still lacking.

Expansions for the free energy density at low temperature have as well been proven in \cite{freeEnCPHM}, \cite{HHNST} and \cite{haberberger2024upper}. In particular, this last paper introduces a method to combine the renormalization of the potential with the Neumann localization, which has inspired the proof of the lower bound of the present paper.

The mixtures of bosons have received an increasing attention in the Mathematical Physics community. At the best of our knowledge, the most recent paper containing the proof of an expansion for the ground state energy of such systems is \cite{olgiatiMix}. The authors derive the expansion of the ground state energy of a trapped two-components gas in both the Gross-Pitaevskii and mean-field regimes. In the former, the energy converges to the minimum of the Gross-Pitaevskii functional; in the latter, to that of the Hartree functional, with the second-order correction described by the lowest eigenvalue of a Bogoliubov Hamiltonian. However, these results do not address the thermodynamic regime. 
Other recent papers rigorously studying the mixtures of bosons are, for instance, \cite{dirkMix}, \cite{jinyeopDyn}, \cite{MichMix}, \cite{michelangeli_mean-field_2017}, \cite{olgiatiMix3} for the convergence of the dynamics, \cite{falcoFragment}, \cite{jinyeopFragment} for fragmented condensation of identical particles with spin.

In this work, beyond establishing a LHY-type, second order expansion of the energy of a mixture of bosons, we give the the following contributions:
\begin{itemize}
\item We prove BEC (Bose-Einstein Condensation, Proposition \ref{propos:BEC}) for the two components of a mixture of bosons at length scales larger than the healing length $\ell_{\text{GP}} = (\rho \bar{a})^{-1/2}, \bar{a} = \max\{a_A,a_B,a_{AB}\}$, corresponding to the Gross-Pitaevskii regime. The condensation estimate is obtained bounding the number of excited particles $n_+^A,n_+^B$ of type $A$ and $B$, respectively, outside the condensate, for states $\Psi$ at low energy:
\begin{equation}
\frac{\langle n_+^A\rangle_{\Psi}}{N}, \frac{\langle n_+^B\rangle_{\Psi}}{N} \leq \frac{\langle n_+^A + n_+^B\rangle_{\Psi}}{N} \leq (\rho \bar{a}^3)^{\frac{1}{17}-\frac{1}{500}} \ll 1.
\end{equation}
\item We prove the Neumann localization in Appendix \ref{app:localization} allowing to compare the energy of the thermodynamic box with the energy of the system localized in small boxes with larger length scale than $\ell_{\text{GP}}$ for a mixture of bosons.

\item We introduce a two-species Bogoliubov transformation \eqref{def:BogTransf} and rigorously minimize the Bogoliubov functional
\begin{multline*}
\mathscr{F}(\alpha, \gamma) = \int_{\mathbb{R}^3} \big((k^2 + \rho_{A,0} \widehat{g}_A(k)) \gamma^{AA}_k + (k^2 + \rho_{B,0} \widehat{g}_B(k)) \gamma^{BB}_k\big) \mathrm{d} k\\
 + \int_{\mathbb{R}^3} \rho_{A,0} \widehat{g}_A(k) \alpha^{AA}_k + \rho_{B,0} \widehat{g}_B(k) \alpha^{BB}_k + 2 \sqrt{ \rho_{A,0} \rho_{B,0}}\,\widehat{g}_{AB}(k)( \gamma^{AB}_k + \alpha^{AB}_k) \mathrm{d}k,
\end{multline*}
in the proof of the upper bound in Section \ref{sec:uppbound}, also giving the explicit expressions of the minimizers \eqref{eq:alphagammaexplicit}.
\end{itemize}

These are important results in their own right, which we hope will serve as useful tools for future research in the context of bosonic mixtures and related topics.

The expansion \eqref{eq:expansMix} is obtained by proving an upper bound and a lower bound. 
\begin{itemize}
\item In this first section we present the main result in Theorem \ref{thm:main} with its Corollary \ref{cor:mainresult}, where we show the two consequences: for integrable potentials we obtain the LHY-type expansion \eqref{exp:rightorder} with the correct second order, for soft potentials we derive \eqref{exp:rightcnst} with the right constant, namely giving \eqref{eq:expansMix}. The proof of the main theorem is in Section \ref{sec:uppbound} for the upper bound and the remaining Sections \ref{sec:lowerbound}-\ref{sec:rhoclose} for the lower bound.
\item In Section \ref{sec:uppbound}, we construct the trial state as a quasi-free state. To diagonalize the effective quadratic Hamiltonian arising from our calculations, we introduce a Bogoliubov transformation for two-species bosons and derive the expansion by minimizing the previously defined Bogoliubov functional.

\item In Section \ref{sec:lowerbound}, we reduce the proof of the lower bound to the estimate of the localized energy in small boxes in Theorem \ref{thm:lowerloc}, thanks to Appendix \ref{app:localization}, and prove this last theorem by referring to the results in the following sections.

\item A fundamental step for the proof of Theorem \ref{thm:lowerloc} is the renormalization of the potential of Lemma \ref{lem:potential_splitting} presented in Section \ref{sec:renorm}. This technique allows to soften the potentials by approximating them with the relative $g$'s (defined in \eqref{eq:scat_defs}) reabsorbing the errors made in this way in the positive $\mathcal{Q}_4$ terms which can be eliminated in a lower bound.

\item In Section \ref{sec:spgap} we extract the spectral gaps from the kinetic energy (Proposition \ref{prop:Hamgapbound}), which are extremely useful positive terms used to bound many of the error terms obtained from the following calculations. Key ingredients of the proof are the condensation estimate in Proposition \ref{propos:BEC} (proven in Appendix \ref{app:BEC}) and the localization of large matrices (proven in Appendix \ref{app:largematrix}), which allow to restrict the action of the Hamiltonian to states with low momenta excitations $n_+^L \ll N$. 

\item We symmetrize the Hamiltonian (Lemma \ref{lem:boundGapTilde}) to make its Fourier coefficients diagonal in the Neumann basis when we rewrite it in momenta space in the second quantization (Proposition \ref{prop:secondquantHam}). In this Section \ref{sec:sym} there are some important technicalities playing a fundamental role: here is the only point where we need the non-increasing assumption \eqref{cond:nonincreasing} on the potentials to estimate the errors coming from the symmetrization. We also artificially add some negative terms  $\mathcal{G}_{\text{conv}}, \mathcal{G}_{\eta}$ in Corollary \ref{cor:secondquant}, to obtain convexity of the energy functional \eqref{eq:Functionalconvex} (proven in Appendix \ref{app:Fconvexfunctional}) and to estimate some errors coming from the cubic term \eqref{eq:whereGetaisneeded} in the non-soft potential case and for the approximation of the Bogoliubov sum with its integral (Lemma \ref{lem:calcBogInt}).

\item In Section \ref{sec:cnumber} we perform the c-number substitution suggested by Bogoliubov in \cite{Bog} and isolate the effective, quadratic Bogoliubov Hamiltonian $\mathcal{K}^{\text{Bog}}$, which can be diagonalized as shown in Appendix \ref{app:bogint}. The outcome of the diagonalization are a positive Hamiltonian $\mathcal{K}^{\text{diag}}$ and a sum giving the correct LHY correction, as shown in Appendix \ref{app:calcBogint}. 

\item The c-number substitution expands the Hamiltonian on a basis of coherent states depending on the parameter $z \in \mathbb{C}^2$, where $|z|^2$ corresponds to the number of bosons in the condensate. In Section \ref{sec:rhofar} we show how, when $|z|^2$ is far from $N$, there is an excess of positive energy from the spectral gap (kinetic energy) which allows an easier bound on the cubic and other remaining terms. In Section \ref{sec:rhoclose} we treat the case $|z|^2\simeq N$. This requires more careful estimates: Lemma \ref{lem:approxQ3-Q3low} and Lemma \ref{lem:softpairs} allow to extract from the cubic terms $\mathcal{Q}_3$ the contribution from the soft pairs which are fundamental to refine the estimates and derive the right LHY coefficient. In Proposition \ref{prop:killQ3} we bound the energy coming from the interaction of the soft pairs by the excess of quadratic Hamiltonian $Z_2^{\text{ext}}$ and the high momenta part of the diagonal Hamiltonian $\mathcal{K}^{\text{diag}}$. Observe that here, in order to estimate the error \eqref{calc:Tquad4} to make the right LHY coefficient emerge, the soft potential assumption \eqref{def:softpotential} is fundamental.
\end{itemize}

For future perspectives, it would be interesting to refine the estimates in order to obtain the full second order expansion \eqref{eq:expansMix} without the soft-potential assumption \eqref{def:softpotential}. This should be possible for the upper bound by defining a trial state analogous to the one introduced in \cite{BCS} adapted to the two-species case. For the lower bound, improved estimates on the cubic term would be needed. We plan to come back to this problem and its extension to more than two species of bosons in a future work. 

Once the expansion \eqref{eq:expansMix} is established for general integrable potentials, a natural question arising would be the extension for the lower bound to hard-core and singular potentials.  We are confident that the approximation technique of singular potentials via integrable potentials for a lower bound introduced in \cite[Section 3]{freeEnCPHM} may be adapted to the two species case. It is still unclear how to solve the upper bound even for the one species case. The energy expansion for the gas in 2D is as well a possible direction of research, possibly adapting the strategy proposed in \cite{2DLHY} for the one species. 

Energy expansions in thermodynamic limit have been studied as well for the fermionic case \cite{FalGiaPo}, \cite{GIACOMELLI2023110073}, \cite{giacomelli2024optimallowerboundlow}, \cite{LAURITSENopt}, \cite{lauritsen2024groundstateenergydilute}, \cite{LAURITSEN2024110320}, \cite{SolFermi}. We hope the present work could be useful also for the study of the systems composed by fermions-fermions and fermions-bosons mixtures. 

\paragraph{Acknowledgements:}
The author was supported by the European Union’s HORIZON-MSCA-2022-PF-01 grant agreement, project number: 101103304 (UniBoGas). Funded by the European Union. Views and opinions expressed are those of the author only and do not necessarily reflect
those of the European Union or the Research Executive Agency. Neither the European Union nor the granting authority can be held responsible for them.

The author thanks S. Fournais for the interesting discussions and suggestions, which have helped improving the estimate of the $\mathcal{Q}_3$ terms and the extraction of the Bogoliubov functional in the upper bound. The author is also thankful to G. Ciccarello, T. Girardot, L. Junge, L. Morin and J. P. Solovej for the time they dedicated to hear of some of the problems involved in the strategy of the proof of the paper and the suggestions they provided. 

\begin{center}
\includegraphics[scale=0.20]{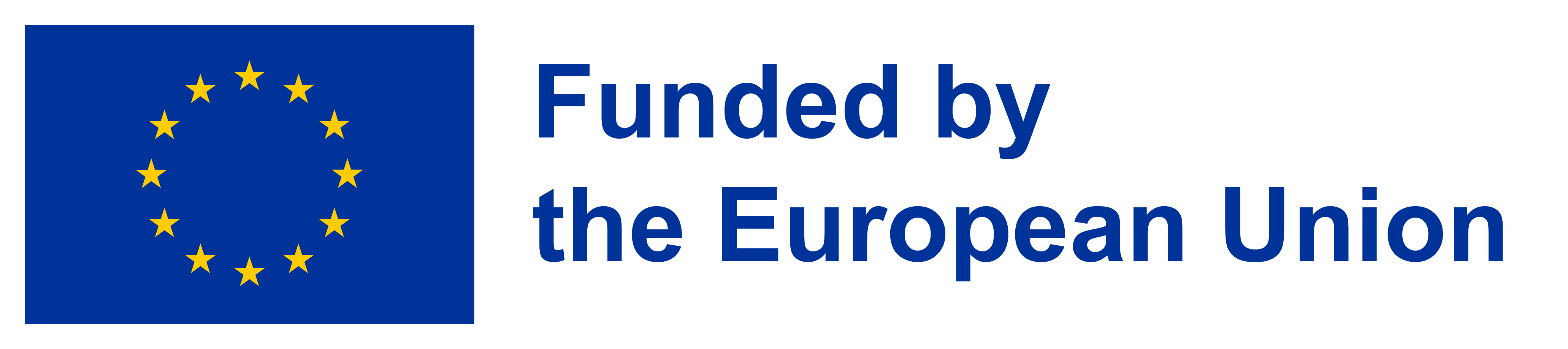}
\end{center}

\bigskip

\subsection{Assumptions and notation}

The interactions in the system are given by pairwise potentials. It is therefore only natural that the energy of the system is tightly connected to the two-body problem of the potentials. In this regard, the \textit{scattering length} is a fundamental quantity which we define following \cite{greenbook}.

\begin{definition}\label{def:scattlength}
Let $v: \mathbb{R}^{3} \to \mathbb{R}^+$ be measurable and radial with support in $B(0,R)$. The scattering length $a=a(v)$ is defined as
	\begin{equation}\label{variational definition of the scattering length}	4\pi a=\inf\Big\{\int_{\mathbb{R}^3}\vert \nabla \varphi\vert^2+\frac{1}{2}v|\varphi|^2 \dd x\, \Big\vert \, \varphi\in \dot{H}^1(\mathbb{R}^3),\quad \lim_{\vert x\vert \rightarrow \infty}\varphi(x)=1\Big\},
	\end{equation}
\end{definition}
	By testing the functional on $\max\{0, 1-\frac{R}{|x|}\}$, we find $a\leq R$. Some standard properties are that $a$ is an increasing function of $v$ and that the minimizer $\varphi$  solves the corresponding Euler-Lagrange equation
\begin{equation} \label{scattering_equation}
	-\Delta \varphi+\frac{1}{2}v\varphi =0,
\end{equation}
in a weak sense. By Newton's theorem
\begin{equation}\label{eq.simple form of phi}
	\varphi(x)=1-\frac{a}{\vert x\vert}, \quad \text{for $\vert x\vert \geq R$},
\end{equation}
and furthermore $\varphi$ is non-decreasing, non-negative and radial. We also introduce the following quantities which we are going to use throughout all the paper
\begin{align}\label{eq:scat_defs}
	\omega := 1- \varphi,
	\qquad g := v \varphi=v(1-\omega).
\end{align}
Clearly,
\begin{align}\label{eq:ScatOmega}
	-\Delta \omega = \frac{1}{2} g, \qquad \text{and} \qquad 
	\widehat{g}(0)=	\int_{\mathbb{R}^3} g \,\dd x = 8\pi a.
\end{align}
By the radial symmetry of $v$ and $g$, we have that $\widehat{v}(k), \widehat{g}(k)$ are real functions.

We list below the possible assumptions we consider for the potentials.

We are going to use the following notation for the total density of the gas
\begin{equation}
\rho := \rho_A + \rho_B.
\end{equation}

\begin{assumption}[Assumptions on the potentials]\label{cond:potentials}

We consider repulsive, spherically symmetric potentials 
\begin{align}
&v_A, v_B, v_{AB} \in L^1(\mathbb{R}^3), \qquad v_A, v_B, v_{AB} \geq 0,\label{cond:integrablepotential}\\
 & v_A,v_B,v_{AB} \text{ non- increasing}, \label{cond:nonincreasing}
\end{align}
with compact supports, and the rays of their supports being $R_A,R_B,R_{AB}>0$ with obvious meaning of the symbols. We denote by $R := \max\{R_A,R_B,R_{AB}\}$.
We assume the \textit{miscibility condition}
\begin{equation}\label{cond:misc}
a_{AB}^2 \leq a_A a_B.
\end{equation}

We denote by 
\begin{equation}
\bar{a} := \max \{a_A,a_B,a_{AB}\}, \qquad \underline{a} := \min \{a_A,a_B,a_{AB}\},
\end{equation}
and make the following assumptions: there exist constants $C_a,C_R >1$, and $C_1>0, \eta\geq 0$ such that 
\begin{equation}\label{cond:aRa}
\bar{a} \leq C_a\underline{a},
\end{equation}
\begin{equation}\label{cond:L1normv}
\|v_A\|_1,\|v_B\|_1, \|v_{AB}\|_1 \leq C_1 \bar{a}.
\end{equation}
\begin{equation}\label{cond:R}
R \leq C_R (\rho \bar{a}^3)^{-\eta} \bar{a}.
\end{equation}
\end{assumption}

Condition \eqref{cond:misc} guarantees that the repulsion between the two species is not effectively stronger than the internal repulsions inside the same species and that the two gases can spatially mix (see also Remark \ref{rem:assmpts}). Condition \eqref{cond:aRa} gives us a comparison bound between the several scattering lengths. Condition \eqref{cond:R} includes both the case of ray of supports $R$ independent of $\rho$, as well as the case of small divergence, as $\rho \rightarrow 0$, needed for the soft potentials introduced below.

Let us introduce the following parameters
\begin{equation}\label{eq:defdeltas}
\delta_A := |\widehat{v}_A(0) - \widehat{g}_A(0)|, \qquad \delta_B := |\widehat{v}_B(0) - \widehat{g}_B(0)|, \qquad \delta_{AB} := |\widehat{v}_{AB}(0) - \widehat{g}_{AB}(0)|,
\end{equation}
and
\begin{equation}\label{eq:defdelta}
\delta:= \max\{\delta_A, \delta_B, \delta_{AB}\}.
\end{equation}

\begin{definition}
\begin{itemize}
\item Let $S=S(\rho_A,\rho_B) \in \mathbb{R}$ and $T = T(\rho_A,\rho_B) \in \mathbb{R}$. We say that 
\begin{equation}
S \ll T \quad \Longleftrightarrow \quad \exists \sigma > 0:  \quad S \leq  (\rho \bar{a}^3)^{\sigma} T.
\end{equation}
\item We call one of the potentials $v_{\#}, \# \in \{A,B,AB\}$ a \textit{soft potential} if 
\begin{equation}\label{def:softpotential}
\delta_{\#} \ll a_{\#}. 
\end{equation}
\end{itemize}
\end{definition}

\subsection{Main results}
We are ready to present the main theorem of the paper. Recall the definition of the ground state energy density \eqref{eq:gsedensity}. We make use of the following parameter, for $\eta\geq 0, \nu >0$,
\begin{equation}
\sigma_{\eta}= \begin{cases}
4\eta + \nu, \quad & \text{if } \eta > 0,\\
0, \quad & \text{if } \eta = 0.
\end{cases}
\end{equation}

\begin{theorem}\label{thm:main}
Let $\eta  \geq 0, \nu >0$ be such that $v_A,v_B,v_{AB}$ are potentials satisfying Assumption \ref{cond:potentials}. There exists a constant $C>0$ such that, for $\rho \bar{a}^{3}\leq C^{-1} $ with
\begin{equation}\label{eq:deltaerrorstatement}
\delta_A,\delta_B\leq C\bar{a} (\rho \bar{a}^3)^{\eta}, \qquad \delta_{AB} \leq C \bar{a} (\rho \bar{a}^3)^{\sigma_{\eta}},
\end{equation}
then
\begin{align}
|e_{3D}(\rho_A, \rho_B) - E_{\text{main}} - E_{\text{LHY}} | \leq  C (\rho \bar{a} )^{5/2}  (\rho \bar{a}^3)^{\eta},\label{LHYformulaMixture}
\end{align}
where 
\begin{align*}
E_{\text{main}} &:= 4\pi (\rho_A^2 a_A +2 \rho_A \rho_B a_{AB} + \rho_B^2 a_B),\\
E_{\text{LHY}} &:= (\rho^2_A a^2_A + 2\rho_A \rho_B a_{AB}^2 +\rho_B^2 a_B^2)^{5/4} I_{AB}(\rho_A,\rho_B),
\end{align*}
with
\begin{align}
I_{AB}(\rho_A,\rho_B) &:=\frac{2\sqrt{2}(8 \pi)^{5/2}}{15\pi^2}(\mu_+^{5/2}(\rho_A,\rho_B) +\mu_-^{5/2}(\rho_A,\rho_B)) = \mathcal{O}(1),\label{def:mu&xi}\\
\mu_{\pm} &:= \frac{1}{2} \big( \sqrt{1+\xi_{AB}} \pm \sqrt{1-\xi_{AB}}\big),\qquad
\xi_{AB} := \frac{2 \rho_A \rho_B (a_A a_B - a_{AB}^2)}{\rho_A^2 a_A^2 + 2\rho_A \rho_B a_{AB}^2+ \rho_B^2 a_B^2}.\label{def:xi}
\end{align}
The sum $E_{\text{main}} + E_{LHY}$ is exactly equal to formula \eqref{eq:expansMix}.
\end{theorem}

Observe that condition \eqref{eq:deltaerrorstatement} is always fulfilled for $\eta = 0$. Also, in this case $E_{LHY}$ is of the same order of the error term, as we can see in Corollary \ref{cor:mainresult}.

\begin{proof}
The proof is obtained combining the proofs for the upper and lower bounds in the following sections. 
The bounds on $\delta_A, \delta_B, \delta_{AB}$ fit the assumptions of Theorems \ref{thm:mainupperbound} and \ref{thm:lowerloc}.
\begin{itemize}

\item \textit{Upper bound}: We consider the Hamiltonian $H_{N_A,N_B}$ with periodic boundary conditions on $\Lambda_L$. We can extend its action as an operator $\mathcal{H}$ introduced in \eqref{secondquantHam:upperbound} on the Fock space $\mathscr{F}(\Lambda)$ and we have therefore
\begin{equation}
E_{N_A,N_B} = \inf_{\substack{\Psi  \in \mathcal{D}_{N_A,N_B} \\ \|\Psi\|_{\mathscr{F}} = 1}} \langle \Psi, \mathcal{H}\Psi \rangle,
\end{equation}
where the infimum is taken over 
\begin{align*}
\mathcal{D}_{N_A,N_B} := \{\Psi \in \mathscr{F}(\Lambda)\,&|\,\Psi = \Psi_{N_A} \otimes \Psi_{N_B}, \Psi_{N_j}^{(n)} = \one_{n=N_j}\psi_{N_j}, j= A,B,\\
&\quad \psi_{N_A}\in C_0^{\infty}(\mathbb{R}^{3N_A}), \psi_{N_B} \in C_0^{\infty}(\mathbb{R}^{3N_B})\}.
\end{align*}
By Theorem \ref{thm:mainupperbound}, there exists $\Psi \in \mathscr{F}(\Lambda)$ such that it is non-zero only in the $(N_A,N_B)-$sector and for which \eqref{eq:energyupperboundTheorem} holds. Dividing by the volume $|\Lambda|$ and taking the thermodynamic limit we have an upper bound matching \eqref{LHYformulaMixture}.
\item \textit{Lower bound}: We realize $H_{N_A,N_B}$ as a self-adjoint operator on its domain with Neumann boundary conditions on $\Lambda_L$. The problem is localized in smaller boxes, for $y \in \mathbb{R}^3$
\begin{equation}\label{eq:locboxes}
\Lambda_{\ell}(y) = \prod_{i=1}^3[y_i-\ell/2,y_i +\ell/2], \qquad \ell = K_{\ell}(\rho \bar{a})^{-1/2}\ll L,
\end{equation}
with $K_{\ell} >0$, such that 
\begin{equation}
\Lambda_L = \bigcup_{j\in L \mathbb{Z}^3}^{M}\Lambda_{\ell}(j), \qquad M = L^3\ell^{-3} \in \mathbb{N}.
\end{equation} 
In Appendix \ref{app:localization} there is the proof of how to obtain Theorem \ref{thm:main} from Theorem \ref{thm:lowerloc}, which gives a lower bound on the energy on each box $\Lambda_{\ell}(j)$. Theorem \ref{thm:lowerloc} is proven in Section \ref{sec:lowerbound}.
\end{itemize}
\end{proof}

The theorem has different consequences depending on the choice of $\delta$ \eqref{eq:defdelta} and the size of the boxes $\ell$ (see \eqref{eq:locboxes})
\[\ell = K_{\ell}(\rho \bar{a})^{-1/2},\]
 where the energy is localized in the lower bound, dependent on $\eta \geq 0$ (the power in the error of the theorem). Indeed, we set $K_{\ell} := (1000C)^{-1}(\rho \bar{a}^3)^{-2\eta}$ and, recalling the definition of $C_1$ in \eqref{cond:L1normv}, we have the following cases:
\begin{itemize}
\item \textit{Integrable potentials}: for $\eta =0$,
\begin{equation}\label{assmptdelta:1}
 K_{\ell}= (1000C)^{-1},  \qquad \text{and} \qquad \delta \leq C_1 \bar{a},
\end{equation}
\item \textit{Soft potentials}: for $\eta > 0$,
\begin{equation}\label{assmptdelta:2}
K_{\ell}  \simeq C (\rho \bar{a}^3)^{-\eta}, \qquad \text{and} \qquad \delta_A,\delta_B \leq  C (\rho \bar{a}^3)^{\eta}\bar{a}, \qquad \delta_{AB} \leq C\bar{a} (\rho \bar{a}^3)^{4\eta + \nu}.
\end{equation}
\end{itemize}
\begin{corollary}\label{cor:mainresult}
Under the same assumptions of Theorem \ref{thm:main}, we have
\begin{itemize}
\item If \eqref{assmptdelta:1} holds, then 
\begin{equation}\label{exp:rightorder}
|e_{3D}(\rho_A, \rho_B) -  4\pi (\rho_A^2 a_A +2 \rho_A \rho_B a_{AB} + \rho_B^2 a_B) |\leq  C (\rho \bar{a} )^{5/2};
\end{equation}
\item If \eqref{assmptdelta:2} holds, then
\begin{multline}\label{exp:rightcnst}
\Big|e_{3D}(\rho_A,\rho_B) -  4\pi \Big((\rho_A^2 a_A +2 \rho_A \rho_B a_{AB} + \rho_B^2 a_B) \\
 +\frac{16 \sqrt{2}}{15\sqrt{\pi}}\sum_{\pm}\Big(\rho_A a_A+ \rho_B a_B \pm \sqrt{(\rho_A a_A - \rho_B a_B)^2 + 4 \rho_A\rho_B a_{AB}^2 } \,\Big)^{\frac{5}{2}}\Big)\Big|\\ \leq C (\rho \bar{a})^{5/2}(\rho \bar{a}^3)^{\eta}.
\end{multline}
\end{itemize}
\end{corollary}

\begin{remark}[Soft potentials for $\eta >0$]
An example of class of potentials satisfying the soft potential condition \eqref{def:softpotential} are those in the form 
\begin{equation}\label{eq:brietzkesoft}
v_R(x) = \frac{1}{R^3} v_1\Big( \frac{|x|}{R}\Big), \qquad \supp v_1 \subseteq [0,1], \qquad v_1 \in L^1\cap L^{\infty}(\mathbb{R}),
\end{equation}
with $R = C a(\rho a^3)^{-\sigma}$, for $\sigma>0$, $a$ being the scattering length of the potential $v_R$.
 As explained in \cite{brietzke}, \cite{brietzkesolovej}, for these potentials there is a Born expansion  
\begin{equation}
8\pi a = \widehat{v}_R(0) + \sum_{k=1}^{\infty} a_k,
\end{equation}
which is convergent and each term $a_k$ is proportional to $R^{-k}$. Therefore we can write
\begin{equation}\label{eq:born}
8\pi a - \widehat{v}(0) = a_1 + a_2 +  \mathcal{O}(a^4 R^{-3}) = \mathcal{O}(a^2 R^{-1}),
\end{equation}
and since $R$ depends on a negative power of $\rho$, then the scattering length can be approximated with a small error by $(8\pi)^{-1}\widehat{v}_R(0)$ and the potential $v_R$ is soft.

\begin{itemize}
\item For the upper bound, as it is clear from the error in Theorem \ref{thm:mainupperbound}, we derive the correct constant for the second order for $\eta>0$ under the condition
\begin{equation}
\delta  \leq C\bar{a} (\rho \bar{a}^3)^{\eta}, 
\end{equation}
that is, when all the three potentials $v_A,v_B,v_{AB}$ are soft with the same rate (less restrictive than \eqref{eq:deltaerrorstatement}). In order to achieve this, we can consider potentials in the form \eqref{eq:brietzkesoft}, for which $\delta = \mathcal{O}(\bar{a}^2R^{-1})$ thanks to \eqref{eq:born}, and then choose $R = \mathcal{O}(\bar{a}(\rho \bar{a}^3)^{-\eta})$.

\item For the lower bound, see Theorem \ref{thm:lowerloc}, in order to derive the correct constant in front of the second order for $\eta > 0$, we do not need to impose any further condition on $v_A,v_B$ other than those in Assumption \ref{cond:potentials}, i.e., the potentials $v_A,v_B$ do not need to be soft.
On the other hand, we need the condition of softness for $v_{AB}$:
\begin{equation}\label{restrictive-condition}
\delta_{AB} \leq C \bar{a}(\rho \bar{a}^3)^{4\eta + \nu}.
\end{equation}

The potentials in the form \eqref{eq:brietzkesoft} for which $\delta_{AB} = \mathcal{O}(\bar{a}R_{AB}^{-1})$ cannot satisfy \eqref{restrictive-condition} and assumption \eqref{cond:R} at the same time.
An example of potentials satisfying both the conditions are given by $v_{AB}^{(\lambda)} = \lambda v_{R}$, i.e., a weak potential, with $v_{R}$ as in \eqref{eq:brietzkesoft} where the parameter $\lambda>0$ can be tuned in order to satisfy \eqref{eq:deltaerrorstatement}. In this way
\begin{equation}
\delta_{AB} = \mathcal{O}(\lambda \bar{a}^2R^{-1}) \leq  C \rho \bar{a}(\rho \bar{a}^3)^{2 \eta + \nu},
\end{equation}
for $R = \mathcal{O} ( \bar{a}(\rho \bar{a}^3)^{-\eta})$ and $\lambda = (\rho \bar{a}^3)^{\eta + \nu}$.

\end{itemize}

\end{remark}

\begin{remark}[Comparison with the 1 species case]
It is clear from the expressions \eqref{def:mu&xi} that in the case of one species if, w.l.o.g., $\rho_B =0 = a_B = a_{AB}$, we have $\mu_{\pm} = 1$ and we recover exactly the classical result conjectured by Lee-Huang-Yang and proven in \cite{BCS}, \cite{FS}, \cite{FS2}, \cite{YY}:
\begin{equation}
\lim_{a_B,a_{AB},\rho_B \rightarrow 0} \Big|e_{3D}(\rho_A,\rho_B) - 4\pi \rho_A^2a_A \Big(1+\frac{128}{15\sqrt{\pi}} (\rho_A a_A^3)^{1/2}\Big)\Big| \leq  C (\rho_A a_A)^{5/2} (\rho_A a_A^3)^{\eta}.
\end{equation}
\end{remark}

\begin{remark}[Comments on the assumptions]\label{rem:assmpts}
The non-increasing condition \eqref{cond:nonincreasing} on the potentials is a technical assumption needed only for the method of the proof we decided to use for the lower bound. In particular, this condition is key for the estimate of the error made when approximating the small-box Hamiltonian with its symmetric version in Lemma \ref{lem:boundGapTilde}. We are confident that the condition may be removed at the cost of a longer and more involved proof, for instance adapting the method used in \cite{2DLHY}, \cite{23DLHY}, \cite{FS}, \cite{FS2}.

The miscibility condition \eqref{cond:misc} has as immediate consequence the bound $\xi_{AB} \in [0,1]$ (recall its definition in \eqref{def:xi}) and it guarantees that the matrix of the scattering lengths 
\begin{equation}\label{scattlength:matrix}
\mathscr{A} = \begin{pmatrix}
a_A & a_{AB}\\
a_{AB} & a_B
\end{pmatrix},
\end{equation}
is semi-definite positive. This fact is used in several parts of the proof and, in particular, it is fundamental to guarantee the convexity of the form $v \cdot \mathscr{A}v$. In \cite{esry}, \cite[Chap. 21]{stringari} it is explained how this condition is key to avoid phase separation between the two components of the gas. In this way, the repulsion between the two species is not stronger than the intra-particles repulsions and the two components of the gas can spatially mix. In \cite{petrov} the author discusses the physical phenomena emerging when such condition is not satisfied. 
\end{remark}

From now on, for a box $\Lambda \subseteq \mathbb{R}^3$, we will denote by 
\begin{equation}
\mathfrak{h}_A := L^2(\Lambda;\mathrm{d}x),\qquad  \mathfrak{h}_B := L^2(\Lambda;\mathrm{d}y),
\end{equation}
the one-boson spaces of type $A$ and $B$, by $\Lambda^*$ the associated space of momenta and by $\{u_k\}_{k \in \Lambda^*}$ and $\{v_h\}_{h \in \Lambda^*}$ two sets of orthonormal bases for $\mathfrak{h}_A$ and $\mathfrak{h}_B$, respectively, and the particular choice will be made case by case to diagonalize the Laplacian with the relative boundary conditions. 
We will then consider the double-component Fock space on $\Lambda$
\begin{equation}\label{eq:Fockspacesmix}
\mathscr{F}(\Lambda) := \mathscr{F}_A(\Lambda)  \otimes \mathscr{F}_B(\Lambda), \qquad \mathscr{F}_j(\Lambda) = \bigoplus_{n =0}^{\infty} \mathfrak{h}_j^{\otimes_s n}, \quad j \in \{A,B\},
\end{equation}
and the creation and annihilation operators for the two types 
\begin{align}
a_k &:= a(u_k)\otimes \one_{\mathscr{F}_B}, \qquad a^*_k :=a^*(u_k) \otimes \one_{\mathscr{F}_B}, \label{eq:creationA}\\
b_h &:= \one_{\mathscr{F}_A} \otimes b(v_h), \qquad b_h^* := \one_{\mathscr{F}_A} \otimes b^*(v_h), \label{eq:creationB}
\end{align}
with the canonical commutation relations (CCR):
\begin{gather}\label{eq:CCR}
[a^*_k,b_h] =[a_k,a_h] = 0 = [a^*_k, a^*_h] = [a_k, b_h], \qquad  [a_k,a_h^*] =\delta_{k,h} \one_{\mathscr{F}_A}, \\
[b^*_k,a_h] =  [b_k,b_h] = 0 = [b^*_k, b^*_h] = [a^*_k, b^*_h],   \qquad [b_k,b_h^*] =\delta_{k,h} \one_{\mathscr{F}_B}.
\end{gather}
We can also introduce the number operators
\begin{align*}
\mathcal{N} &:= \mathcal{N}_A \otimes \one + \one\otimes \mathcal{N}_B, \\
 \mathcal{N}_{\#}\Psi^{(n_{A})} &\otimes \Psi^{(n_B)}  = n_{\#} \Psi^{(n_{A})} \otimes \Psi^{(n_B)}, \qquad \# \in \{A,B\},
 \end{align*}
 and the number of condensated ($n_0^A, n_0^B$) and excited ($n_+^A,n_+^B$) particles per species 
 \begin{align*}
\mathcal{N}_A &:= n_0^A + n_+^A,  \qquad \mathcal{N}_B := n_0^B  + n_+^B,\\
n_0^A &:= a^*_0  a_0 , \qquad n_0^B := b^*_0  b_0, \qquad n_0 = n_0^A + n_0^B,\\ n_+^A &:= \sum_{p \neq 0} a^*_p a_p, \qquad n_+^B := \sum_{p \neq 0} b^*_p b_p, \qquad n_+ = n_+^A + n_+^B.
\end{align*}
It is possible to adapt the quantization rules of 1 species of bosons (see, for example, \cite{Sol_07}) to the two species case. Let $h_A$ and $h_B$ be two symmetric operators on $\mathfrak{h}_A$ and $\mathfrak{h}_B$, respectively, then the following equivalence holds
\begin{equation}\label{secondquant:1body}
\bigoplus_{n=0}^{\infty} \sum_{ j=1 }^n h_A^{(j)} =  \frac{1}{2}\sum_{\alpha, \beta \in \Lambda^*} \langle u_{\alpha} \,, h_A u_{\beta} \rangle a^{*}_{\alpha}  a^{}_{\beta}, \qquad 
\bigoplus_{m=0}^{\infty} \sum_{ j=1 }^m h_B^{(j)} =  \frac{1}{2}\sum_{\alpha, \beta \in \Lambda^*} \langle v_{\alpha} \,, h_B v_{\beta} \rangle b^{*}_{\alpha}  b^{}_{\beta}
\end{equation}
where we used the notation $h_{\cdot}^{(j)} = \one^{\otimes_s (j-1)} \otimes_s h_{\cdot} \otimes_s \one^{\otimes_s (N-j)}$.

Let $W_A, W_B, W_{AB}$ be multiplication operators on $\mathfrak{h}_A \otimes \mathfrak{h}_A$, $\mathfrak{h}_B \otimes \mathfrak{h}_B$, $\mathfrak{h}_A \otimes \mathfrak{h}_B$, respectively, by $2$-variable functions invariant under the exchange of variables. Then 
\begin{align}\label{secondquant:2body}
\bigoplus_{n=0}^{\infty} \sum_{1\leq i<j\leq n } W_A(x_i,x_j) &=  \frac{1}{2}\sum_{\alpha, \beta, \gamma, \delta \in \Lambda^*} \langle u_{\alpha} \otimes u_{\beta}\,, W_A u_{\gamma} \otimes u_{\delta}\rangle a^{*}_{\alpha} a^{*}_{\beta} a^{}_{\gamma} a^{}_{\delta},\\
\bigoplus_{m=0}^{\infty} \sum_{1\leq i<j\leq m } W_B(y_i,y_j) &=  \frac{1}{2}\sum_{\alpha, \beta, \gamma, \delta \in \Lambda^*} \langle v_{\alpha} \otimes v_{\beta}\,, W_B v_{\gamma} \otimes v_{\delta}\rangle b^{*}_{\alpha} b^{*}_{\beta} b^{}_{\gamma} b^{}_{\delta},\\
\bigoplus_{\substack{N=0\\ n+m = N }}^{\infty} \sum_{i=1}^n \sum_{j=1}^m W_{AB}(x_i,y_j) &=  \sum_{\alpha, \beta, \gamma, \delta \in \Lambda^*} \langle u_{\alpha} \otimes v_{\beta}\,, W_{AB} u_{\gamma} \otimes v_{\delta}\rangle a^{*}_{\alpha}  a^{}_{\gamma} b^{*}_{\beta} b^{}_{\delta}.
 \end{align}

Using the previous quantization rules we can therefore extend the $N-$body Hamiltonian to a Hamiltonian $\mathcal{H}$ on $\mathscr{F}(\Lambda)$ which, restricted to $\mathfrak{h}_A^{\otimes_s N_A} \otimes \mathfrak{h}_B^{\otimes_s N_B}$, acts like $H_N$:
\begin{align}\label{eq:secondquantHam}
\mathcal{H} &:= \sum_{h,k \in \Lambda^*}\big( \langle u_k, -\Delta_x u_h\rangle a^*_k a_h + \langle v_k, -\Delta_y v_h\rangle b^*_k b_h \big) \nonumber\\
&\quad +\sum_{h,k,p,q \in \Lambda^*} \Big( \frac{1}{2}V_A^{(h,k,p,q)} a^*_h a^*_k a_p a_q + \frac{1}{2}V_B^{(h,k,p,q)} b^*_h b^*_k b_p b_q + V_{AB}^{(h,k,p,q)} a^*_h b^*_k a_p b_q  \Big),
\end{align}
where
\begin{align*}
V_A^{(h,k,p,q)} := \langle u_{h} \otimes u_{k}\,, &v_A u_{p} \otimes u_{q}\rangle, \qquad V_B^{(h,k,p,q)} := \langle v_{h} \otimes v_{k}\,, v_B v_{p} \otimes v_{q}\rangle,\\
&V_{AB}^{(h,k,p,q)} = \langle u_{h} \otimes v_{k}\,, v_{AB} u_{p} \otimes v_{q}\rangle.
\end{align*}

\section{Upper bound}\label{sec:uppbound}

The strategy to calculate the upper bound is inspired by the one for single species of bosons in \cite{ESY}, where the trial state chosen is a quasi-free state. The estimate of the energy is obtained by minimizing a Bogoliubov functional which represents the effective part of the quadratic form associated to $\mathcal{H}$ calculated in the quasi-free state.

Let us consider the thermodyanic box $\Lambda_L$ with periodic boundary conditions, which in this section we will denote by $\Lambda$ for simplicity, then
\begin{equation}
\Lambda^* := \frac{2\pi}{L} \mathbb{Z}^3.
\end{equation}
We recall that $N_A = \rho_A |\Lambda|, N_B = \rho_B |\Lambda|$, with $N = N_A + N_B$.
For a function $f$ on $\Lambda$ we define the Fourier transform, for $k \in \Lambda^*$,
\begin{equation}
\widehat{f} (k):= \int_{\Lambda} \mathrm{d}x \, e^{-ikx} f(x), \qquad f(x) =\frac{1}{|\Lambda|} \sum_{p \in \Lambda^*} e^{ipx} \widehat{f}(p).
\end{equation}

We state the main theorem of this section below, and dedicate the rest of the section to prove it.

\begin{theorem}\label{thm:mainupperbound}
Let $\eta \in [ 0,\frac{1}{10}]$ such that $v_A,v_B,v_{AB}$ are potentials satisfying assumptions \eqref{cond:integrablepotential}, \eqref{cond:misc}, \eqref{cond:aRa}, \eqref{cond:L1normv} and \eqref{cond:R}. Let $\mathcal{H}$ be the Hamiltonian defined in \eqref{eq:secondquantHam}. There exists a state $\Psi \in \mathscr{F}(\Lambda)$ such that 
\begin{equation}\label{eq:concNANBupper}
\langle \Psi, \mathcal{N}_A \Psi \rangle = N_A, \qquad \langle \Psi, \mathcal{N}_B \Psi\rangle = N_B,
\end{equation}
and a constant $C>0$ such that, for $\rho \bar{a}^{3}\leq C^{-1} $, 
\begin{align}
\langle \Psi, \mathcal{H} \Psi\rangle &\leq 4\pi |\Lambda| (\rho_A^2 a_A + 2 \rho_A \rho_B a_{AB} + \rho_B^2 a_B)\nonumber\\
&\quad +(\rho^2_A a^2_A + 2\rho_A \rho_B a_{AB}^2 +\rho_B^2 a_B^2)^{5/4} I_{AB} + C (\rho \bar{a} )^{5/2}|\Lambda|(\delta \bar{a}^{-1} +(\rho \bar{a}^3)^{\eta}), \label{eq:energyupperboundTheorem}
\end{align}
with $I_{AB}$ defined in \eqref{def:mu&xi} and $C_1$ in \eqref{cond:L1normv}.

\end{theorem}

We choose the set of bases for $\mathfrak{h}_A$ and $\mathfrak{h}_B$ to be plane waves that diagonalize the Laplacians in the relative position variables:
\begin{equation}
u_k(x) := L^{-3/2} e^{ikx}, \qquad v_h (y) := L^{-3/2} e^{ihy}, \qquad k,h \in \Lambda^*.
\end{equation}
Inserting these expressions in \eqref{eq:secondquantHam} and using that, since the potentials are compactly supported in a ball of radius $R \ll L$,
\begin{align*}
&\langle u_h, -\Delta_x u_k\rangle = \delta_{h,k} k^2 = \langle v_h, -\Delta_y v_k\rangle,\\
&V_A^{(h,k,p,q)} = \frac{1}{L^3}\delta_{p+q,h+k} \widehat{v}_A(q-k), \qquad V_B^{(h,k,p,q)} = \frac{1}{L^3}\delta_{p+q,h+k} \widehat{v}_B(q-k),\\
&V_{AB}^{(h,k,p,q)} = \frac{1}{L^3}\delta_{p+q,h+k} \widehat{v}_{AB}(q-k),
\end{align*}
we can write, in this case, the Hamiltonian $\mathcal{H}$ as
\begin{align}
\mathcal{H} &= \sum_{k \in \Lambda^*} k^2 (a^*_k a_k + b^*_k b_k) + \frac{1}{2|\Lambda|} \sum_{k,p,q \in \Lambda^*} \widehat{v}_A(k) a^*_{p+k} a^*_q a_{q+k} a_p \nonumber\\ 
&\quad+\frac{1}{2|\Lambda|} \sum_{k,p,q \in \Lambda^*} \widehat{v}_B(k) b^*_{p+k} b^*_q b_{q+k} b_p  +\frac{1}{|\Lambda|} \sum_{k,p,q \in \Lambda^*} \widehat{v}_{AB}(k) a^*_{p+k} b^*_q b_{q+k} a_p, \label{secondquantHam:upperbound}
\end{align}
acting on the Fock space $\mathscr{F}(\Lambda)$ as defined in \eqref{eq:Fockspacesmix}.
Let $N_0,N_0^A, N_0^B$ be three real, positive parameters such that $N_0 = N_0^A + N_0^B$, which we will use to indicate the number of bosons of type $A, B$ and in total in the condensate, respectively.
We introduce the two-components vector of annihilation and creation operators
\begin{equation}
c_p^{\#} = \begin{pmatrix}
a_p^{\#}\\b_p^{\#}
\end{pmatrix},
\end{equation}
where the notation $c^{\#}_p \in \{c_p,c_p^{*}\}$ denotes a choice between the relative creation and annihilation operators, and the same for $a^{\#}_p$ and $b^{\#}_p$.

We are ready to define the trial state
\begin{equation}\label{def:trialstateUpper}
\Psi := W_{N_0} T_{S} \Omega, 
\end{equation}
where 
\begin{itemize}
\item $\Omega = \Omega_A\otimes \Omega_B \in \mathscr{F}(\Lambda)$ is the vacuum state on the $2-$components Fock space, $\Omega_j = (1,0,0,\ldots) \in \mathscr{F}_j(\Lambda), j \in \{A,B\}$;

\item $W_{N_0} = W_A \otimes W_B$, with $W_A,W_B$ being the Weyl operators acting on the single tensor factors as
\begin{equation}
W_A= e^{\sqrt{N^A_0} (a_0^*-a_0)}, \qquad W_B= e^{\sqrt{N^B_0} (b_0^*-b_0)}, 
\end{equation}
which satisfy the following property
\begin{equation}\label{eq:weylrules}
W^*_{N_0} a^{\#}_0 W_{N_0} = a_0^{\#} + \sqrt{N_0^A},\qquad W_{N_0}^* b^{\#}_0 W_{N_0} = b_0^{\#} + \sqrt{N_0^B}; 
\end{equation}

\item We call $T_S$ the 2-species Bogoliubov transformation 
\begin{equation}\label{def:BogTransf}
T_S = e^{\frac{1}{2}\sum_{p \neq 0} \big(c^*_p \cdot S_p c^*_{-p} +  c_p \cdot S_p c_{-p} \big)},
\end{equation}
for a real, $2\times 2$ symmetric matrix $S_p $, whose expression is in formula \eqref{eq:choice_Psi}.
The following transformation rules apply (which can be obtained by a direct computation)
\begin{align}\label{eq:translationrulesp}
T^* c_p T &= \tau_p c_p + \sigma_p c^*_{-p},\quad &T^* c_{-p} T = \tau_{-p} c_{-p} - \sigma_{-p} c^*_{p},\nonumber\\  
T^* c_p^* T & = \tau_p c^*_p + \sigma_p c_{-p},\quad &T^* c_{-p}^* T  = \tau_{-p} c^*_{-p} - \sigma_{-p} c_{p},
\end{align}
where $\tau_p := \cosh\Big(\frac{1}{2}S_p\Big)$ and $\sigma_p:= \sinh\Big(\frac{1}{2}S_p\Big)$ are matrices defined by their series expansions (which are convergent thanks to the $S_p$ chosen). Since $[\sigma,\tau]= 0$ we have that, calling
\begin{equation}\label{eq:defgammaalpha}
\gamma_p = \begin{pmatrix}
\gamma^{AA} & \gamma^{AB}\\
\gamma^{BA} & \gamma^{BB}
\end{pmatrix} := 
\sigma^2_p, \qquad \alpha_p = \begin{pmatrix}
\alpha^{AA} & \alpha^{AB}\\
\alpha^{BA} & \alpha^{BB}
\end{pmatrix} :=\tau_p\sigma_p,
\end{equation}
these are real, symmetric $2 \times 2$ matrices as well (implying $\gamma^{AB} = \gamma^{BA}, \alpha^{AB} = \alpha^{BA}$). By plugging in formula \eqref{eq:choice_Psi}, we also obtain the following explicit expressions 
\begin{equation}\label{eq:alphagammaexplicit}
\alpha_p = -(1-\beta_p^2)^{-1}\beta_p, \qquad \gamma_p = (1-\beta_p^2)^{-1} \beta_p^2,
\end{equation}
where $\beta$ is a matrix defined in Lemma \ref{lem:diagBog}.
\end{itemize}

By using the previous properties \eqref{eq:weylrules}, \eqref{eq:translationrulesp}, \eqref{eq:defgammaalpha}, we have the following relations
\begin{align}
&\langle a^*_p a_q \rangle_{\Psi}  = \begin{cases}
\gamma^{AA}_p, &\text{ if } p = q \neq 0,\\
N_0^A, &\text{ if } p = q = 0,\\
0, &\text{ if } p \neq q,
\end{cases} \qquad \langle b^*_p b_q \rangle_{\Psi}  = \begin{cases}
\gamma^{BB}_p, &\text{ if } p  = q\neq 0,\\
N_0^B, &\text{ if } p = q =  0,\\
0, &\text{ if } p \neq q,
\end{cases} \nonumber \\
&\langle b^*_p a_q \rangle_{\Psi}^* = \langle a^*_p b_q \rangle_{\Psi}  = \begin{cases}
\gamma^{AB}_p, &\text{ if } p  = q \neq 0,\\
\sqrt{N_0^A N_0^B}, &\text{ if } p = q = 0,\\
0, &\text{ if } p \neq q,
\end{cases} \label{eq:transf_rules}
\end{align}
and 
\begin{align*}
&\langle a_p a_q\rangle_{\Psi} = \langle a^*_p a^*_q \rangle_{\Psi}  = \begin{cases}
\alpha^{AA}_p, &\text{ if } p = -q \neq 0,\\
N_0^A, &\text{ if } p = q = 0,\\
0, &\text{ if } p \neq 0,
\end{cases} \\
&\langle b_p b_q\rangle_{\Psi}  =\langle b^*_p b^*_q \rangle_{\Psi}  = \begin{cases}
\alpha^{BB}_p, &\text{ if } p  = -q\neq 0,\\
N_0^B, &\text{ if } p = q =  0,\\
0, &\text{ if } p \neq q,
\end{cases} \\
&\langle b_p^* a_q^* \rangle_{\Psi} = \langle b_p a_q \rangle_{\Psi}^* =  \langle a_p b_q \rangle_{\Psi} = \langle a^*_p b^*_q \rangle_{\Psi}^* = \begin{cases}
\alpha^{AB}_p, &\text{ if } p  =- q \neq 0,\\
\sqrt{N_0^A N_0^B}, &\text{ if } p = q = 0,\\
0, &\text{ if } p \neq q.
\end{cases}
\end{align*}
By an abuse of notation, we can express the previous relations by writing that
\begin{align}
&\gamma_p = \langle c_p^* \otimes c_p\rangle_{\Psi}, \qquad \alpha_p = \langle c_p \otimes c_{-p}\rangle_{\Psi},\\
&\langle c_0^* \otimes c_0\rangle_{\Psi} = 
\begin{pmatrix} \sqrt{N_0^A} \\ \sqrt{N_0^B}
\end{pmatrix} \otimes \begin{pmatrix}
\sqrt{N_0^A} \\ \sqrt{N_0^B}
\end{pmatrix} = \langle c_0 \otimes c_{0}\rangle_{\Psi},
\end{align}
where the action as quadratic form on $\Psi$ is intended on every element of the matrix.

We fix $N_0^A$ and $N_0^B$ so that 
\begin{equation}\label{choice:N0AN0B}
N_A = N_0^A + \sum_{p \neq 0} \gamma_p^{AA}, \qquad N_B = N_0^B + \sum_{p \neq 0} \gamma_p^{BB}.
\end{equation}

This choice together with \eqref{eq:weylrules}, \eqref{eq:transf_rules} gives us directly the following lemma.

\begin{lemma}\label{lem:concNumberNANBupper}
\begin{equation}
\langle \Psi, \mathcal{N}_A \Psi\rangle = N_A, \qquad \langle \Psi, \mathcal{N}_B \Psi\rangle = N_B.
\end{equation}
\end{lemma}

In the lemma below we show how the form of the Hamiltonian on $\Psi$ can be expressed in terms of $\alpha$ and $\gamma$.

\begin{lemma}
Let $\Psi$ be defined as \eqref{def:trialstateUpper}. Calling $\rho_{A,0}= N_0^A |\Lambda|^{-1}, \rho_{B,0} = N_0^B |\Lambda|^{-1}$, 
\begin{align*}
\langle \mathcal{H}\rangle_{\Psi} &= \mathcal{T} + \mathcal{L}_0^{N_0^A,N_0^B}(v_A,v_B,v_{AB})+ \mathcal{L}_2  + \mathcal{L}_4, \\
\mathcal{L}_2&:= \mathcal{L}_2^{(0)} + \mathcal{L}_2^{(1)},\\
\mathcal{L}_4&:= \mathcal{L}_4^{(0)} + \mathcal{L}_4^{(1)},
\end{align*}
where
\begin{gather}
\mathcal{T} = \sum_{p \neq 0} p^2 \big( \gamma^{AA}_p + \gamma^{BB}_p\big),\\
\mathcal{L}_0^{N_0^A,N_0^B}(v_A,v_B,v_{AB}):= \frac{1}{2|\Lambda|}\big((N_0^A)^2 \widehat{v}_A(0) +  2 N^A_0  N^B_0 \widehat{v}_{AB}(0)  + (N_0^B)^2\widehat{v}_B(0) \big),
\end{gather}
\begin{align*}
\mathcal{L}_2^{(0)} &:=  \sum_{p \neq 0} \big( \rho_{A,0}\widehat{v}_A(0) \gamma_p^{AA} +\rho_{B,0} \widehat{v}_B(0)) \gamma^{BB}_p+\widehat{v}_{AB}(0)  \big(\rho_{B,0} \gamma_p^{AA} +  \rho_{A,0} \gamma_p^{BB}\big)\big),\\
\mathcal{L}_2^{(1)}&:=\rho_{A,0}\sum_{p \neq 0} \widehat{v}_A(p)  \gamma_p^{AA} + \rho_{A,0} \sum_{p \neq 0} \widehat{v}_A(p) \alpha^{AA}_p +\rho_{B,0} \sum_{p \neq 0} \widehat{v}_B(p) \gamma_p^{BB}\\
&\quad + \rho_{B,0} \sum_{p \neq 0} \widehat{v}_B(p) \alpha^{BB}_p + \sqrt{\rho_{A,0}\rho_{B,0}} \sum_{p \neq 0} \widehat{v}_{AB}(p) \big(2\gamma^{AB}_p+ 2\alpha^{AB}_p\big),
\end{align*}
\begin{align*}
\mathcal{L}_4^{(0)} &:= \frac{\widehat{v}_A(0)}{2|\Lambda|} \Big(\sum_{p\neq 0 }\gamma^{AA}_p\Big)^2 +  \frac{\widehat{v}_B(0)}{2|\Lambda|} \Big(\sum_{p\neq 0 }\gamma^{BB}_p\Big)^2 + \frac{\widehat{v}_{AB}(0)}{|\Lambda|} \sum_{p,q \neq 0 } \gamma^{AA}_p \gamma^{BB}_q.\\
\mathcal{L}_4^{(1)} &:= \frac{1}{2|\Lambda|} \sum_{p, p+k \neq 0 } \Big(\widehat{v}_{A}(k) (\gamma_p^{AA} \gamma_{p+k}^{AA} + \alpha^{AA}_p \alpha^{AA}_{p+k}) + \widehat{v}_{B}(k)(\gamma_p^{BB} \gamma_{p+k}^{BB} + \alpha^{BB}_p \alpha^{BB}_{p+k})\\
&\quad +2 \widehat{v}_{AB}(k) (\gamma_p^{AB} \gamma_{p+k}^{AB} + \alpha^{AB}_p \alpha^{AB}_{p+k})\Big).
\end{align*}
\end{lemma}

\begin{proof}
By relations \eqref{eq:transf_rules} we get immediately the expression for $\mathcal{T}$ from the kinetic energy term. For the potential part 
\begin{equation}
\frac{1}{|\Lambda|} \sum_{k,p,q \in \Lambda^*} \Big( \frac{1}{2}\widehat{v}_A(k) a^*_{p+k} a^*_q a_{q+k} a_p  +\frac{1}{2} \widehat{v}_B(k) b^*_{p+k} b^*_q b_{q+k} b_p  +\widehat{v}_{AB}(k) a^*_{p+k} b^*_q b_{q+k} a_p\Big),
\end{equation}
we split it in four terms, according to how many momenta in the annihilation and creation operators are zero.
\begin{itemize}
\item All the momenta are zero: by \eqref{eq:transf_rules} we obtain immediately the expression of $\mathcal{L}_0^{N_0^A,N_0^B}$.

\item There are one or three momenta which are zero: 
we can use a modified version of Wick's Theorem for two species (see \cite[Theorem 10.2]{Sol_07} or by simple application of the calculation rules \eqref{eq:translationrulesp}) we get that these terms contain elements of the form
\begin{equation}
\langle a^{\#}_0 a^{\natural}_p\rangle_{\Psi} = \langle b^{\#}_0 b^{\natural}_p\rangle_{\Psi} = \langle a^{\#}_0 b^{\natural}_p\rangle_{\Psi} =0,
\end{equation} 
which are zero by \eqref{eq:transf_rules}. 

\item Two momenta are zero: for $\cdot, \star \in \{A,B\}$,
\begin{itemize}
\item[$(i)$] $p+k = 0, q=0$: we get terms of the form
\begin{equation}
\widehat{v}_{\cdot\star}(k) \sqrt{N_0^{\cdot}}\sqrt{N_0^{\star}} \alpha^{\cdot \star}_k,
\end{equation}
\item[$(ii)$] $p+k = 0 = q+k$: we get terms
\begin{equation}
\widehat{v}_{\cdot\star}(k) \sqrt{N_0^{\cdot}}\sqrt{N_0^{\star}} \gamma^{\cdot\star}_k,
\end{equation}
\item[$(iii)$] $p+k=0, p=0$: we get terms
\begin{equation}
\widehat{v}_{\cdot\star}(0) \sqrt{N_0^{\cdot}}\sqrt{N_0^{\star}} \gamma^{\cdot\star}_q,
\end{equation}
\item[$(iv)$] $q=0, q+k = 0$: we get terms
\begin{equation}
\widehat{v}_{\cdot\star}(0) \sqrt{N_0^{\cdot}}\sqrt{N_0^{\star}} \gamma^{\cdot\star}_p,
\end{equation}
\item[$(v)$] $q=0=p$: we get terms
\begin{equation}
\widehat{v}_{\cdot\star}(k) \sqrt{N_0^{\cdot}}\sqrt{N_0^{\star}} \gamma^{\cdot\star}_k,
\end{equation}
\item[$(vi)$] $q+k =0 =p$: we get terms
\begin{equation}
\widehat{v}_{\cdot\star}(k) \sqrt{N_0^{\cdot}}\sqrt{N_0^{\star}} \alpha^{\cdot\star}_k,
\end{equation}
\end{itemize}
which all together give the term $\mathcal{L}_2$ in the statement of the lemma.

\item No momentum is zero: we use again Wick's Theorem to get the term $\mathcal{L}_4$ in the statement of the lemma.
\end{itemize}
\end{proof}

This split of the quartic term is useful because we can state that, since 
\begin{equation}
N_{\#}N_{\natural} = \Big( N_0^{\#} + \sum_{p \neq 0} \gamma^{\#\#}_p\Big)\Big(N_0^{\natural} + \sum_{p \neq 0} \gamma^{\natural \natural}_p\Big), \qquad \#,\natural \in \{A,B\},
\end{equation}
then
\begin{equation}
\mathcal{L}_0^{N_0^A,N_0^B}(v_A,v_B,v_{AB}) +\mathcal{L}_2^{(0)} + \mathcal{L}_4^{(0)} = \mathcal{L}_0^{N_A,N_B}(v_A,v_B,v_{AB}).
\end{equation}
We are left with the expression
\begin{equation}
\langle \mathcal{H} \rangle_{\Psi} = \mathcal{T} +  \mathcal{L}_0^{N_A,N_B}(v_A,v_B,v_{AB}) + \mathcal{L}_2^{(1)} + \mathcal{L}_4^{(1)}.
\end{equation}
We use now some algebraic manipulation to rewrite the potential terms:

\begin{itemize}
\item We split $\mathcal{L}_2^{(1)}$ as
\begin{align*}
\mathcal{L}_2^{(1)} &= \mathcal{L}_2^{\alpha}(v_A,v_B,v_{AB}) + \mathcal{L}_2^{\gamma}(v_A,v_B,v_{AB}),\\
\mathcal{L}_2^{\alpha}(v_A,v_B,v_{AB}) &:= \sum_{p \neq 0}\rho_{A,0} \widehat{v}_A(p) \alpha_p^{AA} +\rho_{B,0}\widehat{v}_B(p) \alpha^{BB}_p + 2\sqrt{\rho_{A,0}\rho_{B,0}}\,\widehat{v}_{AB}(p)\alpha^{AB}_p,\\
\mathcal{L}_2^{\gamma}(v_A,v_B,v_{AB}) &=
\sum_{p \neq 0}\rho_{A,0} \widehat{v}_A(p) \gamma_p^{AA} +\rho_{B,0}\widehat{v}_B(p) \gamma^{BB}_p + 2\sqrt{\rho_{A,0}\rho_{B,0}}\,\widehat{v}_{AB}(p)\gamma^{AB}_p.
\end{align*}

\item We rewrite $\mathcal{L}_4^{(1)}$ as
\begin{align*}
\mathcal{L}_4^{(1)} &= \mathcal{D}_{v_A}(\gamma^{AA},\gamma^{AA}) + \mathcal{D}_{v_A}(\alpha^{AA},\alpha^{AA}) +\mathcal{D}_{v_B}(\gamma^{BB},\gamma^{BB})\\
&\quad + \mathcal{D}_{v_B}(\alpha^{BB},\alpha^{BB})+2\mathcal{D}_{v_{AB}}(\gamma^{AB},\gamma^{AB}) + 2\mathcal{D}_{v_{AB}}(\alpha^{AB},\alpha^{AB}),
\end{align*}
where we used the form
\begin{equation}
\mathcal{D}_v(f,g) := \frac{1}{2|\Lambda|}\sum_{p,p+k \neq 0} \widehat{v}(k) f_p g_{p+k}.
\end{equation}

We observe that the following relation holds
\begin{multline}\label{eq:Drenorm}
\mathcal{D}_{v_{AB}} (\alpha^{AB},\alpha^{AB})= \mathcal{D}_{v_{AB}} (\alpha^{AB}+\sqrt{\rho_{A,0}\rho_{B,0}}\,\widehat{\omega}_{AB},\alpha^{AB}+\sqrt{\rho_{A,0}\rho_{B,0}}\, \widehat{\omega}_{AB})\\
  +\frac{\rho_{A,0}\rho_{B,0}}{|\Lambda|}\widehat{(v\omega^2_{AB})}(0) -2 \sqrt{\rho_{A,0}\rho_{B,0}} \mathcal{D}_{v_{AB}} (\alpha^{AB},\widehat{\omega}_{AB}),
\end{multline}
and similar ones for $\alpha^{AA}, \alpha^{BB}$. 
\end{itemize}

We now want to renormalize the main order and the quadratic terms by substituting the $v$'s with the $g$'s:
\begin{itemize}
\item For $\mathcal{L}_2^{\alpha}$ we get
\begin{multline}\label{eq:changealphavg}
\mathcal{L}_2^{\alpha}(v_A,v_B,v_{AB}) -2 \rho_{A,0} \mathcal{D}_{v_A}(\alpha^A,\widehat{\omega}_A) - 2 \rho_{B,0} \mathcal{D}_{v_B}(\alpha^B,\widehat{\omega}_B)\\ -4 \sqrt{\rho_{A,0}\rho_{B,0}} \mathcal{D}_{v_{AB}} (\alpha^{AB},\widehat{\omega}_{AB})
= \mathcal{L}_2^{\alpha}(g_A,g_B,g_{AB}).
\end{multline}
\item For $\mathcal{L}_2^{\gamma}$ we get 
\begin{equation}\label{eq:changegammavg}
\mathcal{L}_2^{\gamma}(v_A,v_B,v_{AB}) = \mathcal{L}_2^{\gamma}(g_A,g_B,g_{AB}) + \mathcal{L}_2^{\gamma}(v_A\omega_A,v_B \omega_B,v_{AB}\omega_{AB}).
\end{equation}
\item For the main order in $\mathcal{L}_0$ we get
\begin{align}
&\mathcal{L}_0^{N_A,N_B}(v_A,v_B,v_{AB})\nonumber\\ 
&= \mathcal{L}_0^{N_A,N_B}(g_A,g_B,g_{AB})  +\mathcal{L}_0^{N_0^A,N_0^B}(g\omega_A,g\omega_B,g\omega_{AB}) 
+ \mathcal{L}_0^{N^A_0,N^B_0}(v\omega^2_A,v\omega^2_B,v\omega^2_{AB})\nonumber \\ 
&+ \frac{|\Lambda|}{2} (\rho_A^2-\rho_{A,0}^2) \widehat{v\omega}_A(0) + \frac{|\Lambda|}{2} (\rho_B^2-\rho_{B,0}^2) \widehat{v\omega}_B(0) +|\Lambda| (\rho_A\rho_B-\rho_{A,0}\rho_{B,0}) \widehat{v\omega}_{AB}(0).\label{eq:change0gv}
\end{align}
and furthermore observe that, by definition, 
\begin{align}\label{eq:change0gv2}
&\mathcal{L}_0^{N^A_0,N^B_0}(v\omega^2_A,v\omega^2_B,v\omega^2_{AB}) \nonumber\\
&=\frac{|\Lambda|}{2}\Big(\rho_{A,0}^2\widehat{(v\omega^2_{A})}(0) +\rho_{B,0}^2\widehat{(v\omega^2_{B})}(0) +2\rho_{A,0}\rho_{B,0}\widehat{(v\omega^2_{AB})}(0)\Big).
\end{align}
\end{itemize}
By using \eqref{eq:Drenorm}, \eqref{eq:changealphavg}, \eqref{eq:changegammavg}, \eqref{eq:change0gv} and \eqref{eq:change0gv2} we proved the following lemma.
\begin{lemma}\label{lem:upperboundKbogE}
Let $\Psi$ be the state defined in \eqref{def:trialstateUpper}, then the following expression holds
\begin{align}
\langle \mathcal{H}\rangle_{\Psi} &=  \mathcal{L}_0^{N_A,N_B}(g_A,g_B,g_{AB}) +\mathcal{L}_2^{\gamma}(v_A\omega_A,v_B\omega_B,v_{AB}\omega_{AB})\nonumber \\
&\quad +\mathcal{L}_0^{N_0^A,N_0^B}(g\omega_A,g\omega_B,g\omega_{AB})+ \mathcal{K}^{\text{Bog}} + \widetilde{\mathcal{E}},
\end{align}
where
\begin{equation}
\mathcal{K}^{\text{Bog}} = \mathcal{T} +\mathcal{L}_2^{\gamma}(g_A,g_B,g_{AB}) + \mathcal{L}_2^{\alpha}(g_A,g_B,g_{AB}),
\end{equation}
and
\begin{multline}
\widetilde{\mathcal{E}} = \mathcal{D}_{v_{A}} (\alpha^{A}+\rho_{A,0}\widehat{\omega}_{A},\alpha^{A}+\rho_{A,0}\widehat{\omega}_{A}) +\mathcal{D}_{v_{B}} (\alpha^{B}+\rho_{B,0}\widehat{\omega}_{B},\alpha^{B}+\rho_{B,0}\widehat{\omega}_{B}) \\+ \mathcal{D}_{v_{AB}} (\alpha^{AB}+\sqrt{\rho_{A,0}\rho_{B,0}}\widehat{\omega}_{AB},\alpha^{AB}+\sqrt{\rho_{A,0}\rho_{B,0}}\widehat{\omega}_{AB})\\
+\mathcal{D}_{v_A}(\gamma^{AA},\gamma^{AA})  +\mathcal{D}_{v_B}(\gamma^{BB},\gamma^{BB}) +2\mathcal{D}_{v_{AB}}(\gamma^{AB},\gamma^{AB})\\
+\frac{|\Lambda|}{2} (\rho_A^2-\rho_{A,0}^2) \widehat{v\omega}_A(0) + \frac{|\Lambda|}{2} (\rho_B^2-\rho_{B,0}^2) \widehat{v\omega}_B(0) + |\Lambda|(\rho_A\rho_B-\rho_{A,0}\rho_{B,0}) \widehat{v\omega}_{AB}(0).
\end{multline}
\end{lemma}

In the next lemma we rewrite the form associated to the Bogoliubov Hamiltonian by diagonalizing it.

\begin{lemma}\label{lem:upperboundKbogDiag}
The following equivalence holds
\begin{equation}\label{eq:ubKbogdiag}
\mathcal{K}^{\text{Bog}}= \sum_{k \in \Lambda^*}\Big(\frac{1}{2} \Big(\sqrt{k^4  +2\lambda_+(k) k^2} +\sqrt{k^4  +2\lambda_-(k) k^2}  \Big)- k^2 - \frac{1}{2}(\lambda_+(k) + \lambda_-(k))  \Big),
\end{equation}
where $\lambda_{\pm}(k) = [\lambda_{\pm}(\rho_{A,0},\rho_{B,0})](k)$ are defined in \eqref{def:lambdapm}.
\end{lemma}

\begin{proof}

We observe that we can write
\begin{equation}
\mathcal{K}^{\text{Bog}} = \Big\langle\Psi, \sum_{k \neq 0} \Big(c^*_k \cdot \mathcal{A}_k c_k + \frac{1}{2} (c_k\cdot \mathcal{B}_k c_{-k} + c^*_k \cdot\mathcal{B}_k c^*_{-k})\Big)\Psi\Big\rangle,
\end{equation}
with
\begin{equation}
\mathcal{A}_k = k^2 \one_2 + \mathcal{B}_k, \qquad \mathcal{B}_k = \begin{pmatrix}
\rho_{A,0} \widehat{g}_A(k) & \sqrt{\rho_{A,0}\rho_{B,0}} \, \widehat{g}_{AB}(k)\\
\sqrt{\rho_{A,0}\rho_{B,0}}\,\widehat{g}_{AB}(k) & \rho_{B,0}  \widehat{g}_{B}(k)
\end{pmatrix}.
\end{equation}
Using Lemma \ref{lem:diagBog} from Appendix \ref{app:bogint} we can diagonalize the previous form introducing the operators
\begin{equation}
d_k := c_k + \beta_k c_{-k}^*,
\end{equation}
so that we can write
\begin{equation}\label{eq:Ddiag+trace}
\mathcal{K}^{\text{Bog}} = \Big\langle \Psi, \sum_{k \neq 0} d^*_k \cdot \mathcal{D}_k d_k \Psi \Big\rangle  - \mathrm{Tr}(\beta^2_{\text{diag}} \mathcal{D}_{\text{diag}}),
\end{equation}
where $\mathcal{D}_k, \mathcal{D}^{\text{diag}}_k$ and $\beta_k, \beta^{\text{diag}}_k$ are defined in \eqref{eq:Ddiagexpressions}, \eqref{eq:Bdiagexpressions}, and
\begin{equation}
\mathrm{Tr}( \mathcal{D}_{\mathrm{diag}}\beta^2_{\mathrm{diag}}) = \sum_{k \in \Lambda^*}\Big( k^2 + \frac{1}{2}(\lambda_+ + \lambda_-) - \frac{1}{2} \Big(\sqrt{k^4  +2\lambda_+ k^2} +\sqrt{k^4  +2\lambda_- k^2}  \Big)\Big).
\end{equation}
We choose $\Psi$ such that $d_k \Psi = 0$, which, by using the transformation rules \eqref{eq:transf_rules}, corresponds to ask 
\begin{equation}\label{eq:choice_Psi}
T (\sigma_k + \beta_k \tau_k) c^*_{-k}\Omega_A \otimes \Omega_B = 0 \iff \sigma_k + \beta_k \tau_k = 0 \iff S_k = \log((1+\beta_k)^{-1}(1-\beta_k)),
\end{equation}
where the $\log$ is well defined because $\beta \leq \one_2$.
Therefore, with this choice, the first term on the r.h.s. of \eqref{eq:Ddiag+trace} vanishes, and \eqref{eq:ubKbogdiag} is proven.
\end{proof}

The choice made in \eqref{eq:choice_Psi} for $\Psi$ lets us obtain an explicit behavior of the $\alpha,\gamma$ and therefore we can study their asymptotic behavior.

\begin{lemma}\label{lem:estimates_alphagamma}
The elements of the matrices $\alpha$ and $\gamma$ satisfy the following bounds:
\begin{itemize}
\item For $|k| \leq 2\sqrt{\rho \bar{a}}$,
\begin{equation}
|\gamma_k^{(\#)}|, |\alpha_k^{(\natural)}| \leq C \frac{\sqrt{\rho \bar{a}}}{|k|}, \qquad \#,\natural \in \{AA,BB,AB\};
\end{equation}
\item For $|k|> 2 \sqrt{\rho \bar{a}}$
\begin{align*}
&\gamma^{AA}_k \leq \frac{\rho_{A,0}^2 \widehat{g}_A^2(k) + \rho_{A,0}\rho_{B,0} \widehat{g}_{AB}(k)}{4 k^4}, \quad \gamma^{BB}_k \leq \frac{\rho_{B,0}^2 \widehat{g}_B^2(k) + \rho_{A,0}\rho_{B,0} \widehat{g}_{AB}(k)}{4 k^4}, \\
&\gamma^{AB}_k \leq \frac{(\rho_{A,0} \widehat{g}_A + \rho_{B,0} \widehat{g}_B)  \sqrt{\rho_{A,0}\rho_{B,0}}\, \widehat{g}_{AB}(k)}{4 k^4}, \\ 
&\alpha^{AA}_k \leq -\frac{\rho_{A,0}\widehat{g}_A(k)}{2 k^2} - \frac{\rho_{A,0}^2 \widehat{g}_{A}^2(k) + \rho_{A,0}\rho_{B,0}\widehat{g}_{AB}^2(k)}{2k^4}, \\
&\alpha^{BB}_k \leq -\frac{\rho_{B,0}\widehat{g}_B(k)}{2 k^2}- \frac{\rho_{B,0}^2 \widehat{g}_{B}^2(k) + \rho_{A,0}\rho_{B,0}\widehat{g}_{AB}^2(k)}{2k^4}, \\
&\alpha^{AB}_k \leq -\frac{\sqrt{\rho_{A,0}\rho_{B,0}}}{2k^2} - \frac{\sqrt{\rho_{A,0}\rho_{B,0}}\,\widehat{g}_{AB}(k)(\rho_{A,0}\widehat{g}_A (k)+\rho_{B,0}\widehat{g}_{B}(k))}{2k^4}.
\end{align*}
\end{itemize}
This implies the following estimates
\begin{equation}\label{eq:estimateexcitationsupper}
\sum_{k \neq 0 } \gamma_k^{\#} \leq  C |\Lambda| (\rho \bar{a})^{3/2}, \qquad \# \in \{AA,BB,AB\}.
\end{equation}
\end{lemma}

\begin{proof}
By \eqref{eq:alphagammaexplicit} and by using the unitary transformation $U$ from Appendix \ref{app:bogint}, we get
\begin{align*}
\gamma_k = U  \frac{\beta_{\text{diag}}^2}{1-\beta^2_{\text{diag}}} U^*, \qquad \alpha_k = -U \frac{\beta_{\text{diag}}}{1-\beta_{\text{diag}}^2}U^*.
\end{align*}
For $|k| > 2\sqrt{\rho \bar{a}}$,
\begin{equation}
\gamma_k \simeq U \,\mathrm{diag}\Big(\frac{\lambda_+^2}{4k^4}, \frac{\lambda_-^2}{4k^4}\Big) U^* = \frac{1}{4k^4} \mathcal{B}^2_k, \qquad \alpha_k \simeq U\, \mathrm{diag}\Big(\frac{\lambda_+}{2k^2},\frac{\lambda_-}{2k^2} \Big) U^* = \frac{1}{2k^2} \mathcal{B}_k,
\end{equation}
while for $|k| \leq 2\sqrt{\rho \bar{a}}$,
\begin{equation}
\gamma_k \simeq C U \,\mathrm{diag}\Big(\frac{\sqrt{\lambda_+}}{|k|}, \frac{\sqrt{\lambda_-}}{|k|}\Big)U^* = \frac{C}{|k|} \mathcal{B}^{1/2}_k, \quad \alpha_k \simeq C U\, \mathrm{diag}\Big(\frac{\sqrt{\lambda_+}}{|k|}, \frac{\sqrt{\lambda_-}}{|k|}\Big)U^* = \frac{C}{|k|} \mathcal{B}^{1/2}_k,
\end{equation}
and they have the desired asymptotic behavior in the statement of the lemma.
We use the previous bounds to prove the last inequality
\begin{equation}
\sum_{k \neq 0 } \gamma^{\#}_k \leq  C\sum_{|k| \leq 2\sqrt{\rho \bar{a}}} \frac{\sqrt{\rho \bar{a}}}{|k|} + C\sum_{|k| > 2\sqrt{\rho \bar{a}}} \frac{\rho^2 \bar{a}^2}{k^4} \leq  C |\Lambda| (\rho \bar{a})^{3/2}.
\end{equation}
\end{proof}

Thanks to the previous lemma we can estimate the error term $\widetilde{\mathcal{E}}$. We recall the definition of $\delta$ in \eqref{eq:defdelta}.

\begin{lemma}\label{lem:upperboundErrestimate}
There exists a $C>0$ such that
\begin{equation}
|\widetilde{\mathcal{E}}+ \mathcal{L}_2^{\gamma}(v_A\omega_A,v_B\omega_B,v_{AB}\omega_{AB})| \leq C\rho^2|\Lambda| (\rho \bar{a}^3)^{1/2}  \delta + C _1 |\Lambda| (\rho \bar{a})^3.
\end{equation}
\end{lemma}

\begin{proof}
Let us estimate term by term. Using that $N_{\#} = N_{0,\#} + \sum_{k \neq 0} \gamma^{\#}_k, \#\in \{A,B\}$, and Lemma \ref{lem:estimates_alphagamma} we get
\begin{equation}
|\Lambda|(\rho_{\#}\rho_{\natural}-\rho_{0,\#}\rho_{0,\natural}) \widehat{v\omega}_{\#\natural}(0) \leq C \rho |\Lambda| \sum_{k \neq 0} \gamma_k^{\#} (\widehat{v}_{\#}(0) - \widehat{g}_{\#}(0))\leq  C \rho^2 |\Lambda|(\rho \bar{a}^3)^{1/2}  \delta,
\end{equation}
for $\#, \natural \in \{A,B\}$. Using again Lemma \ref{lem:estimates_alphagamma} and assumption \eqref{cond:aRa},
\begin{equation}
\mathcal{D}_{v_{\#}}(\gamma^{\#},\gamma^{\#}) \leq \frac{C}{|\Lambda|} \widehat{v}_{\#}(0) \Big(\sum_{k \neq 0} \gamma^{\#}_k\Big)^2 \leq C  \widehat{v}_{\#}(0) |\Lambda| (\rho \bar{a})^3 \leq C_1 \rho^2\bar{a} |\Lambda| \rho \bar{a}^3.
\end{equation}
Now we turn the attention to the $\mathcal{D}$ form in $\alpha$.
Let us write the calculation for the $AB$ component, the others being totally analogous. We denote $\varphi_{AB}(p) := (\alpha^{AB}(p)-\sqrt{\rho_{A,0}\rho_{B,0}}\widehat{\omega}_{AB}(p))$ and, by Lemma \ref{lem:estimates_alphagamma} and using that 
$2 p^2\widehat{\omega}_{AB}(p) = \widehat{g}_{AB}(p)$ by \eqref{eq:ScatOmega}, we have
\begin{equation}
\Big|\sum_{p \in \Lambda^*}\varphi_{AB}(p) \Big|\leq C\rho \bar{a}\sum_{|p| \leq 2 \sqrt{\rho \bar{a}}} \frac{1}{p^2} + C(\rho \bar{a})^2 \sum_{|p| > 2 \sqrt{\rho\bar{a}}} \frac{1}{p^4} \leq C(\rho \bar{a})^{3/2}|\Lambda|,
\end{equation}
which therefore gives
\begin{equation}
|\mathscr{D}(\varphi_{AB},\varphi_{AB})| \leq  \frac{1}{|\Lambda|}\Big|\sum_{p \in \Lambda^*}\varphi_{AB}(p) \Big|^2\widehat{v}_{AB}(0) = C_1 (\rho \bar{a})^3|\Lambda|.
\end{equation}
For the last term, using that $\widehat{v\omega} = \widehat{v} - \widehat{g}$,
\begin{align*}
&|\mathcal{L}_2^{\gamma}(v_A\omega_A,v_B\omega_B,v_{AB}\omega_{AB})| \\
&=
\Big|\sum_{p \neq 0}\rho_{A,0} \widehat{v\omega}_A(p) \gamma_p^{AA} +\rho_{B,0}\widehat{v\omega}_B(p) \gamma^{BB}_p + 2\sqrt{\rho_{A,0}\rho_{B,0}}\widehat{v\omega}_{AB}(p)\gamma^{AB}_p\Big|\\
&\leq \rho \delta \max_{j = AA, BB, AB}\sum_{k \neq 0} |\gamma^{j}_k | \leq C \rho^2|\Lambda| (\rho \bar{a}^3)^{1/2}  \delta.
\end{align*}
\end{proof}

We are ready to prove the main theorem of this section.
\begin{proof}[Proof of Theorem \ref{thm:mainupperbound}]
We first observe that choosing $\Psi$ as in \eqref{def:trialstateUpper}, condition \eqref{eq:concNANBupper} is satisfied thanks to Lemma \ref{lem:concNumberNANBupper}. 
By Lemma \ref{lem:upperboundKbogE}, Lemma \ref{lem:upperboundKbogDiag} and Lemma \ref{lem:upperboundErrestimate} we have 
\begin{align*}
\langle \mathcal{H}\rangle_{\Psi} &\leq \mathcal{L}_0^{N_A,N_B}(g_A,g_B,g_{AB}) +  \mathcal{L}_0^{N_0^A,N_0^B}(g\omega_A,g\omega_B,g\omega_{AB})  +\mathcal{S}\\
&\quad +C\rho^2|\Lambda| (\rho \bar{a}^3)^{1/2}  \delta + C_1 |\Lambda| (\rho \bar{a})^3.
\end{align*}
where 
\begin{equation}
\mathcal{S} := \sum_{k \in \Lambda^*}\Big( \frac{1}{2} \Big(\sqrt{k^4  +2\lambda_+(k) k^2} +\sqrt{k^4  +2\lambda_-(k) k^2}  \Big)-k^2 - \frac{1}{2}(\lambda_+(k) + \lambda_-(k))\Big).
\end{equation}
By introducing 
\begin{equation}
G_{AB} := \rho_{A,0}^2 \widehat{g}_A^2 +2 \rho_{A,0} \rho_{B,0} \widehat{g}_{AB}^2 + \rho_{B,0}^2 \widehat{g}_B^2,
\end{equation}
we have by Lemma \ref{lem:reconstructiongomega} that 
\begin{equation}
\mathcal{L}_0^{N_0^A,N_0^B}(g\omega_A,g\omega_B,g\omega_{AB}) \leq Z_0 + C \rho N \frac{\bar{a}^2}{L}, \qquad Z_0 := \sum_{k \in \Lambda^*}  \frac{G_{AB}(k)}{4k^4}. 
\end{equation}
We see that, since we are in the thermodynamic box, we can approximate the sum $\mathcal{S}+Z_0$ by the associated integral paying a negligible error
\begin{align*}
\mathcal{S}+ Z_0 &\simeq S_{AB},\\
S_{AB} &:=\frac{|\Lambda|}{(2\pi)^3}\int_{\mathbb{R}^3}\mathrm{d}k\, \Big( -k^2 - \frac{1}{2}(\lambda_+(k) + \lambda_-(k)) \\
&\qquad + \frac{1}{2} \Big(\sqrt{k^4  +2\lambda_+(k) k^2} +\sqrt{k^4  +2\lambda_-(k) k^2} \Big) +\frac{G_{AB}(k)}{4k^2} \Big). 
\end{align*}

By Lemma \ref{lem:calcBogInt} we finally have, choosing in the lemma $\rho_{z_A} = \rho_{A,0}, \rho_{z_B} = \rho_{B,0}, \ell = L$,
\begin{equation*}
\mathcal{S} + Z_0 \leq  |\Lambda|(\rho_{A,0}^2 a_A^2 + 2 \rho_{A,0} \rho_{B,0} a_{AB}^2 + \rho_{B,0}^2 a_B^2)^{5/4} I_{AB} (\rho_{A,0}, \rho_{B,0}) + \mathcal{E},
\end{equation*}
where
\begin{align*}
I_{AB} (\rho_{A,0}, \rho_{B,0})= (8\pi)^{5/2}\frac{2\sqrt{2}}{15\pi^2} (\mu_+^{5/2}(\rho_{A,0},\rho_{B,0}) + \mu_-^{5/2}(\rho_{A,0},\rho_{B,0})) = \mathcal{O}(1),
\end{align*}
with $\mu_{\pm}$ as defined in \eqref{def:mu&xi}, and, since $R \leq C_R (\rho \bar{a}^3)^{-\eta}\bar{a}$ and $\eta \leq \frac{1}{10}$ we can estimate
\begin{align*}
|\mathcal{E}| &\leq C(\rho \bar{a})^{5/2} L^3 \Big( \big(\rho \bar{a})^{-1/2}L^{-1}\log(L \bar{a}^{-1}) +  (\rho \bar{a}^3)^{2\eta} \big) \one_{\eta\neq 0} +  \one_{\eta = 0}\Big) \leq C(\rho\bar{a})^{5/2} L^3 (\rho \bar{a}^3)^{\eta},
\end{align*}
where we used that $L^{-1}\log(L\bar{a}^{-1}) \ll (\rho \bar{a}^3)^{1/2 +\eta}$.
Now we observe that, thanks to \eqref{eq:estimateexcitationsupper}, we have $\rho_{\#} - \rho_{\#,0} \leq C (\rho \bar{a})^{3/2}$ and therefore we can substitute the $\rho_{A,0},\rho_{B,0}$ with the $\rho_A,\rho_B$ by a small negligible error, obtaining
\begin{equation*}
\mathcal{S} + Z_0 \leq  |\Lambda|(\rho_A^2 a_A^2 + 2 \rho_A \rho_B a_{AB}^2 + \rho_B^2 a_B^2)^{5/4} I_{AB} + C(\rho\bar{a})^{5/2} L^3 (\rho \bar{a}^3)^{\eta}.
\end{equation*}
This proves the result.

\end{proof}


\section{Lower bound}\label{sec:lowerbound}

An important step in the proof of Theorem \ref{thm:main} is the localization of the problem in small boxes $\Lambda_{\ell} = [-\ell/2, \ell/2]^3$ of size
\begin{equation}
\ell = K_{\ell} (\rho \bar{a})^{-1/2},
\end{equation}
where $K_{\ell} = (1000C)^{-1}(\rho \bar{a}^3)^{-2\eta}$, $\eta \geq 0$.
 We also introduce the parameter
\begin{equation}
K_z = (\rho \bar{a}^3)^{-\nu} \gg 1,
\end{equation}
for $\nu > 0$, which offers a threshold for the control of the number of particles in the condensate.
For any $n,m \in \mathbb{N}$, let $H_{n,m}$ be defined analogously to \eqref{def:HN} and acting on $L^2_s(\Lambda_{\ell}^n;\mathrm{d}x)\otimes L^2_s(\Lambda_{\ell}^m;\mathrm{d}y)$ with Neumann boundary conditions on $\Lambda_{\ell}$. We introduce 
\begin{equation}
E_{n,m} (\ell) := \inf \mathrm{Spec}(H_{n,m}).
\end{equation}
We have the following lower bound for the energy on the small boxes.

\begin{theorem}\label{thm:lowerloc}
Let $\eta \in [0,\frac{1}{2000}]$ and $\nu \in (0,\frac{1}{10000} ]$ such that $v_A,v_B,v_{AB}$ are potentials satisfying Assumption \ref{cond:potentials}. There exists a constant $C>0$ such that, for $\rho \bar{a}^3 \leq C^{-1}$ and for $n, m \leq 100 C_a^2 \rho \ell^3$, and assuming that 
\begin{equation}\label{assmpt:reabsorbSpGap}
K_{\ell}^{2}K_z \delta_{AB}\bar{a}^{-1}\leq (1000 C)^{-1}, \qquad \text{ for } \eta \neq 0,
\end{equation}
the following lower bound holds
\begin{align*}
E_{n,m}(\ell) &\geq \frac{4\pi}{\ell^3} (n^2 a_A + 2  n m a_{AB} + m^2 a_B )\\
&\quad + \ell^3 \Big( \Big(\frac{n}{\ell^3} \Big)^2 a^2_A + 2 \frac{nm}{\ell^6}a_{AB}^2+ \Big(\frac{m}{\ell^3} \Big)^2 a^2_B \Big)^{5/4} I_{AB} - C\ell^3 \big( \rho\bar{a}\big)^{5/2}(\rho \bar{a}^3)^{\eta},
\end{align*}
where $I_{AB} = I_{AB}(n\ell^{-3},m \ell^{-3})$ is defined in \eqref{def:mu&xi}.
\end{theorem}

\begin{proof}
We choose here the parameters which are going to satisfy all the conditions requested in the propositions and lemmas used in the following. 
Recall the condition \eqref{cond:R}: $R \leq C_R (\rho \bar{a}^3)^{-\eta}\bar{a}.$
We pick $0 \leq \eta \leq \frac{1}{2000}$, $\nu = \frac{1}{10000}, \gamma = \frac{1}{50}, M = 15000$ and choose 
\begin{equation}
K_{\ell} = \frac{(\rho \bar{a}^3)^{-2\eta}}{1000C},\quad  K_H= (\rho \bar{a}^3)^{-\frac{1}{250}-\frac{3}{10000}}, \quad K_z = (\rho \bar{a}^3)^{-\frac{1}{10000}}, \quad \mathcal{M} = \rho \ell^3 (\rho\bar{a}^3)^{\frac{1}{50}}.
\end{equation}
By the assumptions we also have
\begin{equation}
\delta_{AB} \leq (1000C)^{-1} (\rho \bar{a}^3)^{4\eta+\nu} \bar{a}.
\end{equation}

 We split the analysis in two cases. First let us assume that $n+m\leq (\rho \bar{a}^3)^{-\frac{1}{17}}$. In this case we observe that the terms of the expansion we want to prove satisfy
\begin{align*}
 &\frac{4\pi}{\ell^3} (n^2 a_A + 2  n m a_{AB} + m^2 a_B ) \leq C K_{\ell}^{-3} (\rho \bar{a}^3)^{-\frac{2}{17}} \bar{a} (\rho \bar{a})^{\frac{3}{2}} = C\bar{a}^{-2} (\rho \bar{a}^3)^{\frac{3}{2}-\frac{2}{17}+6\eta},\\
 &\ell^3 \Big( \Big(\frac{n}{\ell^3} \Big)^2 a^2_A + 2 \frac{nm}{\ell^6}a_{AB}^2+ \Big(\frac{m}{\ell^3} \Big)^2 a^2_B \Big)^{5/4} I_{AB} \leq  C\bar{a}^{-2} (\rho \bar{a}^3)^{\frac{245}{68}-6\eta},\\
 &C \ell^3(\rho \bar{a})^{5/2}(\rho \bar{a}^3)^{\eta} = CK_{\ell}^3 (\rho \bar{a})^{-3/2} (\rho \bar{a})^{5/2} (\rho \bar{a}^3)^{\eta}= C \bar{a}^{-2}(\rho \bar{a}^3)^{1 - 6 \eta}.
\end{align*}
It is clear that the first two terms are of a smaller order than the negative error term. Therefore the desired bound is
\begin{equation}
E_{n,m}(\ell) \geq -C \ell^3(\rho \bar{a})^{5/2}(\rho \bar{a}^3)^{{\eta}},
\end{equation}
which is automatically satisfied because $H_{n,m} \geq 0$.

Then from now on we assume that $(\rho \bar{a}^3)^{-\frac{1}{17}}\leq  n+m \leq 100 C_a^2\rho \ell^3$.
We consider a state $\Psi \in L^2(\Lambda_{\ell}^{n_A})\otimes L^2(\Lambda_{\ell}^{n_B})$ satisfying the low energy condition
\begin{equation}
\langle \Psi, H_{n,m} \Psi \rangle \leq \frac{4\pi}{\ell^3} (n^2 a_A + 2  n m a_{AB} + m^2 a_B ) + C_{AB}\ell^3 \Big( \Big(\frac{n+m}{\ell^3} \Big) \bar{a} \Big)^{5/2},
\end{equation}
for $C_{AB} > 2I_{AB}$ uniformly in $\rho_A,\rho_B$. If such state doest not exist, then the theorem is proven. By Proposition \ref{propos:BEC} we have the condensation estimate
\begin{equation}
\langle \Psi, (n_+^A + n_+^B ) \Psi \rangle \leq (n+m) K_{\ell}^2 (\rho \bar{a}^3)^{\frac{1}{17}}.
\end{equation}

Since $\gamma < 4\eta + \frac{1}{34}$ and $K_H^3 K_{\ell} \leq C(\rho \bar{a}^3)^{-1/2}$, by Proposition \ref{prop:Hamgapbound}, there exists a sequence $\{\Psi_{(s,t)}\}_{s,t \in \mathbb{Z}}$ such that 
\begin{multline*}
\langle \Psi, H_{n,m} \Psi \rangle \geq \sum_{4|s+t| \leq \mathcal{M}}\langle \Psi_{(s,t)}, \big( H_{n,m}^{\text{gap}} + G \big) \Psi_{(s,t)} \rangle  - C (\rho \bar{a})^{5/2}\ell^3 (\rho \bar{a}^3)^{{\eta}}\\
 + \Big(\frac{4\pi}{\ell^3} (n^2 a_A + 2  n m a_{AB} + m^2 a_B ) + C_{AB}\ell^3 \Big( \Big(\frac{n+m}{\ell^3} \Big) \bar{a} \Big)^{5/2}\Big) \sum_{4|s+t|>\mathcal{M}}\|\Psi_{(s,t)}\|^2, 
\end{multline*}
If $|s+t| \leq \mathcal{M}/2$, then 
\begin{equation}\label{eq:statecontroln+1}
\Psi_{(s,t)}  = \one_{[0,\mathcal{M}/2]^2}(n_+^{AL},n_+^{BL}) \Psi_{(s,t)},
\end{equation}
and if we prove the lower bound for these states, using that $\sum_{(s,t) \in \mathbb{Z}^2}\|\Psi_{(s,t)}\|^2 =1$, we concluded the proof.
Therefore, we assume from now on to be working with a state $\widetilde{\Psi}$ satisfying \eqref{eq:statecontroln+1}, which implies
\begin{equation}
\langle \widetilde{\Psi}, n_+^{AL}\widetilde{\Psi}\rangle\leq \frac{\mathcal{M}}{2}, \qquad \langle \widetilde{\Psi}, n_+^{BL}\widetilde{\Psi}\rangle \leq \frac{\mathcal{M}}{2}.
\end{equation}
Since $K_H \geq C K_{\ell}^4$ and $K_{\ell} K_{H}^3 \leq (\rho \bar{a}^3)^{-\frac{1}{2}}$, we can use Lemma \ref{lem:boundGapTilde}, Proposition \ref{prop:secondquantHam}, Corollary \ref{cor:secondquant} and Lemma \ref{lem:c*substit} to get  
\begin{align}
H_{n,m}^{\text{gap}} +  G &\geq \widetilde{H}_{n,m} +\frac{1}{2}G - C (n+m) \rho \bar{a} \frac{R}{\ell} \nonumber \\
&\geq \frac{4\pi}{\ell^3} (n^2 a_A+ 2 n m a_{AB} +m^2 a_B) +\mu_A n +\mu_B m \nonumber\\
&\quad + \frac{1}{\pi^2}\int_{\mathbb{C}^2} \mathrm{d}z\; \mathcal{L}_{\mu_A,\mu_B}(z)|z_A\otimes z_B\rangle \langle z_A \otimes z_B| - C \ell^3(\rho \bar{a})^{5/2}(\rho \bar{a}^3)^{{\eta}}, \label{lowerbound:HgapG}
\end{align}
for chemical potentials satisfying $0< \mu_A,\mu_B \leq C \ell^{-2}$, where we used that $\ell = K_{\ell} (\rho \bar{a})^{-1/2}$ and $R\leq C_R \bar{a}(\rho \bar{a}^3)^{-\eta}$. We recall here the expression of the resulting Hamiltonian 
\begin{align}
\mathcal{L}_{\mu_A,\mu_B}  &=\mathcal{R}(z)+\mathcal{G}_{\text{gap}}(z)+\mathcal{G}_{\text{conv}}(z)+ \mathcal{G}_{\eta}(z)-\mu_A |z_A|^2 -\mu_B |z_B|^2, \nonumber\\
\mathcal{R}(z) &= \mathcal{K} +Z_0(z)+ Z^{ex}_2(z) + Z_{3,L}(z),
\end{align}
with all the operators being defined in Section \ref{sec:cnumber}. Denoting by $|z\rangle = |z_A \otimes z_B\rangle$, we can split the integral 
\begin{equation}
\frac{1}{\pi^2}\int_{\mathbb{C}^2} \mathrm{d}z\; \mathcal{L}_{\mu_A,\mu_B}(z) |z\rangle \langle z|= \mathfrak{I}_< + \mathfrak{I}_>  
\end{equation}
where $\mathfrak{I}_<, \mathfrak{I}_> $ are the same integrals on the regions $|z|^2\leq K_z(n+m)$ and $|z|^2> K_z (n+m)$, respectively.
\begin{itemize}
\item For $|z|^2 \leq K_z (n+m)$, thanks to assumption \eqref{cond:R} on $R$, since $\eta < \frac{1}{18}$, and 
\begin{align*}
\mathcal{M} &\leq \rho \ell^3 K_H^{-3} K_{\ell}^{-4} K_z^{-3},\quad K_{\ell} K_z^3K_H^2 \ll (\rho \bar{a}^3)^{-1/2}, \quad K_{\ell}^4 K_z^3 \leq CK_H,\\
K_z^5 &\leq C(\rho \bar{a}^3)^{-4\eta}, \qquad 
K_{\ell}^2K_z \delta_{AB}\bar{a}^{-1} \leq (1000C)^{-1} \quad (\text{both for } \eta \neq 0) ,
\end{align*}
by a combination of Corollary \ref{cor:lowerboundrightregion} and Lemma \ref{lem:calcBogInt}, using the assumption \eqref{assmpt:reabsorbSpGap} for the case $\eta \neq 0$, we have
\begin{align}
\mathcal{R}(z)+ \frac{1}{2}\mathcal{G}_{\text{gap}}(z)+ \mathcal{G}_{\eta}(z) \geq  \frac{|\Lambda|}{(2\pi)^3} G_{\rho_{z_A},\rho_{z_B}} ^{5/2} I_{\rho_{z_A},\rho_{z_B}}  -C(\rho \bar{a})^{5/2}(\rho \bar{a}^3)^{{\eta}}\ell^3, \label{eq:boundR<}
\end{align}
with
\begin{align}
G_{r_A,r_B} &= (r_A^2 a_A^2 + 2r_A r_B a_{AB}^2 + r_B a_B^2)^{1/2},\\
I_{r_A,r_B} &=  (8\pi)^{5/2}\frac{2\sqrt{2}}{15\pi^2} (\mu_+^{5/2} (r_A,r_B)+ \mu_-^{5/2}(r_A,r_B)\big),
\end{align}
where $\mu_{\pm}(r_A,r_B)$ are defined in \eqref{def:mu&xi}. We consider the functional 
\begin{align}
F(|z_A|^2,|z_B|^2) &:= \frac{8\pi\bar{a}}{|\Lambda|}(\rho \bar{a}^3)^{1/4}(|z_A|^4+ |z_B|^4) +\frac{|\Lambda|}{(2\pi)^3} G_{\rho_{z_A},\rho_{z_B}} ^{5/2} I_{\rho_{z_A},\rho_{z_B}} \nonumber \\
&\quad  -\mu_A |z_A|^2 -\mu_B |z_B|^2. \label{eq:Functionalconvex}
\end{align}
We observe that this functional is in the form \eqref{app:Fconvexfunctional}. Lemma \ref{lem:Fconvexfunctional} guarantees that $F$ is convex in $(|z_A|^2,|z_B|^2)$, and we can choose $(\mu_A,\mu_B)$ in the expression \eqref{eq:gradF} such that $\nabla F(n,m)= (0,0)$. This gives that $F$ attains its minimum for $(|z_A|^2,|z_B|^2)= (n,m)$, and we can therefore lower bound, using \eqref{eq:boundR<} and recalling the definition of $\mathcal{G}_{\text{conv}}$ in \eqref{eq:Gconv(z)},
\begin{align}
\mathfrak{I}_< \geq &\Big(\frac{|\Lambda|}{(2\pi)^3} G_{\frac{n}{|\Lambda|},\frac{m}{|\Lambda|}}^{5/2} I_{\frac{n}{|\Lambda|},\frac{m}{|\Lambda|}}   -\mu_A n -\mu_B m  -C(\rho \bar{a})^{5/2}(\rho \bar{a}^3)^{{\eta}}\ell^3 \Big) \nonumber\\
&\times \frac{1}{\pi^2} \int_{|z|^2\leq K_z (n+m)}\mathrm{d}z \,|z\rangle \langle z|.\label{bound:mathfrakI<}
\end{align}

\item For $|z|^2 > K_z (n+m)$, since $M\geq 3 +\nu^{-1}(12\eta +1)$, $K_H^3 \ll (\rho \bar{a}^3)^{\frac{1}{2} -\nu(M-2)} K_{\ell}^{-5}$, by Proposition \ref{prop:z>N} and since $\mathcal{G}_{\text{conv}}(z) \geq 0$,
\begin{multline}
\mathcal{R}(z) + \mathcal{G}_{\text{gap}}(z)+\mathcal{G}_{\text{conv}}(z) -\mu_A |z_A|^2 - \mu_B|z_B|^2\\
\geq \frac{1}{2}\rho^2 \ell^3 \bar{a}(\rho \bar{a}^3)^{\frac{1}{2}+3\eta -\nu M}\geq \frac{|\Lambda|}{(2\pi)^3} G_{\frac{n}{|\Lambda|},\frac{m}{|\Lambda|}}^{5/2} I_{\frac{n}{|\Lambda|},\frac{m}{|\Lambda|}},
\end{multline}
giving 
\begin{equation}\label{bound:mathfrakI>}
\mathfrak{I}_> \geq  \frac{|\Lambda|}{(2\pi)^3} G_{\frac{n}{|\Lambda|},\frac{m}{|\Lambda|}}^{5/2} I_{\frac{n}{|\Lambda|},\frac{m}{|\Lambda|}} \frac{1}{\pi^2}\int_{|z|^2> K_z (n+m)}\mathrm{d}z \,|z\rangle \langle z|.
\end{equation}
\end{itemize}
Therefore, plugging the bounds \eqref{bound:mathfrakI<} and \eqref{bound:mathfrakI>} into \eqref{lowerbound:HgapG}, we obtain
\begin{align*}
H_{n,m}^{\text{gap}} +  G &\geq  \frac{4\pi}{\ell^3} (n^2 a_A+ 2 n m a_{AB} +m^2 a_B) + \frac{|\Lambda|}{(2\pi)^3} G_{\frac{n}{|\Lambda|},\frac{m}{|\Lambda|}}^{5/2} I_{\frac{n}{|\Lambda|},\frac{m}{|\Lambda|}}  -C (\rho \bar{a})^{5/2}(\rho \bar{a}^3)^{{\eta}}\ell^3,
\end{align*}
where we used and using that $\pi^2 = \int_{\mathbb{C}^2}\mathrm{d}z \,|z\rangle \langle z|$. This proves the result.

\end{proof}

\section{Splitting of the Potential Energy and Renormalization}\label{sec:renorm}

We choose two parameters $(n,m) \in \mathbb{N}^2$ such that  $n,m\leq 100 C_a^2 \rho \ell^3$. Denoting by $u_0(x) = |\Lambda|^{-1/2}$ the normalized constant function on $\Lambda$, we define the four projectors
\begin{align*}
P^A &= \lvert u_0 \rangle \langle u_0 \rvert \otimes \one_{\mathscr{H}_B}, \qquad Q^A = (1 -P^A) \otimes \one_{\mathscr{H}_B},\\
P^B &= \one_{\mathscr{H}_A} \otimes \lvert u_0 \rangle \langle u_0 \rvert  , \qquad Q^B =   \one_{\mathscr{H}_A}\otimes(1 -P^B),
\end{align*}
such that $P^A + Q^A = 1 \otimes \one_{\mathscr{H}_B}$ and $P^B + Q^B = \one_{\mathscr{H}_A}\otimes 1$. 
These projectors let us introduce the number of particles in the condensate $n_0$ and the excited particles $n_+$ for both species
\begin{align}
n_0^A &:= \sum_{j=1}^{n} P^A_{x_j}, \qquad n_+^A := \sum_{j=1}^{n} Q^A_{x_j} = n_A -n_0^A, \label{eq:defn+A}\\
n_0^B &:= \sum_{j=1}^{m} P^B_{y_j}, \qquad n_+^B := \sum_{j=1}^{m} Q_{y_j}^B = n_B -n_0^B, \label{eq:defn+B}
\end{align}
and denote the total number of condensated and excited particles as $n_0:= n_0^A+n_0^B$, and $n_+ := n_+^A + n_+^B$, respectively. Obviously, we have 
\begin{equation}
n = n_0^A + n_+^A, \qquad m = n_0^B + n_+^B.
\end{equation}
By means of the projectors onto and outside the condensate, we split the potential in a sum of operators. We exploit the relations between the $P$'s and $Q$'s to write 
\begin{align}
	v_A(x_{i}-x_{j})&=(P^A_{x_i}+Q^A_{x_i})(P^A_{x_j}+Q^A_{x_j})v_A(x_{i}-x_{j})(P^A_{x_j}+Q^A_{x_j})(P^B_{x_i}+Q^A_{x_i}),\nonumber\\
		v_B(y_{i}-y_{j})&=(P^B_{y_i}+Q^B_{y_i})(P^B_{y_j}+Q^B_{y_j})v_B(y_{i}-y_{j})(P^B_{y_j}+Q^B_{y_j})(P^B_{y_i}+Q^B_{y_i}),\nonumber\\
			v_{AB}(x_{i}-y_{j})&=(P^A_{x_i}+Q^A_{x_i})(P^B_{y_j}+Q^B_{y_j})v_{AB}(x_{i}-y_{j})(P^A_{x_j}+Q^A_{x_j})(P^B_{y_i}+Q^B_{y_i}),\nonumber
\end{align}
and reorganize it as a sum of $\mathcal{Q}_j$, where in each $\mathcal{Q}_j$, the projector $Q$ is present $j$ times. We use suitable algebraic computations to make the $\mathcal{Q}_j$, for $j=0,1,2,3$, depend only on $g$, and the remaining terms we regroup in the positive $\mathcal{Q}_4$ term.
\begin{lemma}\label{lem:potential_splitting}
The following algebraic identities hold
\begin{equation}
\frac{1}{2}\sum_{i\neq j} v_A(x_i - x_j) = \sum_{j=0}^4 {\mathcal Q}_j^{A}, \qquad \frac{1}{2}\sum_{i\neq j} v_B(x_i - x_j) = \sum_{j=0}^4 {\mathcal Q}_j^{B},
\end{equation}
\begin{equation}
\sum_{i=1}^{n}\sum_{j=1}^{m} v_{AB}(x_i - x_j) = \sum_{j=0}^4 {\mathcal Q}_j^{AB},
\end{equation}
where,
\begin{align}
0 \leq {\mathcal Q}_4^{A}&:=
\frac{1}{2} \sum_{i\neq j} \Pi^{A*}_{ij} v_A(x_i-x_j) \Pi^A_{ij}\nonumber \\
\Pi^A_{ij} &:= Q^A_{x_i} Q^A_{x_j} + \omega_A(x_i-x_j)\left(P^A_{x_i} P^A_{x_j} + P^A_{x_i} Q^A_{x_j} + Q^A_{x_i} P^A_{x_j}\right),\label{eq:SF_DefQ4A}\\
{\mathcal Q}_3^{A}&:=
\sum_{i\neq j} P^A_{x_i} Q^A_{x_j} g_A(x_i-x_j) Q^A_{x_j} Q^A_{x_i} + h.c.,
 \label{eq:SF_DefQ3A}\\
{\mathcal Q}_2^{A}&:=
\sum_{i\neq j} P^A_{x_i} Q^A_{x_j} (g_A+ g_A\omega_A)(x_i-x_j)P^A_{x_j} Q^A_{x_i}    \nonumber\\
&\quad + \sum_{i\neq j} P^A_{x_i} Q^A_{x_j} (g_A+g_A\omega_A)(x_i-x_j)Q^A_{x_j} P^A_{x_i}   \nonumber\\
&\quad+\frac{1}{2}\sum_{i\neq j} P^A_{x_i} P^A_{x_j} g_A(x_i-x_j) Q^A_{x_j} Q^A_{x_i} + h.c.,
 \label{eq:SF_DefQ2A}\\
{\mathcal Q}_1^{A}&:=\sum_{i,j} \big(Q^A_{x_i} P^A_{x_j} (g_A + g_A \omega_A)(x_i-x_j) P^A_{x_j} P^A_{x_i}  + h.c.\big ) = 0,  \label{eq:SF_DefQ1A}
\intertext{and}
{\mathcal Q}_0^{A}&:= \frac{1}{2} \sum_{i \neq j} P^A_{x_i} P^A_{x_j} (g_A+ g_A\omega_A)(x_i-x_j) P^A_{x_j} P^A_{x_i}. \label{eq:SF_DefQ0A}
\end{align}
Analogous expressions hold for the $\mathcal{Q}^B_j$'s. 

For the $AB$ part:
\begin{align}
0 \leq {\mathcal Q}_4^{AB}&:=
 \sum_{i=1}^{n} \sum_{j=1}^{m} \Pi^{AB*}_{ij}  v_{AB}(x_i-y_j)\Pi^{AB}_{ij},\nonumber \\
\Pi^{AB}_{ij} &:=  Q^A_{x_i}Q^B_{y_j} + \omega_{AB}(x_i-y_j) \left(P^B_{y_j} P^A_{x_i} + P^B_{y_j} Q^A_{x_i} + Q^B_{y_j} P^A_{x_i}\right),\label{eq:SF_DefQ4AB}\\
{\mathcal Q}_3^{AB}&:=
\sum_{i=1}^{n}\sum_{j=1}^{m} \left(P^A_{x_i} Q^B_{y_j} + Q^A_{x_i}P^B_{y_j}\right)g_{AB}(x_i-y_j) Q^A_{x_j} Q^B_{y_i} + h.c.,
 \label{eq:SF_DefQ3AB}\\
{\mathcal Q}_2^{AB}&:=
\sum_{i=1}^{n} \sum_{j=1}^{m} \left( P^A_{x_i} Q^B_{y_j} + Q^A_{x_i} P^B_{y_j} \right) (g_{AB}+ g_{AB}\omega_{AB})(x_i-y_j)\left( P^A_{x_i} Q^B_{y_j} + Q^A_{x_i} P^B_{y_j}   \right) \nonumber \\
&\quad+\sum_{i=1}^{n} \sum_{j=1}^{m}P^A_{x_i} P^B_{y_j} g_{AB}(x_i-y_j) Q^B_{y_j} Q^A_{x_i} + h.c.,
 \label{eq:SF_DefQ2AB}\\
0={\mathcal Q}_1^{AB}&:=\sum_{i=1}^{n} \sum_{j=1}^{m} \big(Q^A_{x_i} P^B_{y_j} + P^A_{x_i} Q^B_{y_j} \big)(g_{AB} + g_{AB} \omega_{AB})(x_i-y_j) P^A_{x_i} P^B_{y_j}  + h.c.\big ),  \label{eq:SF_DefQ1AB}
\intertext{and}
{\mathcal Q}_0^{AB}&:= \sum_{i=1}^{n} \sum_{j=1}^{m} P^A_{x_i} P^B_{y_j} (g_{AB}+ g_{AB}\omega_{AB})(x_i-y_j) P^B_{y_j} P^A_{x_i}. \label{eq:SF_DefQ0AB}
\end{align}
\end{lemma} 

\begin{proof}
The lemma is proven by algebraic computations using that $v = g + v\omega$ to have elements composing $\mathcal{Q}_3$, that $v = g+ g\omega + v \omega^2$ for those composing $\mathcal{Q}_0$ and $\mathcal{Q}_1$ and both the relations for $\mathcal{Q}_2$. The remaining part gives $\mathcal{Q}_4$. The $\mathcal{Q}_1$'s are zero because, for any $f \in L^1(\Lambda)$,
\begin{equation*}
Q_i P_j f(x_i-x_j) P_j P_i =\frac{1}{\lvert \Lambda \rvert} \|f\|_{L^1} Q_i P_i = 0. 
\end{equation*}
\end{proof}

\section{Extraction of the spectral gaps}\label{sec:spgap}

The action of the Hamiltonian can be split into two classes of states, \textit{i.e.}, those with low and high energy. Let us recall the definition of \eqref{def:mu&xi} of $I_{AB}$.
\begin{assumption}\label{cond:psi_low_energ}
The normalized state $\Psi \in L^2(\Lambda_{\ell}^{n} )\otimes L^2(\Lambda_{\ell}^m)$ satisfies the \textit{low energy condition} if there exists a constant $C_{AB} >2 I_{AB}$ uniformly in $n,m$, such that 
\begin{align}
\langle \Psi, H_{n,m} \Psi \rangle \leq \frac{4\pi}{\ell^3} (n^2 a_A + 2  n m a_{AB} + m^2 a_B ) + C_{AB}\ell^3 \Big( \Big(\frac{n+m}{\ell^3} \Big) \bar{a} \Big)^{5/2}.
\end{align}
\end{assumption}
The states at high energy are clearly those satisfying the opposite inequality. For the last ones, the energy bound to the Lee-Huang-Yang order is given for free, therefore we focus our attention to states satisfying the low energy condition.

One of the reason why it is key to localize the problem in boxes $\Lambda_{\ell}$, is that in these boxes it is possible to prove Bose-Einstein condensation for the states at low energy, content of the next proposition.

\begin{proposition}\label{propos:BEC}
If there exists a normalized state $\Psi$ satisfying Assumption \ref{cond:psi_low_energ}, then, if $(\rho \bar{a}^3)^{-\frac{1}{17}} \leq n+m \leq C \rho \ell^3$,
\begin{equation}
\langle \Psi, (n_+^A + n_+^B )\Psi \rangle \leq (n+m) K_{\ell}^2(\rho \bar{a}^3)^{\frac{1}{17}}.
\end{equation}
\end{proposition}

The proof is presented in Appendix \ref{app:BEC}.
Many of the error terms throughout the paper can be expressed in terms of the number of excited particles $n_+^A,n_+^B$. The condensation estimate is fundamental for the control of the excitations by showing that we can restrict the analysis to states with a bounded number of low momenta excitations. The bounds in the following sections can be better expressed in momenta space, in this case represented by the space
\begin{equation}
\Lambda^* :=  \frac{\pi}{\ell} \mathbb{N}_0^3,
\end{equation}
and the spaces, for $K_H \gg 1$,
\begin{equation}
\mathcal{P}_L := \Big\{ p \in \Lambda^*,\, 0 < |p| \leq K_H \ell^{-1}\big\}, \qquad \mathcal{P}_H := \Big\{ p \in \Lambda^*,\, |p| > K_H \ell^{-1}\Big\},
\end{equation}
of low and high momenta, respectively, and the relative projectors
\begin{equation}
Q^{A,L}_{x_j} = \one_{\mathcal{P}_L} ((-\Delta_{x_j})^{1/2}), \qquad Q^{A,H}_{x_j} = \one_{\mathcal{P}_H} ((-\Delta_{x_j})^{1/2}).
\end{equation}
and analogous definition for the $B$ versions, where $-\Delta$ is the Neumann Laplacian on $\Lambda$. We introduce, as well, the number of low and high excitations
\begin{equation}
n_+^{A,L} = \sum_{j=1}^{n} Q^{A,L}_{x_j}, \qquad n_+^{A,H} = \sum_{j=1}^{n} Q^{A,H}_{x_j}.
\end{equation} 
We have clearly that 
\begin{equation}
1 = P^{A} + Q^{A,L} + Q^{A,H}, \qquad n^A_+ = n_+^{A,H} + n_+^{A,L}, 
\end{equation}
and analogous with $B$, with the total high and low excitations being $n_+^{H} = n_+^{A,H} + n_+^{B,H}$, $n_+^{L} = n_+^{A,L} + n_+^{B,L}$, respectively.
We are now ready to extract from the kinetic energy some terms that we will call spectral gaps:
\begin{equation}\label{gap:G}
G := \frac{\pi n_+}{4 \ell^4} + \frac{K_H n_+^H}{2 \ell^2}+ \frac{\pi n_+ n_+^L}{4\mathcal{M}\ell^2} + \frac{K_H n_+^L n_+^H}{2 \mathcal{M}\ell^2},
\end{equation}
for some large constant $\mathcal{M} \gg 1$ to be fixed later. We can therefore introduce the modified kinetic energy 
\begin{align*}
\sum_{j=1}^{N_A} \mathcal{T}_{x_j}^A  &=  \sum_{j=1}^{N_A} -\Delta_{x_j} - \frac{\pi n_+^A}{2 \ell^2} - \frac{K_H }{\ell^2}n_+^{A,H}, \quad  &&\sum_{j=1}^{N_B} \mathcal{T}_{y_j}^B  =  \sum_{j=1}^{N_B} -\Delta_{y_j} - \frac{\pi n_+^B}{2 \ell^2}- \frac{K_H }{\ell^2}n_+^{B,H} , \\
\mathcal{T}_{x_j}^A &:= -\Delta_{x_j} -\frac{\pi}{2\ell^2} Q_{x_j}^A - \frac{K_H}{\ell^2} Q_{x_j}^{A,H}\geq 0 , \quad &&\mathcal{T}_{y_j}^B := -\Delta_{y_j} -\frac{\pi}{2\ell^2} Q_{y_j}^B - \frac{K_H}{\ell^2} Q_{y_j}^{B,H}\geq 0. 
\end{align*}

The extraction of the $\mathcal{M}-$dependent gap is of particular importance because it lets us bound terms dependent on $n_+^2$, showing how the main contribute to the energy is given by particles whose low-momenta excitations are below a threshold $\mathcal{M}$. The following proposition contains the gap extraction useful for this purpose.

\begin{proposition}\label{prop:Hamgapbound}
There exists a constant $C>0$ such that the following holds. Consider a normalized state $\Psi  \in L^2(\Lambda_{\ell}^{n} )\otimes L^2(\Lambda_{\ell}^m)$ satisfying Assumption \ref{cond:psi_low_energ}, then there exists a sequence $\{\Psi_{(s,t)}\}_{s,t \in \mathbb{Z}} \subseteq L^2(\Lambda_{\ell}^{n} )\otimes L^2(\Lambda_{\ell}^m)$ such that $\sum_{s,t \in\mathbb{Z}} \|\Psi_{(s,t)}\|^2 = 1$ and
\begin{equation}
\Psi_{(s,t)}  = \mathbbm{1}_{[0, \frac{\mathcal{M}}{4}+ s]\times [0, \frac{\mathcal{M}}{4}+ t]}(n_+^{AL}, n_+^{BL}) \Psi_{(s,t)},
\end{equation}
and for $(\rho \bar{a}^3)^{-\frac{1}{17}}\leq  n+m \leq 100 C_a^2\rho \ell^3$, and $\mathcal{M} \geq \rho \ell^3 (\rho \bar{a}^3)^{\gamma}$, with $\gamma < 4\eta + \frac{1}{34}$ and $K_H^3 K_{\ell} \leq C (\rho \bar{a}^3)^{-1/2}$, then for $\rho \bar{a}^3 \leq C^{-1}$, we have that
\begin{multline*}
\langle \Psi, H_{n,m} \Psi \rangle \geq \sum_{4|s+t| \leq \mathcal{M}}\langle \Psi_{(s,t)}, \big( H_{n,m}^{\text{gap}} + G \big) \Psi_{(s,t)} \rangle  - C (\rho \bar{a})^{5/2}\ell^3 (\rho \bar{a}^3)^{\eta}\\
 + \Big(\frac{4\pi}{\ell^3} (n^2 a_A + 2  n m a_{AB} + m^2 a_B ) + C_{AB}\ell^3 \Big( \Big(\frac{n+m}{\ell^3} \Big) \bar{a} \Big)^{5/2}\Big) \sum_{4|s+t|>\mathcal{M}}\|\Psi_{(s,t)}\|^2, 
\end{multline*}
where
\begin{align}
	H_{n,m}^{\text{gap}} =& \sum_{j=1}^{n}  \mathcal{T}_{x_j}^A + \sum_{1 \leq i < j \leq n} v_A(x_i-x_j) \nonumber\\
	  &+ \sum_{j=1}^{m}  \mathcal{T}_{y_j}^B + \sum_{1 \leq i < j \leq m} v_B(x_i-x_j) + \sum_{j=1}^{n} \sum_{k=1}^{m} v_{AB}(x_j-y_k).
\end{align}
\end{proposition}

\begin{proof}
By Lemma \ref{lem:locLargeMatrix} we can define $\Psi_{(s,t)} = \tilde{\theta}(n_+^{AL}-s, n_+^{BL}-t)\Psi$ such that 
\begin{align*}
&\langle \Psi, H_{n,m} \Psi\rangle  -  \sum_{4|s +t|\leq \mathcal{M}}\langle \Psi_{(s,t)}, H_{n,m} \Psi_{(s,t)} \rangle \\
&\geq \sum_{4|s+t|\geq \mathcal{M}}\langle \Psi_{(s,t)}, H_{n,m} \Psi_{(s,t)} \rangle
 - \frac{C}{\mathcal{M}^2}\sum_{h,k = 0,1,2} \langle \Psi, d_{h,k}\Psi\rangle.
\end{align*}
By construction, the states $\Psi_{(s,t)}$ for which $4|s+t| > \mathcal{M}$, since $n_+ = n_+^L + n_+^H$, also have that 
\begin{equation}
\langle \Psi_{(s,t)}, n_+\Psi_{(s,t)}\rangle \geq \langle \Psi_{(s,t)}, n_+^L\Psi_{(s,t)}\rangle \geq \frac{\mathcal{M}}{2}  \|\Psi_{(s,t)}\|^2,
\end{equation}
where we used the localization functions $\bar{\theta}(n_+^{AL} - s,n_+^{BL}-t)$ in the definition of $\Psi_{(s,t)}$. Due to the condition $\mathcal{M} \geq \rho \ell^3 (\rho \bar{a}^3)^{\gamma}$ and $\gamma < 4\eta + \frac{1}{34}$, this goes against the result of Proposition \ref{propos:BEC}, therefore those $\Psi_{(s,t)}$'s cannot satisfy the low energy condition in Assumption \ref{cond:psi_low_energ}, that is 
\begin{equation}
\langle \Psi_{(s,t)}, H_{n,m} \Psi_{(s,t)} \rangle  \geq \Big(\frac{4\pi}{\ell^3} (n^2 a_A + 2  n m a_{AB} + m^2 a_B ) + C_{AB}\ell^3 \Big( \Big(\frac{n+m}{\ell^3} \Big) \bar{a} \Big)^{5/2}\Big) \|\Psi_{(s,t)}\|^2.
\end{equation}
By Lemma \ref{lem:estimdAB}, the condensation estimate, the assumptions on $\mathcal{M}$, $\gamma< 4\eta + \frac{1}{34}, K_H^3 K_{\ell} \leq C (\rho \bar{a}^3)^{-1/2}$ and $K_{\ell}= (1000C)^{-1}(\rho \bar{a}^3)^{-2\eta}$, we get 
\begin{align*}
 \frac{1}{\mathcal{M}^2}\sum_{h,k = 0,1,2} \langle \Psi, d_{h,k}\Psi\rangle  &\leq \frac{C}{\mathcal{M}^2} \big( \langle H_{n,m}\rangle_{\Psi}  +C_1 K_H^3 \ell^{-3}(n+m) \langle n_+ \rangle_{\Psi}  \big)\\
&\leq   C K_{\ell}^{-6} (\rho \bar{a}^3)^{1-2\gamma} \rho \ell^3 \bar{a}\big( 1+ K_H^3 K_{\ell}^2 (\rho \bar{a}^3)^{\frac{1}{17}}\big) \leq C (\rho \bar{a})^{5/2} \ell^3 (\rho \bar{a}^3)^{\sigma_{\eta}},
\end{align*}
and this concludes the proof.
\end{proof}

\section{Symmetrization and second quantization}\label{sec:sym}

We need the following lemma to facilitate the emergence of the contribution of the so-called ``soft pairs": in the cubic term a pair of incoming particles with high momenta may interact and turn into a couple of particles, one with low momentum and the other in the condensate. We first deal with the outgoing momenta for technical reasons, postponing the extraction of the incoming momenta to Lemma \ref{lem:softpairs}.

\begin{lemma}\label{lem:approxQ3-Q3low}
For any $\varepsilon >0$, there exists a $C>0$ such that if $\rho \bar{a}^3 \leq C^{-1}$ and $K_H \geq CK_{\ell}^4$, then, for $\xi \in \{A,B,AB\}$,
\begin{equation}
Q_{3}^{\xi} \geq Q_{3}^{\xi,\text{low}} - \varepsilon Q_4^{\xi} - \varepsilon\frac{n_+}{\ell^2} - \varepsilon \frac{K_H n_+^H}{\ell^2}, 
\end{equation}
where 
\begin{equation}
Q_3^{\xi,\text{low}} := \sum_{i \neq j} P_i^{\xi} Q_j^{\xi,L} g(x_i-x_j) Q_i^{\xi} Q_j^{\xi} + h.c.
\end{equation}
\end{lemma}
Choosing a small number for $\varepsilon$, for example $\varepsilon = \frac{1}{100}$, is enough to absorb the last two error terms in the spectral gap $G$.

\begin{proof}
The estimates are the same as in \cite[Lemma 2.4]{freeEnCPHM} treating each potential term separately for $\xi \in \{A,B\}$, and we refer to that paper for the proof, weighting the Cauchy-Schwarz by $\varepsilon K_H^{-1}$ instead of $\varepsilon K_{\ell}^{-1}$.
\end{proof}

In order to deal with the Neumann boundary conditions, we introduce the orthonormal bases $\{u_k\}_{k \in \Lambda^*}$ and $\{v_h\}_{h \in \Lambda^*}$ for the Neumann Laplacians $\sum_{j=1}^{N_A}-\Delta_{x_j}$ and $\sum_{j=1}^{N_B} -\Delta_{y_j}$ respectively, defined as
\begin{equation}\label{neu:basis}
u_k(x) = \frac{1}{\sqrt{|\Lambda|}}\prod_{j=1}^3 \sigma_{k_j}\cos(k_j x_j), \qquad v_h(y) = \frac{1}{\sqrt{|\Lambda|}}\prod_{j=1}^3  \sigma_{k_j}\cos(h_j y_j), 
\end{equation}
with
\begin{equation}\label{sigmajdef}
 \sigma_{k_j} = \begin{cases} 1,  &\text{ if } k_j = 0, \\
\sqrt{2}, &\text{ if } k_j \neq 0. \end{cases}
\end{equation}
The Fourier coefficients of the potentials $v_A,v_B,v_{AB}$ are not diagonal in these bases, therefore we need to replace them with their symmetrized versions and then estimate the errors made.
We recall and adapt the symmetrization result from \cite[Section 2.5]{freeEnCPHM} to the case of two types of bosons. The proof is totally analogous treating each of our potentials separately.
For a function $f \in L^1(\Lambda)$, we define its symmetrization by
\begin{equation}
\tilde{f}(x,y) := \sum_{z \in \mathbb{Z}^3} f(p_z(x) - y),
\end{equation} 
where $p_z$ is the mirror transformation defined by 
\begin{equation}
(p_z(x))_i = (-1)^{z_i} \Big( x_i - \frac{\ell}{2}\Big) + \frac{\ell}{2} + \ell z_i, \qquad i = 1,2,3.
\end{equation}
If $f$ is radial, then the Fourier coefficients of $\tilde{f}$ are diagonal in the Neumann basis.

We denote by $\widetilde{\mathcal{Q}}_j^{A},\widetilde{\mathcal{Q}}_j^{B}, \widetilde{\mathcal{Q}}_j^{AB}$ the symmetrized versions of $\mathcal{Q}_j^{A}, \mathcal{Q}_j^{B}, \mathcal{Q}_j^{AB}$ for each $j=0,2,3$, respectively, where, in their definition, the $g$'s have been replaced by $\widetilde{g}$'s, and  
\begin{align}
	\widetilde{H}_{n,m} := \sum_{j=1}^{n}  \mathcal{T}_{x_j}^A + \sum_{j=1}^{m}  \mathcal{T}_{y_j}^B+ \sum_{j=0,2} \widetilde{\mathcal{Q}}_j^A + \sum_{j=0,2}\widetilde{\mathcal{Q}}_j^B + \sum_{j=0,2} \widetilde{\mathcal{Q}}_j^{AB} + \sum_{\xi \in \{A,B,AB\}}\widetilde{\mathcal{Q}}_3^{\xi,L}.
\end{align}

The error made substituting $H_{n,m}$ with its symmetrized version $\widetilde{H}_{n,m}$ can be reabsorbed in a fraction of the spectral gap and by a small error for the required precision.

\begin{lemma}\label{lem:boundGapTilde}
Let $v_A, v_B, v_{AB}$ be non-increasing functions. For any $\varepsilon >0$, there exists a constant $C>0$ such that, if $\rho \bar{a}^{3} \leq C^{-1}$, $K_H \geq C K_{\ell}^4$ and $K_{\ell} K_{H}^3 \leq (\rho \bar{a}^3)^{-\frac{1}{2}}$, then 
\begin{equation}
H_{n,m}^{\text{gap}} \geq \widetilde{H}_{n,m}- C  \rho (n+m)\bar{a} \frac{R}{\ell} - \varepsilon G.
\end{equation}
\end{lemma}

We remark that here it is fundamental the assumption of $v_A,v_B,v_{AB}$ descreasing to estimate the errors coming from the substitutions of the potential terms by their symmetrized versions.

\begin{proof}
We first use Lemma \ref{lem:approxQ3-Q3low} to substitute the $\mathcal{Q}_3^{\xi}$ terms with the $\mathcal{Q}_3^{\xi,\text{low}}$ sacrificing a small fraction of the $\mathcal{Q}_4^{\xi}$ terms and of the spectral gap. We then estimate the remaining part of the $\mathcal{Q}_4^{\xi}$'s by zero. We then use \cite[Theorem 2.6]{freeEnCPHM} to substitute the $\mathcal{Q}_j^{\xi}$'s by the $\widetilde{\mathcal{Q}}_j^{\xi}$'s separately for each $\xi \in \{A,B,AB\}$, and obtain the result summing back everything together.
\end{proof}

We continue our analysis in momentum space considering the second quantization of the Hamiltonian. We use the construction introduced in the first section starting from formula \eqref{eq:Fockspacesmix} applied with our choice of the Neumann basis \eqref{neu:basis}. Let us introduce the creation and annihilation operators $a_k, a^*_k, b_k, b^*_k$ of bosons with momentum $k \in \Lambda^* $ of type $A$ and $B$, respectively, as in \eqref{eq:creationA}, \eqref{eq:creationB}. They satisfy the canonical commutation relations (CCR) \eqref{eq:CCR} for $h,k \in \frac{\pi}{\ell}\mathbb{N}_0^3$.
Note that for zero momentum, $a_0^\dagger$ creates the function $u_{0}=1$ and $b_0^{\dagger}$ creates the function $v_0=1$, corresponding to \emph{condensates} of species $A$ and $B$ in $\Lambda$, respectively. 

The operator $\widetilde{H}_{n,m}$ can be written, by abuse of notation, as the restriction on the $(n,m)-$boson space of a second quantized Hamiltonian $\mathcal{H}$ acting on the Fock space 
\begin{equation}
\mathscr{F} = \mathscr{F} _A \otimes \mathscr{F} _B =  \bigoplus_{N=0}^{\infty} \bigoplus_{\substack{n,m\geq 0\\ n+m = N}} L^2_s(\Lambda^n) \otimes  L^2_s(\Lambda^m).
 \end{equation}
 
We can extend the definition of the creation and annihilation operators to momenta $p \in \frac{\pi}{\ell}\mathbb{Z}^3$ as 
\begin{equation}
a_p^{\#} := a^{\#}(u_{(|p_1|,|p_2|, |p_3|)}), \qquad b_p^{\#} := b^{\#}(u_{(|p_1|,|p_2|, |p_3|)}),
\end{equation}
and also introduce the space $\Lambda^*_+ =\Lambda^*\setminus\{0\}$ and the set of generalized low momenta as 
\begin{equation}
\mathcal{P}_{L}^{\mathbb{Z}} := \Big\{p \in \frac{\pi}{\ell} \mathbb{Z}^3\,\Big|\, 0 < |p| < K_H \ell^{-1}\Big\}.
\end{equation}

In the proposition below, we write the explicit expression of the second quantized Hamiltonian $\mathcal{H}$.
\begin{proposition}\label{prop:secondquantHam}
We have the following identities on the $(n,m)-$sector of the Fock space $\mathscr{F}$:
\begin{equation}
\mathcal{H}|_{L^2_s(\Lambda^n)\otimes L^2_s(\Lambda^m)} = \widetilde{H}_{n,m}, \qquad \mathcal{G}|_{L^2_s(\Lambda^n)\otimes L^2_s(\Lambda^m)} = G,
\end{equation}
where $\mathcal{G}$ has the same definition \eqref{gap:G} as $G$, with the proper extension of the number operators, and
\begin{equation}
\mathcal{H} = \mathcal{H}_A + \mathcal{H}_B + \mathcal{H}_{AB},
\end{equation}
with
\begin{align*}
\mathcal{H}_A &:= \mathcal{Z}_0^A + \mathcal{Z}_2^A  +\mathcal{Z}_3^{A,L},\\
\mathcal{H}_B &:= \mathcal{Z}_0^B + \mathcal{Z}_2^B  +\mathcal{Z}_3^{B,L},\\
\mathcal{H}_{AB} &:= \mathcal{Z}_0^{AB} + \mathcal{Z}_2^{AB}  +\mathcal{Z}_3^{AB,L}.   
\end{align*}
The $A$ part reads
\begin{align*}
\mathcal{Z}_0^A &:= \mathcal{Z}_{0g}^A + \mathcal{Z}_{o\omega}^A, \\
\mathcal{Z}_{0g}^A &:=\frac{\widehat{g}_A(0)}{2|\Lambda|} (n(n-1) -n^A_+ (n^A_+-1)), \qquad \mathcal{Z}_{0\omega}^A :=\frac{\widehat{g\omega}_A(0)}{2|\Lambda|} a^*_0a^*_0 a_0a_0 , \\
\mathcal{Z}_2^A &:= \sum_{p \in {\Lambda}^*_+} \Big( \tau(p) a_p^* a_p + \frac{\widehat{g}_A(p)}{|\Lambda |} a_0^* a_p^* a_p a_0 + \frac{ \widehat{g}_A(p) }{2 |\Lambda |}\big( a_0^* a_0^* a_p a_p + h.c. \big) \Big)\\
& \quad + \frac{1}{|\Lambda|} \sum_{p\in {\Lambda}^*_+}(\widehat{g \omega}_A(0)+\widehat{g \omega}_A(p))a_0^* a_p^*a_p a_0 ,\\
\mathcal{Z}_3^{A,L} &:=  \frac{1}{|\Lambda|} \sum_{k \in {\Lambda}^*_+, p \in \mathcal P_L^{\mathbb Z} }\sigma(q,k) \widehat{g}_A(k) \big( a_0^* a_p^* a_{p-k} a_k + h.c. \big) ,
\end{align*}
where 
\begin{equation}\label{eq:tau}
\tau(p) = |p|^2 - \frac{\pi}{2\ell^2}\one_{\lbrace p \neq 0 \rbrace} -\frac{K_H}{\ell^2}\one_{\{p \in \mathcal{P}_H\}}.
\end{equation}
is the symbol of the kinetic energy $T$ and $c(q,s)$ are the normalizing factors given by 
\begin{equation}
\sigma(q,s) := \prod_{i=1}^3 \frac{\sigma_{q_i-s_i}}{\sigma_{q_i} \sigma_{s_i}},
\end{equation}
which is equal to $\frac{1}{8}$ if all the momenta in the product are different from zero.
The $B$ part is totally analogous substituting the $a$'s with the $b$'s and $g_A, \omega_A$ with $g_B, \omega_B$.
For the $AB$ part we have
\begin{align*}
\mathcal{Z}_0^{AB} &:= \mathcal{Z}_{0g}^{AB} + \mathcal{Z}_{0\omega}^{AB},\\
\mathcal{Z}_{0g}^{AB} &:= \frac{\widehat{g}_{AB}(0)}{|\Lambda|} (nm -n_+^A n_+^B), \qquad \mathcal{Z}_{0\omega}^{AB}:=\frac{\widehat{g\omega}_{AB}(0)}{|\Lambda |} a_0^*a_0 b^*_0 b_0, \\
\mathcal{Z}_2^{AB} &:= \frac{1}{|\Lambda|}\sum_{p \in {\Lambda}^*_+} (\widehat{g}_{AB}(p) + \widehat{g\omega}_{AB}(p))\big( b^*_0 a_0 a^*_p b_p + a^*_0 b_0 b^*_p a_p \big)  \\
&\quad  +\frac{1}{|\Lambda|} \sum_{p\in {\Lambda}^*_+}  \widehat{g}_{AB}(p) \big( a_0^* b_0^* a_p b_p + h.c. \big) 	+ \frac{\widehat{g \omega}_{AB}(0)}{|\Lambda|} \sum_{p\in {\Lambda}^*_+} \big( a_0^*  b^*_p b_p a_0  +  b^*_0 a^*_p a_p b_0 \big),\\
\mathcal{Z}_3^{AB,L} &:=  \frac{1}{|\Lambda|} \sum_{	k \in {\Lambda}^*_+,  p \in \mathcal{P}_L^{\mathbb Z}} \sigma(p,k) \widehat{g}_{AB}(k) \big( a_0^* b^*_{p} b_{p-k} a_k   + b^*_0 a^*_p a_{p-k} b_k + h.c. \big). 
\end{align*}
\end{proposition}

\begin{proof}
We follow the same strategy as \cite[Lemma 2.8]{freeEnCPHM} using the quantization rules \eqref{secondquant:1body}, \eqref{secondquant:2body}  for the 1- and 2-body operators introduced in the first section of the present paper.  
We also exploit the relation, for a radial integrable function $f:\mathbb{R}^3 \rightarrow \mathbb{R}$ with $\mathrm{supp} f \subseteq B(0,R)$, for $R \leq \ell/2$ and for $p,q \in \frac{\pi}{\ell}\mathbb{N}^3_0$,
\begin{equation}
\int_{\Lambda^2} \mathrm{d}x \, \mathrm{d} y \, u_p(x) \tilde{f}(x,y) u_q(y) = \delta_{p,q} \widehat{f}(p),
\end{equation}
where we recall that $\tilde{f}$ is the symmetrization of $f$. The same applies substituting one of both of the $u$'s by $v$'s, being the same functions. This gives, for example, the second quantization of the term $\widetilde{\mathcal{Q}}_0^A$ to be
\begin{equation}\label{eq:aaaa2ndquantA}
\frac{1}{2|\Lambda|}(\widehat{g}_A(0) + \widehat{g\omega}_A(0)) a^*_0a^*_0 a_0a_0,
\end{equation}
while for the $\widetilde{\mathcal{Q}}_2^A$ the second quantization reads
\begin{equation}\label{eq:Z22ndquantA}
\mathcal{Z}_2^A + \frac{\widehat{g}_A(0)}{|\Lambda|} \sum_{p \in \Lambda^*} a^*_0 a^*_p a_p a_0 = \mathcal{Z}_2^A + \frac{\widehat{g}_A(0)}{2|\Lambda|}(2n^A_+(n-n^A_+)) .
\end{equation}
Now, using that $a^*_0 a^*_0 a_0 a_0 = n(n-1) + (n^A_+)^2-2n n^A_+ +n^A_+$, and adding the two terms
\begin{equation}
\eqref{eq:aaaa2ndquantA} + \eqref{eq:Z22ndquantA}  =\mathcal{Z}_0^A + \mathcal{Z}_2^A.
\end{equation}
An analogous calculation gives the $\mathcal{Z}_0^B + \mathcal{Z}_2^B$.

For the quantization of the $\widetilde{\mathcal{Q}}_0^{AB}$ we get 
\begin{equation}\label{eq:abbBa2ndquantAB}
\frac{1}{|\Lambda |} \big(\widehat{g}_{AB}(0) +\widehat{g\omega}_{AB}(0)\big)  a_0^*a_0 b^*_0 b_0,
\end{equation}
while for the quantization of the $\widetilde{Q}_2^{AB}$ term is
\begin{equation}\label{eq:Z22ndquantAB}
\mathcal{Z}^{AB}_2 + \frac{\widehat{g}_{AB}(0)}{|\Lambda|} \sum_{p \in \Lambda^*} (a^*_0 b^*_p b_p a_0 + b^*_0 a^*_p a_p b_0) = \mathcal{Z}^{AB}_2 + \frac{\widehat{g}_{AB}(0)}{|\Lambda|} (nn_+^B + m n_+^A - 2n_+^A n_+^B).
\end{equation}
We observe that $a_0^* a_0 b^*_0 b_0 = nm + n_+^A n_+^b -n_+^Am -n_+^B n$, therefore 
\begin{equation}
\eqref{eq:abbBa2ndquantAB} + \eqref{eq:Z22ndquantAB} = \mathcal{Z}_0^{AB} + \mathcal{Z}_2^{AB}.
\end{equation}

The expression of the $\mathcal{Z}_3$ terms can be obtained following the lines of the argument for 1 type of boson in \cite[Lemma 2.8]{freeEnCPHM}.
Collecting all the previous equations we obtain the result.
\end{proof}

\begin{remark}
We observe, recalling that $K_{\ell} = (1000C)^{-1}(\rho \bar{a}^3)^{-2\eta}$, that we can add to $\mathcal{H}$ the following terms,
\begin{align*}
\mathcal{G}_{\text{conv}} &:= \frac{8\pi \bar{a}}{|\Lambda|}(\rho \bar{a}^3)^{1/4} ((n_0^A)^2-n^2 + (n_0^B)^2 -m^2),\\
\mathcal{G}_{\eta} &:=  \frac{64\pi \bar{a}}{|\Lambda|} \Big((n_0^A-n +n_0^B-m )n_+  + K_{\ell}^{-1}|\Lambda|^{-1/2} \big((n_0^A+n_0^B\big)^{5/2} - (n+m)^{5/2} \big)\Big)\one_{\eta = 0}, 
\end{align*}
for free for a lower bound because they come from adding  negative terms. Indeed, when restricted to $L^2_s(\Lambda^n)\otimes L^2_s(\Lambda^m)$, we have the bounds $n_0^A \leq n, n_0^B \leq m$.
The presence of $\mathcal{G}_{\text{conv}}$ is useful, a posteriori, to guarantee the convexity of the functional $F$ defined in \eqref{eq:Functionalconvex} which gives that $(|z_A|^2, |z_B|^2) = (n,m)$ is the point of minimum for $\mathcal{L}_{\mu_A,\mu_B}$ thanks to Lemma \ref{lem:Fconvexfunctional} in Appendix \ref{sec:Fconvex}.
The term $\mathcal{G}_{\eta}$, on the other hand, helps to deal with the error $E_{\omega}$ coming from the analysis of the $\mathcal{Q}_3$ term in \eqref{def:Eomega} for the case $\eta = 0$. Indeed, without this term, $E_{\omega} = \mathcal{O}(K_z \delta_{AB} \bar{a}^{-1} n_+\ell^{-2})$ would not be reabsorbed in the spectral gap because we need $K_z \gg 1$ but we allow $\delta_{AB}$ not to be small in this case. Also, it helps to bound the artificial error obtained in the approximation of the series by the Bogoliubov integral in Lemma \ref{lem:calcBogInt}.

 We have therefore the following corollary of Proposition \ref{prop:secondquantHam}.
\end{remark}

\begin{corollary}\label{cor:secondquant}
Under the same assumptions of Proposition \ref{prop:secondquantHam}, we have the following bound
\begin{equation}
\widetilde{H}_{n,m}+ G \geq \big(\mathcal{H}+ \mathcal{G}+ \mathcal{G}_{\text{conv}}+\mathcal{G}_{\eta}\big)|_{L^2_s(\Lambda^n)\otimes L^2_s(\Lambda^m)}. 
\end{equation}
\end{corollary}

\section{c-number substitution and Bogoliubov Hamiltonian}\label{sec:cnumber}

We first recall here a useful result, immediate consequence of \cite[Lemma 7.2]{freeEnCPHM}, thanks to which we have that the $\widehat{g\omega}(0)-$terms can be approximated by the following sums.
\begin{lemma}\label{lem:reconstructiongomega} 
We introduce the three sums
\begin{equation}
G_{\omega}^{\#} := \frac{1}{8 |\Lambda|} \sum_{k \in \frac{\pi}{\ell}\mathbb{Z}^3\setminus\{0\}} \frac{\widehat{g}_{\#}^2(k)}{2\tau_k}, \qquad \#\in \{A,B,AB\}. 
\end{equation}
There exists a constant $C>0$ such that the following estimates hold
\begin{align}
\Big|G_{\omega}^{\#}- \widehat{g\omega}_{\#}(0) \Big| &\leq  C \frac{\bar{a}^2}{\ell},\\
\frac{1}{8 |\Lambda|} \sum_{k \in \mathcal{P}_L^{\mathbb{Z}}} \frac{\widehat{g}_{\#}^2(k)}{2k^2} &\leq K_H\frac{\bar{a}^2}{\ell},  \qquad \# \in \{A,B,AB\}. \label{ineq:Gomegalow}
\end{align}
\end{lemma}
\begin{proof}
The second inequality is exactly \cite[Lemma 7.2 (ii)]{freeEnCPHM}. By \cite[Lemma 7.2 (i)]{freeEnCPHM} we have 
\begin{equation}
\Big| \frac{1}{8 |\Lambda|} \sum_{k \in \frac{\pi}{\ell}\mathbb{Z}^3\setminus\{0\}} \frac{\widehat{g}_{\#}^2(k)}{2 k^2}- \widehat{g\omega}_{\#}(0) \Big| \leq  C \frac{\bar{a}^2}{\ell}.
\end{equation}
It remains to estimate \eqref{ineq:Gomegalow},
\begin{align*}
\Big| G_{\omega}^{\#}- \frac{1}{8 |\Lambda|} \sum_{k \in \frac{\pi}{\ell}\mathbb{Z}^3\setminus\{0\}} \frac{\widehat{g}_{\#}^2(k)}{2 k^2}\Big| &\leq \frac{C\bar{a}^2}{\ell^5} \sum_{k \Lambda^*}\frac{1}{k^4}\leq  \bar{a}^2 \ell^{-1},
\end{align*}
which concludes the proof.
\end{proof}

Bogoliubov \cite{Bog} suggests in his paper the following heuristic approximation
\begin{equation}
a^{\#}_0 \simeq \sqrt{n}, \qquad b^{\#}_0 \simeq \sqrt{m},
\end{equation}
that is, almost all the particles should condensate. In order to rigorously perform the aforementioned approximation, we use a technique called c-number substitution, content of the next lemma.

We now split $L^2(\Lambda) = \mathrm{Ran}P_A \oplus \mathrm{Ran}Q_A$ and same for the $B$, which gives the decomposition on the Fock spaces
\begin{equation}
\mathscr{F}_A\otimes \mathscr{F}_B = \mathscr{F}_s(\mathrm{Ran}P_A) \otimes \mathscr{F}_s(\mathrm{Ran}Q_A) \otimes \mathscr{F}_s(\mathrm{Ran}P_B) \otimes \mathscr{F}_s(\mathrm{Ran}Q_B).
\end{equation}
We introduce the following families of coherent states, for $z \in \mathbb{C}^2$  
\begin{equation}
|z_A\rangle  = e^{-\big( \frac{|z_A|^2}{2} + z_A a^*_0\big)}\Omega_A \in \mathscr{F}_s(\mathrm{Ran}P_A), \qquad |z_B\rangle  = e^{-\big( \frac{|z_B|^2}{2} + z_B b^*_0\big)}\Omega_B \in \mathscr{F}_s(\mathrm{Ran}P_B),
\end{equation}
which are eigenvectors of the annihilation operators: $a_0 |z_A\rangle = z_A |z_A\rangle,$ and $b_0 |z_B\rangle = z_B |z_B\rangle$, and for which the following decomposition holds
\begin{equation}\label{coherent:expansion}
\pi^2 = \int_{\mathbb{C}^2} \mathrm{d}z_A\, \mathrm{d} z_B \; |z_A \otimes z_B \rangle \,\langle z_A \otimes z_B|,
\end{equation}
where $\langle z_A \otimes z_B|$ is a partial trace on $\mathscr{F}_s(\mathrm{Ran}P_A) \otimes \mathscr{F}_S(\mathrm{Ran}P_B)$, giving that, for any $\Psi \in \mathscr{F}_A\otimes \mathscr{F}_B$, the state
\begin{equation}
\Phi(z) = \langle z_A\otimes z_B | \Psi\rangle \in \mathscr{F}_s(\mathrm{Ran}Q_A)\otimes \mathscr{F}_s(\mathrm{Ran}Q_B).
\end{equation}
Since expanding for $z \in \mathbb{C}^2$ implies that $|z_A|^2,|z_B|^2$, corresponding to the number of particles of type $A$ and $B$ in the condensate, can take any value, we introduce the chemical potentials $(\mu_A,\mu_B)$ to make sure to have a control on the number of bosons.

\begin{lemma}\label{lem:c*substit}
There exists a $C >0$ and $\varepsilon > 0$ such that for all $n+m \leq 100 C_a^2 \rho \ell^3$, and $C \leq \mathcal{M} < C^{-1} \ell \bar{a}^{-1}$, $K_H \geq C K_{\ell}^4$, $0 <10 \mu_A,10\mu_B\leq \ell^{-1}$ we have, for any $M \in \mathbb{N}, M >3$, indicating by $z=(z_A,z_B) \in \mathbb{C}^2$,
\begin{align}
\mathcal{H} + \frac{1}{2}\mathcal{G} + \mathcal{G}_{\text{conv}} + \mathcal{G}_{\eta}&\geq \frac{4\pi}{|\Lambda|} (n^2 a_A+ 2 n m a_{AB} + m^2 a_B) + \mu_A n + \mu_B m  \nonumber\\
&\quad + \frac{1}{\pi^2}\int_{\mathbb{C}^2} \mathrm{d}z\; \mathcal{L}_{\mu_A,\mu_B}(z) |z_A\otimes z_B\rangle \langle z_A\otimes z_B| - C K_{\ell}\rho \bar{a},\label{eq:cnumberformula}
\end{align}
where
\begin{align*}
\mathcal{L}_{\mu_A,\mu_B}(z) &:= \mathcal{L}_A (z)+ \mathcal{L}_B (z)+ \mathcal{L}_{AB}(z) + \mathcal{G}_{\text{gap}} (z)+ \mathcal{G}_{\mathrm{conv}}(z) + \mathcal{G}_{\eta}(z)-\mu_A |z_A|^2 -\mu_B |z_B|^2,\\
\mathcal{L}_A &:= Z_0^A + Z_2^A  +Z_3^{A,L},\\
\mathcal{L}_B &:= Z_0^B + Z_2^B  +Z_3^{B,L},\\
\mathcal{L}_{AB} &:= Z_0^{AB} + Z_2^{AB}  +Z_3^{AB,L}.   
\end{align*}
Denoting by $\rho_{z_A} := |z_A|^2 \ell^{-3}$, $\rho_{z_B} := |z_B|^2 \ell^{-3}$, the $A$ part reads
\begin{align*}
Z_0^A &:= \frac{\rho_{z_A}}{2} G_{\omega}^A , \\
Z_2^A &:= \sum_{p \in {\Lambda}^*_+} \Big( \tau(p) a_p^* a_p + \rho_{z_A} \widehat{g}_A(p)  a_p^* a_p  + \frac{ \widehat{g}_A(p) }{2 |\Lambda |}\big( \bar{z}_A^2 a_p a_p + h.c. \big) \Big)\\
& \quad + \rho_{z_A} \sum_{p\in {\Lambda}^*_+}(\widehat{g \omega}_A(0)+\widehat{g \omega}_A(p))a_p^*a_p  ,\\
Z_3^{A,L} &:=  \frac{1}{|\Lambda|} \sum_{k \in {\Lambda}^*_+, p \in \mathcal P_L^{\mathbb Z} }\sigma(p,k) \widehat{g}_A(k) \big( \bar{z}_A a_p^* a_{p-k} a_k + h.c. \big) ,
\end{align*}
and similar for $B$. The $AB$ part, on the other hand, is
\begin{align*}
Z_0^{AB} &:=  \rho_{z_A} \rho_{z_B} |\Lambda |G_{\omega}^{AB} , \\
Z_2^{AB} &:= \widehat{g\omega}_{AB}(0)\sum_{p \in {\Lambda}^*_+}  \big(\rho_{z_A} b_p^* b_p + \rho_{z_B} a_p^* a_p  \big)\\
&\quad + \frac{1}{|\Lambda|} \sum_{p\in {\Lambda}^*_+}  (\widehat{g}_{AB}(p)+ \widehat{g\omega}_{AB}(p)) \big( \bar{z}_A z_B a^*_p b_p + h.c. \big) \\
&\quad +  \frac{1}{|\Lambda|} \sum_{p\in {\Lambda}^*_+}  \widehat{g}_{AB}(p) \big( \bar{z}_A \bar{z}_B a_p b_p + h.c. \big),	\\
Z_3^{AB,L} &:=  \frac{1}{|\Lambda|} \sum_{k \in {\Lambda}^*_+, p \in \mathcal P_L^{\mathbb Z} }\sigma (p,k) \widehat{g}_{AB}(k) \big( \bar{z}_A  b^*_{p} b_{p-k} a_k   + \bar{z}_B a^*_p a_{p-k} b_k + h.c. \big), 
\end{align*}
and for the gap, $\eta$ and conv parts we have
\begin{align}
\mathcal{G}_{\text{gap}}(z) &:= \frac{\pi n_+}{8 \ell^4} + \frac{K_H n_+^H}{2 \ell^2}+ \frac{\pi n_+^L n_+}{4\mathcal{M}\ell^2} + \frac{K_H n_+^L n_+^H}{2 \mathcal{M}\ell^2} \nonumber \\
&\quad +  \frac{\rho \bar{a}}{4(n+m)^M} (|z|^{2M} + n_+ |z|^{2M-2} + n_+^2 |z|^{2M-4}),\\
\mathcal{G}_{\text{conv}} (z)&:= \frac{8 \pi \bar{a}}{|\Lambda|}(\rho \bar{a}^3)^{1/4}(|z_A|^4 + |z_B|^4-n^2-m^2), \label{eq:Gconv(z)}\\
\mathcal{G}_{\eta}(z) &:=  \frac{64\pi \bar{a}}{|\Lambda|} \Big((|z_A|^2+|z_B|^2)n_+ + K_{\ell}^{-1}|\Lambda|^{-1/2} \big(|z_A|^2+|z_B|^2\big)^{5/2}\Big) \one_{\eta=0},
\end{align}
with $|z|^2 = |z_A|^2 + |z_B|^2$ the norm in $\mathbb{C}^2$.
\end{lemma}

\begin{proof}
The proof is similar to the one in \cite[Theorem 2.9]{freeEnCPHM}. First, we observe that
\begin{align*}
\mathcal{Z}_{0g}^A + \mathcal{Z}_{0g}^B +\mathcal{Z}_{0g}^{AB} =	\frac{4\pi}{|\Lambda|}& \big(n^2 a_A+ 2 n m a_{AB} + m^2 a_B \\
& -n_+^A(n_+^A - 1)a_A
- n_+^B(n_+^B-1)a_B- n_+^A n_+^Ba_{AB}\big).  
\end{align*}
In this way, we obtained the first term on the r.h.s. of \eqref{eq:cnumberformula}, plus a term which, since $n_+ = n_+^L + n_+^H$, 
\begin{equation}
\frac{4\pi}{|\Lambda|}(n_+^A(n_+^A - 1) + n_+^B(n_+^B-1)+ n_+^A n_+^B ) \leq C K_{\ell}^2\frac{n^H_+}{\ell^2} + C \frac{n_+^L}{\mathcal{M}} \frac{n_+}{\ell^2}\frac{\mathcal{M}}{\ell} \leq \varepsilon \mathcal{G},
\end{equation}
where we used that $\ell = K_{\ell}(\rho \bar{a})^{-1/2}, K_{\ell}^4\leq C K_H, \mathcal{M} < C^{-1} \ell \bar{a}^{-1}$ to reabsorb the error in a small fraction of the spectral gap.
We also use Lemma \ref{lem:reconstructiongomega} to substitute the $\frac{n_0^2}{|\Lambda|}\widehat{g\omega}_{\#}(0)$ with the relative $\frac{n_0^2}{|\Lambda|}G_{\omega}^{\#}$. This gives errors which can be bounded by
\begin{equation}
\frac{(n+m)^2}{|\Lambda|}\frac{\bar{a}^2}{\ell} \leq C K_{\ell}^{-1}(\rho \bar{a})^{5/2}\ell^3,
\end{equation}
coherent with the error in the statement of the lemma. We also observe that, in the term $\mathcal{G}_{\eta}$, the part which reads, for $\eta=0$,
\begin{equation}
\frac{64 \pi \bar{a}}{|\Lambda|^{3/2}}K_{\ell}^{-1} \big(n+m\big)^{5/2}\leq K_{\ell}^{-1} (\rho \bar{a})^{5/2} |\Lambda|,
\end{equation}
which can, as well, be reabsorbed in the error term.
For the remaining part of $\mathcal{G}_{\eta}$, we observe that, using that $K_{\ell} = \frac{1}{1000C}$ when $\eta=0$,
\begin{equation}
- \frac{64 \pi\bar{a}}{|\Lambda|}(n+m)n_+ \one_{\eta=0} \geq  -C \rho \bar{a} n_+ \one_{\eta =0} \geq -\frac{1}{1000} \frac{n_+}{\ell^2},
\end{equation}
which can be reabsorbed in a small fraction of the spectral gap.

Now, we insert a term of the form $\rho\bar{a}\frac{\mathcal{N}^M}{N^M}$, which is $\rho \bar{a}$ in the $(n,m)$-sector of the Fock space, and therefore an error term in the original Hamiltonian.

By \eqref{coherent:expansion} and since $|z_A\rangle, |z_B\rangle$ are eigenvectors of $a_0,b_0$, respectively, we are allowed to replace the operators with the symbols
\begin{equation}
a_0 \mapsto z_A,  \quad 
 a_0^* \mapsto \bar{z}_A, \quad  b_0 \mapsto z_B,\quad b^*_0 \mapsto \bar{z}_B, 
\end{equation}
and the following polynomials, thanks also to the commutation rules,
\begin{align*}
&n_0^A =a^*_0 a_0 \mapsto |z_A|^2 - 1, \qquad &n_0^B = b^*_0 b_0 \mapsto |z_B|^2 - 1,\\
&a^*_0 a^*_0 a_0 a_0 \mapsto |z_A|^4 - 4 |z_A|^2 + 2, \qquad &b^*_0 b^*_0 b_0 b_0 \mapsto |z_B|^4 - 4 |z_B|^2 + 2,\\
&a^*_0 a_0 b^*_0 b_0 \mapsto |z_A|^2|z_B|^2 - |z_A|^2 - |z_B|^2 +1.
\end{align*}

We consider the following lower bound
\begin{equation}\label{eq:firstboundNmnm}
\rho \bar{a} \frac{\mathcal{N}^M}{{(n+m)}^M} \geq \rho \bar{a} \frac{n_0^M + n_0^{M-1} n_+ + n_0^{M-2}n_+^2}{{(n+m)}^M},
\end{equation}
Using the aforementioned substitution rules, the CCR and the fact that $a^{\#}_0$ and $b_0^{\#}$ commute, we have that, for $h \in \mathbb{N}$,
\begin{equation*}
n_0^h = (a_0^*a_0 + b^*_0b_0)^h \quad \mapsto \quad p_h(z) = |z|^{2h} + \text{smaller order terms},
\end{equation*}
recalling that $|z|^2 = |z_A|^2 + |z_B|^2$, where the smaller order terms have constant coefficients (bounded by $n+m$) and can be explicitly calculated. The c-number substitution of the r.h.s. of \eqref{eq:firstboundNmnm} is, therefore, 
\begin{equation}\label{cnumber:Ngap}
\frac{\rho \bar{a}}{(n+m)^M} (p_M(z) + p_{M-1}(z) n_+ + p_{M-2}(z) n_+^2).
\end{equation}
We observe that, for $|z|^2> 2 (n+m)$, 
\begin{align*}
|z|^{2h} \geq 2(n+m)\sum_{k=0}^{h-1} \binom{h-1}{k} |z_A|^{2k} |z_B|^{2(h-1-k)},
\end{align*}
which tells us that a fraction of $|z|^{2h}$ can bound all the smaller order terms in $p_h(z)$, giving the bound
\begin{equation}
p_h(z) \geq \frac{1}{2}|z|^{2h}.
\end{equation}
By inserting it into \eqref{eq:firstboundNmnm}, we get
\begin{equation}\label{eq:emergenceGapTerm_cnumber}
\rho \bar{a}\frac{\mathcal{N}^M}{{(n+m)}^M} \geq \frac{\rho \bar{a}}{2{(n+m)}^M} (|z|^{2M} + |z|^{2M-2}n_+ + |z|^{2M-4} n_+^2) .
\end{equation}
On the other hand, for the $z$'s such that $|z|\leq 2 (n+m)$, substituting \eqref{cnumber:Ngap} with the r.h.s. of \eqref{eq:emergenceGapTerm_cnumber} we get an error of order $C\rho \bar{a}$, which can be reabsorbed in the error term in \eqref{eq:cnumberformula}.

The rest of the terms are transformed following the substitution rules and creating errors which can be either reabsorbed in the gap term \eqref{eq:emergenceGapTerm_cnumber} or are of order $CK_{\ell}\rho \bar{a}$ thanks to the assumptions on the parameters (see \cite{freeEnCPHM} for further details).
\end{proof}

We now introduce the operator that we will call Bogoliubov Hamiltonian for two species
\begin{align*}
\mathcal{K} &:= \sum_{p \in {\Lambda}^*_+} \Big( \tau(p) a_p^* a_p + \rho_{z_A} \widehat{g}_A(p)  a_p^* a_p  + \frac{ \widehat{g}_A(p) }{2 |\Lambda |}\big( \bar{z}_A^2 a_p a_p + h.c. \big) \\
&\quad + \sum_{p \in {\Lambda}^*_+} \Big( \tau(p) b_p^* b_p + \rho_{z_B} \widehat{g}_B(p)  b_p^* b_p  + \frac{ \widehat{g}_B(p) }{2 |\Lambda |}\big( \bar{z}_B^2 b_p b_p + h.c. \big)\\
&\quad +	\sum_{p \in {\Lambda}^*_+} \widehat{g}_{AB}(p) \Big(\frac{\bar{z}_A z_B}{|\Lambda|} a^*_p b_p + \frac{\bar{z}_B z_A}{|\Lambda|} b^*_p a_p +\frac{1}{|\Lambda|} \big( \bar{z}_A \bar{z}_B a_p b_p + h.c. \big) \Big).
\end{align*}

We can rewrite this Hamiltonian using the following matrices and vectors
\begin{equation}
\mathcal{A}(p) = \tau(p) \one_2 + \mathcal{B}(p), \qquad \mathcal{B}(p) = \begin{pmatrix}
\rho_{z_A} \widehat{g}_A(p) & \sqrt{\rho_{z_A}\rho_{z_B}} \,\widehat{g}_{AB}(p)\\
\sqrt{\rho_{z_A}\rho_{z_B}}\, \widehat{g}_{AB}(p) & \rho_{z_B} \widehat{g}_B(p)
\end{pmatrix},
\end{equation}
\begin{equation}
c_p = \left(\begin{array}{c}
\frac{\bar{z}_A}{|z_A|}a_p\\
\frac{\bar{z}_B}{|z_B|}b_p
\end{array}\right),
\end{equation}
so that
\begin{equation}
\mathcal{K}= \sum_{p \in {\Lambda}^*_+}\Big(  c^*_p \cdot \mathcal{A}(p) c_p +\frac{1}{2} \big( c_p \cdot \mathcal{B}(p) c_p + c^*_p \cdot \mathcal{B}(p) c^*_p \big) \Big).
\end{equation}
We further introduce the following operator, sum of the residual quadratic terms
\begin{align}
Z_{2}^{ex}(z) =\sum_{p \in \Lambda^*}  c^*_p \cdot G_{\omega}(p) c_p
\end{align}
with 
\begin{equation}
G_{\omega} = \begin{pmatrix} 
 \rho_{z_A}(\widehat{g\omega}_A(p)+\widehat{g\omega}_A(0))  +\rho_{z_B}\widehat{g\omega}_{AB}(0) & \sqrt{\rho_{z_A}\rho_{z_B}} \widehat{g\omega}_{AB}(p) \\
  \sqrt{\rho_{z_A}\rho_{z_B}}\widehat{g\omega}_{AB}(p) &  \rho_{z_B}(\widehat{g\omega}_B(p) +\widehat{g\omega}_B(0) ) +\rho_{z_A}\widehat{g\omega}_{AB}(0)
\end{pmatrix},
\end{equation}
the sum of the zero terms
\begin{equation}\label{def:Zgomegaterms}
Z_0(z):= Z_0^A+ Z_0^B+Z_0^{AB} = \frac{|\Lambda|}{2}(\rho_{z_A}^2 G_{\omega}^A + 2\rho_{z_A}\rho_{z_B} G_{\omega}^{AB}+\rho_{z_B}^2 G_{\omega}^B )
\end{equation}
 and the sum of cubic terms:
\begin{equation}
Z_{3,L}(z) = Z_{3,L}^A(z) + Z_{3,L}^B(z) + Z_{3,L}^{AB}(z)=\frac{1}{|\Lambda|}\sum_{\substack{k \in \Lambda^*\\ p \in \mathcal{P}_L^{\mathbb{Z}}}} \sigma(p,k) w_{p,k} \cdot   F_k c_k+ h.c.,
\end{equation}
where 
\begin{equation}
 w_{p,k} =  \begin{pmatrix} 
 a^*_p a_{p-k}\\
 b^*_p b_{p-k}
\end{pmatrix}, \qquad F_k = \begin{pmatrix}
|z_A| \widehat{g}_A(k) & |z_B| \widehat{g}_{AB} (k)\\
|z_A| \widehat{g}_{AB} (k)& |z_B| \widehat{g}_B(k)
\end{pmatrix}.
\end{equation} 
In this way, we can rewrite the Hamiltonian 
\begin{align}
\mathcal{L}_{\mu_A,\mu_B}  &=\mathcal{R}(z)+\mathcal{G}_{\text{gap}}(z)+\mathcal{G}_{\text{conv}}(z)+ \mathcal{G}_{\eta}(z)-\mu_A |z_A|^2 -\mu_B |z_B|^2, \nonumber\\
\mathcal{R}(z) &:= \mathcal{K} +Z_0(z)+ Z^{ex}_2(z) + Z_{3,L}(z). \label{eq:newHamiltcnumber}
\end{align}

In the next sections we split the analysis in the two cases when 
$\rho_z := \rho_{z_A} + \rho_{z_B}$ is bigger or smaller than $K_{z}\rho$, where
\begin{equation}
K_z = (\rho \bar{a}^3)^{-\nu} \gg 1,
\end{equation}
for a certain $\nu > 0$. While in the first case the bounds are easier because there is an excess of energy in this ``less favourable" configuration, in the second case more precise estimates are needed.

\section{Case $\rho_z$ far from $\rho$}\label{sec:rhofar}

Let us consider the case $|z|^2 \geq K_{z} (n+m)$: we show how, in this regime, the spectral gap is large enough to absorb several terms which are, on the other hand, relevant in the other regime. 

\begin{proposition}\label{prop:z>N}
If $|z|^2 \geq K_{z} (n+m)$ and $M\geq 3 +\nu^{-1}(12\eta +1)$, $0 <\mu_A,\mu_B \leq C \ell^{-2}$, $K_H^3 \ll K_{\ell}^{-5}(\rho \bar{a}^3)^{\frac{1}{2} -\nu(M-2)} $, then
\begin{equation}
\mathcal{R}(z) +  \mathcal{G}_{\text{gap}}(z)-\mu_A |z_A|^2 - \mu_B |z_B|^2 \geq \frac{1}{2}\rho^2 \ell^3 (\rho \bar{a}^3)^{\frac{1}{2}+3 \eta -\nu M}.
\end{equation}
\end{proposition}

\begin{proof}
We first observe that
\begin{equation}
\sum_{p \in \Lambda^*} c^*_p \cdot\big(G_{\omega}(p) + \mathcal{B}(p) \big) c_p \leq C \frac{|z|^2}{|\Lambda|}\bar{a} n_+\leq  C \frac{|z|^2}{(n+m)} \rho \bar{a} n_+.
\end{equation}
Since $|z|^2 \geq K_{z} (n+m)$, this last term can be absorbed by a fraction of the spectral gap $\mathcal{G}_{\text{gap}}$.
By using this, that $\mathcal{A}_p = \tau_p \one_2 + \mathcal{B}_p$ and dropping some non-negative terms in \eqref{eq:newHamiltcnumber} we are left with 
\begin{equation}
\mathcal{R}(z)+ \frac{1}{2}\mathcal{G}_{\text{gap}}(z)\geq  Z_0(z) + \sum_{p \in \Lambda^*} \Big(\tau_p c^*_p  c_p + \frac{1}{2}\big( c_p \cdot \mathcal{B}_p c_p + c^*_p\cdot \mathcal{B}_p c_p\big)\Big)  + Z_{3,L}(z).
\end{equation}
By a Cauchy-Schwarz inequality, for $\varepsilon>0$, we have 
\begin{align}
 Z_{3,L}(z) &= \frac{1}{|\Lambda|}\sum_{\substack{k \in \Lambda^*\\ p \in \mathcal{P}_L^{\mathbb{Z}}}} \sigma(p,k) w_{p,k} \cdot   F_k c_k+ h.c. \nonumber \\
 &\leq C \varepsilon |\mathcal{P}_L| \sum_{k \in \Lambda^*} k^2 c^*_k c_k + \frac{\varepsilon^{-1}}{|\Lambda|^2} \sum_{\substack{k \in \Lambda^*\\p \in \mathcal{P}_L^{\mathbb{Z}}}} \frac{1}{k^2} w_{p,k} \cdot F_k F_k^*  w_{p,k}^*.\label{eq:estimatewithKinlargez}
 \end{align}
 Let us analyse in more detail the second term on the r.h.s.:
 \begin{align*}
 &\sum_{\substack{k \in \Lambda^*\\p \in \mathcal{P}_L^{\mathbb{Z}}}} \frac{1}{k^2} w_{p,k} \cdot F_k F_k^*  w_{p,k}^*  = S_{\text{diag}} + S_{\text{off}},\\
  &S_{\text{diag}}= \sum_{\substack{k \in \Lambda^*\\p \in \mathcal{P}_L^{\mathbb{Z}}}} \frac{1}{k^2}\Big(\big(  |z_A|^2 \widehat{g}_A^2(k)  + |z_B|^2 \widehat{g}_{AB}^2(k) \big) a^*_p a_{p-k} a^*_{p-k} a_p \\
  &\qquad \qquad \qquad + \big(  |z_B|^2 \widehat{g}_B^2(k)  + |z_A|^2 \widehat{g}_{AB}^2(k) \big) b^*_p b_{p-k} b_{p-k}^* b_p\Big),\\
  &S_{\text{off}} = \sum_{\substack{k \in \Lambda^*\\p \in \mathcal{P}_L^{\mathbb{Z}}}} \frac{\widehat{g}_{AB}(k)}{k^2}  \big(|z_A|^2\widehat{g}_A(k) + |z_B|^2\widehat{g}_B (k)\big)\big(  b^*_p b_{p-k} a^*_{p-k}a_p + a^*_p a_{p-k} b_{p-k}^*b_p\big).
 \end{align*}
Starting from the off-diagonal term, we see that by a Cauchy-Schwarz inequality and using that $|k| \geq C\ell^{-1}$, we can bound
\begin{equation}\label{eq:intermediatestep1Sdiag}
S_{\text{off}} \leq C\ell^{2} |z|^2  \bar{a}^2 \sum_{\substack{k \in \Lambda^*\\p \in \mathcal{P}_L^{\mathbb{Z}}}} \big(b^*_p a^*_{p-k}a_{p-k} b_p + a^*_p b^*_{p-k} b_{p-k} a_p\big)\leq C \ell^2 |z|^2\bar{a}^2 n_+^2.
\end{equation} 
 For the diagonal part we use the CCR, again that $|k| \geq C\ell^{-1}$, the reconstruction of $\widehat{g\omega}(0)$'s by Lemma \ref{lem:reconstructiongomega}
and that $|\widehat{g\omega}_{\#}(0)|\leq C \bar{a}$, to write 
\begin{align}
S_{\text{diag}} &\leq \sum_{k \in \Lambda^*} \frac{|z|^2}{k^2}\big(\big(   \widehat{g}_A^2(k)  +  \widehat{g}_{AB}^2 (k)\big) n_+^A+ \big( \widehat{g}_B^2 (k) + \widehat{g}_{AB}^2 (k)\big) n_+^B\big)  + C|z|^2\ell^2 \bar{a}^2 n_+^2 \nonumber\\
&\leq  C |z|^2 \ell^3\bar{a}\big(  n_+ + \bar{a}\ell^{-1}  n_+^2\big).\label{eq:intermediatestep2Sdiag}
\end{align}
Inserting \eqref{eq:intermediatestep1Sdiag} and \eqref{eq:intermediatestep2Sdiag} back in \eqref{eq:estimatewithKinlargez} and choosing $\varepsilon = \frac{1}{2}C^{-1} |\mathcal{P}_L|^{-1}$, with $|\mathcal{P}_L| =K_H^3$, we get
\begin{equation}
Z_{3,L} \leq \frac{1}{2} \sum_{k \in \Lambda^*} \tau_k c^*_k c_k +  \frac{K_H^3}{|\Lambda|} |z|^2\bar{a} (n_+ + \bar{az}\ell^{-1} n_+^2).
\end{equation}
Using the assumptions on $M$ and since $K_H^3 \ll (\rho \bar{a}^3)^{\frac{1}{2} -\nu(M-2)}K_{\ell}^{-5}$, we can absorb the last term on the r.h.s. in a fraction of $\mathcal{G}_{\text{gap}}$. 
We now use that we can replace $\tau_p$ by $p^2$ by absorbing the negative spectral gap terms, and introducing a new creation and annihilation operator
\begin{equation}
d_p := c_p + \frac{1}{p^2}\mathcal{B}_p c^*_p,
\end{equation}
we, get by completing the square,
\begin{equation}
\frac{1}{2}\sum_{p \in \Lambda^*}\big(  p^2 c^*_p c_p + c_p \cdot \mathcal{B}_p c_p + c^*_p \cdot \mathcal{B}_p c_p^*\big) = \frac{1}{2}\sum_{p \in \Lambda^*} \big(p^2 d^*_p d_p -  \frac{1}{p^2} c_p \cdot \mathcal{B}^*_p \mathcal{B}_p c_p^*\big).
\end{equation}
Similarly to above we can use the CCR and the reconstruction of the $\widehat{g\omega}_{\#}(0)$ to obtain
\begin{equation}
\sum_{p \in \Lambda^*} \frac{1}{p^2} c_p \cdot\mathcal{B}^*_p  \mathcal{B}_p c_p^* \leq \frac{|z|^4}{|\Lambda|} \bar{a} (n_+ +1),
\end{equation}
and reabsorb the last term in a fraction of the spectral gap because $M > 3+\nu^{-1}(12\eta+1)$, while we bound by zero the positive term $d^*_p d_p \geq 0$. Finally, since $\mu_A,\mu_B \leq C\ell^{-2}$,
\begin{equation}
-\mu_A|z_A|^2-\mu_B |z_B|^2 \geq -\frac{C}{\ell^2(n+m)} \frac{|z|^4}{(n+m)},
\end{equation}
which can be reabsorbed in a small fraction of the spectral gap, provided that $M \geq 2+\nu^{-1}(2\eta + 1/2)$, which is satisfied by our assumptions.
\end{proof}

\section{Case $\rho_z$ close to $\rho$}\label{sec:rhoclose}

In the Appendix \ref{app:bogint}, Lemma \ref{lem:diagBog} lets us diagonalize the Bogoliubov Hamiltonian by introducing the new creation and annihilation operators
\begin{equation}
d_k = c_k + \beta_k \cdot c^*_k, \qquad d_k^* = c_k^* + \beta_k \cdot c_k,
\end{equation}
giving us
\begin{equation}\label{eq:KKdiag+S}
\mathcal{K} =  \mathcal{K}^{\text{diag}}+ \mathcal{S}, \qquad \mathcal{K}^{\text{diag}}:=\sum_{k \in \Lambda^*}d^*_k \cdot \mathcal{D}_k d_k \geq 0,
\end{equation}
where $\mathcal{D}_k = \mathcal{D}_k(\rho_{z_A},\rho_{z_B})$ and $\beta_k= \beta_k (\rho_{z_A},\rho_{z_B})$ are defined in \eqref{eq:Ddiagexpressions}, \eqref{eq:Bdiagexpressions},
\begin{align}
\mathcal{S} = \sum_{k \in\Lambda^*} \Big( \frac{1}{2} \Big(\sqrt{\tau_k^2 +2\lambda_+ (k)\tau_k} +\sqrt{\tau_k^2  +2\lambda_-(k) \tau_k}  \Big)- \tau_k - \frac{1}{2}(\lambda_+(k) + \lambda_-(k))\Big),
\end{align}
and $\lambda_{\pm} = \lambda_{\pm}(\rho_{z_A},\rho_{z_B})$ in \eqref{def:lambdapm} are the eigenvalues of $\mathcal{B}_k$.
The sum of this last term with the $Z_0$ defined in \eqref{def:Zgomegaterms} gives the Lee-Huang-Yang-type term correction. Indeed, by Lemma \ref{lem:calcBogInt} we have 
\begin{align}\label{eq:SZ0sum}
\mathcal{S}+ Z_0(z) &\geq (8\pi)^{5/2}|\Lambda|(\rho_{z_A}^2 a_A^2 + 2 \rho_{z_A} \rho_{z_B} a_{AB}^2 + \rho_{z_B}^2 a_B^2)^{5/4}  \frac{2\sqrt{2}}{15\pi^2} (\mu_+^{5/2} + \mu_-^{5/2}) - \mathcal{E},\\
|\mathcal{E}|&\leq  C (\rho \bar{a})^{5/2}\ell^3(R^2 K_z^{7/2} K_{\ell}^3 \rho \bar{a}  + K_{z}^{5/2}K_{\ell}^{-1})\one_{\eta \neq 0} + \frac{1}{2}\mathcal{G}_{\eta},
\end{align}
where $\mu_{\pm}= \mu_{\pm}(\rho_{z_A},\rho_{z_B})$ are defined in \eqref{appdef:mupm}.

As anticipated in the introduction, the correction to renormalize the Bogoliubov integral produces the $\mathcal{Q}_2^{\mathrm{ex}}$ terms. We show in this section how the contribution of the soft pairs from the $\mathcal{Q}_3^{\mathrm{ren}}$ terms, together with a small fraction of the positive, diagonal operator from the diagonalization of the Bogoliubov Hamiltonian, delete the errors inherited from the $\mathcal{Q}_2^{\mathrm{ex}}$ and so renormalize the Bogoliubov functional.
In the following lemma we extract the contribution of the soft pairs from the $Z_3^L$-terms.

\begin{lemma}\label{lem:softpairs}
For all $\varepsilon>0$ there exists a $C>0$ such that, if $\rho \bar{a}^3 \leq C^{-1}$, then for $|z|^2 \leq K_z (n+m)\leq C_R^2 K_z\rho \ell^3$ and $\mathcal{M} \leq  C \rho \ell^3 K_H^{-3}K_{\ell}^{-4} K_z^{-3}$, we have
\begin{equation}
Z_3^{\xi,L} -Z_3^{\xi,\text{soft}} \geq \varepsilon \mathcal{G}_{\text{gap}}, \qquad \xi \in \{A,B,AB\},
\end{equation}
where 
\begin{align*}
Z_3^{A,\text{soft}} &= \frac{1}{|\Lambda|} \sum_{k \in \mathcal{P}_H, p \in \mathcal P_L^{\mathbb Z} }\sigma(p,k) \widehat{g}_A(k) \big( \bar{z}_A a_p^* a_{p-k} a_k + h.c. \big),\\
Z_3^{B,\text{soft}} &= \frac{1}{|\Lambda|} \sum_{k \in \mathcal{P}_H, p \in \mathcal P_L^{\mathbb Z} }	\sigma(p,k) \widehat{g}_B(k) \big( \bar{z}_B b_p^* b_{p-k} b_k + h.c. \big),\\
Z_3^{AB,\text{soft}} &= \frac{1}{|\Lambda|} \sum_{k \in \mathcal{P}_H, p \in \mathcal P_L^{\mathbb Z} }\sigma(p,k) \widehat{g}_{AB}(k) \big( \bar{z}_A b^*_{p} b_{p-k} a_k   + \bar{z}_B a^*_p a_{p-k} b_k + h.c. \big).
\end{align*}
\end{lemma}

\begin{proof}
The proof follows the lines of \cite[Lemma 7.1]{freeEnCPHM} applied to each $Z_3$ separately, and we write here only the slightly different part, $Z_3^{AB}$. We apply a Cauchy-Schwarz inequality, for $\varepsilon >0$, and using that $|z|^2 \leq K_z \rho |\Lambda|$,
\begin{align*}
Z_3^{AB,L} -Z_3^{AB,\text{soft}} &=  \frac{1}{|\Lambda|} \sum_{k \in \mathcal{P}_L, p \in \mathcal P_L^{\mathbb Z} }\sigma(p,k) \widehat{g}_{AB}(k) \big(\bar{z}_A b^*_{p} b_{p-k} a_k   + \bar{z}_B a^*_p a_{p-k} b_k + h.c. \big)\\
&\geq - C \frac{\widehat{g}_{AB}(0)}{|\Lambda|}  \sum_{k \in \mathcal{P}_H, p \in \mathcal P_L^{\mathbb Z} } \big( \varepsilon (|z_A|^2  b^*_p b_p  + |z_B|^2 a^*_p a_p)\\
&\qquad +  \varepsilon^{-1} \big(b^*_{p-k} a^*_k a_k b_{p-k} + a^*_{p-k} b^*_k b_k a_{p-k}\big) \big)\\
&\geq - C  \varepsilon K_z \rho \bar{a} K_H^3 n_+- C \varepsilon^{-1}\rho\bar{a}\frac{\mathcal{M}}{\rho\ell^3}  \frac{n_+^L}{\mathcal{M}}n_+,
\end{align*}
where from the first to the second line, we bounded $c(p,k)$ by a constant and we used the canonial commutation relations, and to obtain the last line we used that $\sum_{p \in \mathcal{P}_L}1 = K^3_H$. We focus the attention on the last line: for the first term, we use that $\rho_{z_A} + \rho_{z_B} \leq C K_{z} \rho$, and choosing $\delta =\frac{1}{3}C^{-1} \varepsilon K_z^{-1} K_{\ell}^{-2} K_H^{-3}$, it can be absorbed in a fraction of the spectral gap $\mathcal{G}_{\text{gap}}$. For the second term, we use that $\mathcal{M} \leq  C \rho \ell^3 K_H^{-3}K_{\ell}^{-4} K_z^{-1}$ and thanks to the condition we assumed, it can be absorbed in $\mathcal{G}_{\text{gap}}$ as well.

\end{proof}

We introduce the following operators, regrouping the \textit{soft} pairs, cubic terms:
\begin{equation}
Z_{3}^{\text{soft}}(z) = Z_{3}^{A,\text{soft}}(z) + Z_{3}^{B,\text{soft}}(z) + Z_{3}^{AB,\text{soft}}(z)=\frac{1}{|\Lambda|}\sum_{\substack{k \in \mathcal{P}_H\\ p \in \mathcal{P}_L^{\mathbb{Z}}}} \sigma(p,k) w_{p,k} \cdot   F_k c_k+ h.c.
\end{equation}
We only need part of the diagonalized Bogoliubov Hamiltonian to control the errors from the $Z_3^{\text{soft}}$ and $Z_2^{\text{ext}}$ terms. We introduce the following split:
\begin{equation}
\mathcal{K}^{\text{diag}} = \mathcal{K}_H^{\text{diag}} + \mathcal{K}_L^{\text{diag}},
\end{equation}
where
\begin{equation}
 \mathcal{K}_H^{\text{diag}}  := \sum_{ k \in \mathcal{P}_H} d^*_k \cdot \mathcal{D}_k d_k, \qquad \mathcal{K}_L^{\text{diag}}  := \sum_{ k \in \mathcal{P}_L} d^*_k \cdot \mathcal{D}_k d_k.
\end{equation}

\begin{proposition}\label{prop:killQ3}
There exists a $C>0$ such that, if $\rho \bar{a}^3 \leq C^{-1}$, then for $|z|^2 \leq K_z (n+m)\leq C_a^2 K_z\rho \ell^3$, if we assume that $\mathcal{M}\leq \rho \ell^3 K_H^{-3}K_{\ell}^{-4}K_z^{-3}$, $K_{\ell}K_z^3 K_H^2 \ll (\rho \bar{a}^3)^{-1/2}$, $K_{\ell}^4 K_z^3\leq C K_H$, and furthermore
\begin{equation}\label{assmpt:deltaKz}
K_{\ell}^2 K_z \delta_{AB} \bar{a}^{-1} \leq \frac{1}{1000 C}, \qquad \text{ for } \eta \neq 0,
\end{equation}
we have the following inequality:
\begin{equation*}
\mathcal{K}_H^{\mathrm{diag}} + Z_2^{\mathrm{ext}} + Z_3^{\mathrm{soft}} + \varepsilon \mathcal{G}_{\text{gap}} + \frac{1}{2}\mathcal{G}_{\eta}\geq  -C(\rho \bar{a})^{5/2}(\rho \bar{a}^3)^{{\eta}}\ell^3. 
\end{equation*}
\end{proposition}

\begin{proof}

We change variables in order to compare $Z_3^{\mathrm{soft}}$ with $\mathcal{K}^{\mathrm{diag}} = \sum d^*_k \cdot \mathcal{D}_k d_k$: 
\begin{equation}
c_k = (1-\beta_k^2)^{-1} (d_k -\beta_k d^*_k), 
\end{equation} 
and we can consider the following splitting
\begin{align*}
Z_3^{\mathrm{soft}} &= \mathcal{T}_1  - \mathcal{T}_{\beta}, \\
\mathcal{T}_1 &= \frac{1}{|\Lambda|}\sum_{k \in \mathcal{P}_H, p \in \mathcal{P}_L^\mathbb{Z}} \sigma(p,k) w_{p,k} \cdot F_k (1-\beta_k^2)^{-1} d_k + h.c.,\\
\mathcal{T}_{\beta} &= \frac{1}{|\Lambda|}\sum_{k \in \mathcal{P}_H, p \in \mathcal{P}_L^\mathbb{Z}} \sigma(p,k) w_{p,k} \cdot  F_k  (1-\beta_k^2)^{-1} \beta_k d_k^*+ h.c.
\end{align*} 
We prove the following inequalities, which combined give the proof of the lemma:
\begin{align}
\mathcal{T}_1 + (1-K_H^{-1}) \mathcal{K}^{\mathrm{diag}}_H + Z_2^{\mathrm{ext}} +\varepsilon \mathcal{G}_{\text{gap}} &\geq  - C (K_{\ell}^{2}K_z \delta_{AB}\bar{a}^{-1}) \frac{n_+}{\ell^2},\label{eq:estimateT1}\\
\mathcal{T}_{\beta} + K_H^{-1} \mathcal{K}_H^{\text{diag}} + \varepsilon \mathcal{G}_{\text{gap}}&\geq -C (\rho \bar{a})^{5/2}(\rho \bar{a}^3)^{{\eta}}\ell^3. \label{eq:estimateTalpha}
\end{align}
Let us start from proving \eqref{eq:estimateTalpha}. 
For any $\varepsilon >0$, by a Cauchy-Schwarz inequality we have
\begin{equation}\label{eq:Talphacalculations}
\mathcal{T}_{\beta} \leq C\varepsilon^{-1} \sum_{k \in \mathcal{P}_H, p \in \mathcal{P}_L^\mathbb{Z}} w_{p,k} \cdot w_{p,k}^* + \frac{C\varepsilon K_H^3 }{|\Lambda|^2}\sum_{k \in \mathcal{P}_H} |F_k  (1-\beta_k^2)^{-1} \beta_k d_k^*|^2
\end{equation}
where we estimated the coefficients $\sigma(p,k)$ by a constant and the sum in $p$ in the second term by $K_H^3$. We now estimate the operator norm of the matrices by a constant times the maximum of the components of the matrices, obtaining $\|F_k\|^2 \leq C |z|^2\bar{a}^2 \leq C K_{z} \rho \ell^{3} \bar{a}^2$ and, for $k \in \mathcal{P}_H$, by \eqref{eq:Bdiagexpressions},
\begin{equation}
\|(1-\beta_k^2)^{-1} \beta_k\|^2 \leq C \max\{ \rho_{z_A},\rho_{z_B}\}^2 \bar{a}^2 k^{-4} \leq C K_{z}^2 \rho^2 \bar{a}^2 k^{-4} .
\end{equation}
We also observe that the first term on the r.h.s. of \eqref{eq:Talphacalculations} is a sum of objects of the form $\eta_p^* \eta_{p-k} \xi_{p-k}^*\xi_p$, for $\eta,\xi\in \{a,b\}$, and by the commutation relations and summing in $p \in \mathcal{P}_L^{\mathbb{Z}}$ and $k$, they can all be estimated by $Cn_+^L(n_+ + 1)$, giving
\begin{equation}\label{Tbeta:intermediate}
\mathcal{T}_{\beta} \leq C \varepsilon^{-1} n_+^L(n_+ + 1) +  \frac{C}{\ell^3}\varepsilon K_H^3 K_z^3 (\rho \bar{a})^3\bar{a}\sum_{k \in \mathcal{P}_H} \frac{1}{k^4} (d^*_k d_k +1).
\end{equation}
We choose $\varepsilon = \ell^3 K_H^{-2} K_{z}^{-3}K_{\ell}^{-2} \bar{a}^{-1}$ which, \textit{a posteriori}, is going to be a suitable choice to have small errors, and the first term on the r.h.s. of \eqref{Tbeta:intermediate} is bounded by
\begin{equation}
\frac{C}{\ell} K_H^{2} K_z^{3} K_{\ell}^2\bar{a} \Big(\rho \ell^3\frac{\mathcal{M}}{\rho \ell^3} \frac{n_+^L}{\mathcal{M}} \frac{n_+}{\ell^2} + \frac{n_+}{\ell^2} \Big),
\end{equation}
which, since $\mathcal{M}\leq \rho \ell^3 K_H^{-2} K_z^{-3} K_{\ell}^{-4}$ and $K_H^2 K_z^3 K_{\ell} \ll (\rho\bar{a}^3)^{-1/2}$, can be reabsorbed in the spectral gap.

We split the second term in \eqref{Tbeta:intermediate} in two parts corresponding to the addends inside the parentheses. For the one with $d^*_k d_k$, we use that $k^6 \geq K_H^6 \ell^{-6}$ to bound it by
\begin{equation}\label{Tbeta:Ddiag}
CK_H^{-5} K_{\ell}^{4} \sum_{k \in \mathcal{P}_H} k^2 d^*_k d_k \leq C K_H^{-5} K_{\ell}^{4} \sum_{k \in \mathcal{P}_H} d^*_k \cdot \mathcal{D}_k d_k,
\end{equation}
where we used that, for $k \in \mathcal{P}_H$, $\|\mathcal{D}_k\|^{-1} \leq C k^2$. By the assumptions we have that $C K_{H}^{-5} K_{\ell}^{6}  \ll K_H^{-1}$ which makes it possible to reabsorb the term on the r.h.s. in $K_H^{-1} \mathcal{K}^{\text{diag}}_H$.
For the remaining term, we see that the sum converges and it is of order $K_H^{-1}\ell^4$, and thus can be bounded by
\begin{equation}\label{Tbeta:error}
C(\rho \bar{a})^{5/2}\ell^3 (\rho \bar{a}^3)^{\eta}. 
\end{equation} 
By \eqref{Tbeta:intermediate}, \eqref{Tbeta:Ddiag} and \eqref{Tbeta:error}, we get \eqref{eq:estimateTalpha}.

We now prove \eqref{eq:estimateT1}. By a Cauchy-Schwarz inequality we obtain (denoting by $\cdot^T$ the transposed matrix)
\begin{equation}\label{eq:T1boundT2}
\mathcal{T}_1 = \sum_{\substack{k \in \mathcal{P}_H,\\ p \in \mathcal{P}_L^\mathbb{Z}}} \frac{\sigma(p,k)}{|\Lambda|} w_{p,k} \cdot F_k (1-\beta_k^2)^{-1} d_k + h.c.\geq -\mathcal{T}_2   -  (1-K_H^{-1})\sum_{k \in \mathcal{P}_H}  d^*_k \cdot\mathcal{D}_k d_k.
\end{equation}
with 
\begin{equation}
\mathcal{T}_2 = \frac{(1+CK_H^{-2}K_{\ell}^2)}{1-K_H^{-1}}\sum_{\substack{k \in \mathcal{P}_H,\\ p,s \in \mathcal{P}_L^\mathbb{Z}}} \frac{\sigma(p,k) \sigma(s,k)}{|\Lambda|^2 k^2} F_k^T w_{p,k} \cdot F_k^T w^*_{s,k} 
\end{equation}
where we used 
\begin{equation}
(1-\alpha^2_k)^{-1} \mathcal{D}_k^{-1}(1-\alpha^2_k)^{-1} \leq (1 + CK_H^{-2}K_{\ell}^2) k^{-2}\one_2, \qquad \text{for } k \in \mathcal{P}_H,
\end{equation}
obtained by a form bound of expressions \eqref{eq:Ddiagexpressions}, \eqref{eq:Bdiagexpressions}.
The second term on the r.h.s. of \eqref{eq:T1boundT2} is exactly $-(1-K_H^{-1})\mathcal{K}_H^{\mathrm{diag}}$.
Let us observe that 
\begin{equation}
F_k^Tw_{p,k}  =  \begin{pmatrix}
|z_A| \big( \widehat{g}_A(k) a^*_p a_{p-k} + \widehat{g}_{AB}(k) b^*_{p} b_{p-k}\big)\\
|z_B| \big( \widehat{g}_{AB}(k) a^*_p a_{p-k} + \widehat{g}_B(k) b^*_p b_{p-k}\big),
\end{pmatrix}
\end{equation}
which gives
\begin{align}
F_k^T w_{p,k} \cdot \frac{1}{k^2} F_k^T w^*_{s,k} =   \frac{1}{k^2} &\big( (|z_A|^2 \widehat{g}_A^2(k) + |z_B|^2 \widehat{g}_{AB}^2(k)) a^*_p a_{p-k} a^*_{s-k} a_s \label{eq:aaaa} \\
&+(|z_A|^2 \widehat{g}_A(k)  + |z_B|^2 \widehat{g}_B(k))\widehat{g}_{AB}(k) a^*_p a_{p-k} b^*_{s-k} b_s \label{eq:aabb}\\
&+ (|z_A|^2 \widehat{g}_A(k) + |z_B|^2 \widehat{g}_B(k))\widehat{g}_{AB}(k) b^*_p b_{p-k}a^*_{s-k} a_s \label{eq:bbaa}\\
&+(|z_A|^2 \widehat{g}_{AB}^2(k) +|z_B|^2 \widehat{g}_B^2(k)) b_p^* b_{p-k} b^*_{s-k} b_s \big) \label{eq:bbbb}.
\end{align}
We now rearrange the operators to put them in normal ordering, but \eqref{eq:aaaa} and \eqref{eq:bbbb} give some commutator terms 
\begin{equation}
\mathcal{T}_{\text{norm}} := \mathcal{T}_2 - \mathcal{T}_{\text{comm}},
\end{equation}
where, denoting by $C_H = (1+C K_H^{-2}K_{\ell}^2)(1-K_H^{-1})^{-1}$ for readability,
\begin{align*}
\mathcal{T}_{\text{comm}} =C_H \sum_{\substack{k \in \mathcal{P}_H\\ p,s \in \mathcal{P}_L^{\mathbb{Z}}}}\frac{\sigma(p,k) \sigma(s,k)}{k^2|\Lambda|^2} \big( &(|z_A|^2 \widehat{g}_A^2 + |z_B|^2 \widehat{g}_{AB}^2 ) a^*_p [a_{p-k}, a^*_{s-k}] a_s\\
+&(|z_A|^2 \widehat{g}_{AB}^2 +|z_B|^2 \widehat{g}_B^2) b_p^* [b_{p-k} ,b^*_{s-k}] b_s \big).
\end{align*}
We would like to contract indices by using the CCR: $[a_{p},a_{q}^*] = \delta_{p,q} =  [b_{p}, b_{q}]$ and obtain quadratic terms in creation and annihilation operators, but they only hold for $p,q \in \frac{\pi}{\ell} \mathbb{N}_0^3$, which does not hold in this case. Nevertheless, we observe that 
\begin{equation}
a_p = a_{(|p_1|,|p_2|,|p_3|)}, \qquad b_p = b_{(|p_1|,|p_2|,|p_3|)},
\end{equation}
which implies that $[a_{p-k}, a^*_{s-k}] \neq 0$ if and only if $s=p$ or $s \in \mathcal{P}_{p,k}$, where
\begin{equation}
\mathcal{P}_{p,k}:= \{s\in \mathcal{P}_L^{\mathbb{Z}} \setminus \{p\}\,|\, p_j = s_j \text{ or } 2k_j = p_j + s_j \text{ for any } j \in \{1,2,3\}\}.
\end{equation}
For this latter case, since $p,s \in \mathcal{P}_L^{\mathbb{Z}}$, this also implies that $|k_j|\leq K_H \ell^{-1}$. Therefore, we can write 
\begin{equation}
\mathcal{T}_{\text{comm}} = \mathcal{T}_{\text{quad}} + \mathcal{T}_{\text{rest}},
\end{equation}
where,
\begin{align*}
\mathcal{T}_{\text{quad}} := \sum_{\substack{k \in \mathcal{P}_H\\ p \in \mathcal{P}_L^{\mathbb{Z}}}}\frac{C_H|\sigma(p,k)|^2}{k^2|\Lambda|^2} \big( (|z_A|^2 \widehat{g}_A^2 + |z_B|^2 \widehat{g}_{AB}^2 ) a^*_p  a_p+(|z_A|^2 \widehat{g}_{AB}^2 +|z_B|^2 \widehat{g}_B^2) b_p^*  b_p \big),
\end{align*}
\begin{align*}
\mathcal{T}_{\text{rest}} :=\sum_{\substack{k \in \mathcal{P}_H, p \in \mathcal{P}_L^{\mathbb{Z}},\\  s \in \mathcal{P}_{p,k}}}C_H\theta_H(k)\frac{\sigma(p,k) \sigma(s,k) }{k^2|\Lambda|^2} &\big( (|z_A|^2 \widehat{g}_A^2 + |z_B|^2 \widehat{g}_{AB}^2 ) a^*_p  a_{s}\\
&+ (|z_A|^2 \widehat{g}_{AB}^2 +|z_B|^2 \widehat{g}_B^2) b_p^*  b_{s} \big),
\end{align*}
where $\theta_H(k) =1$ if $|k_j| \leq K_H \ell^{-1}$ for every $j=1,2,3$, zero otherwise. Recalling that $K_{\ell}^4 \leq CK_H$, we bound $\mathcal{T}_{\text{rest}}$ by a Cauchy-Schwarz as 
\begin{align*}
\mathcal{T}_{\text{rest}} \leq C \sum_{k \in \mathcal{P}_H} \frac{\one_{|k_1|\leq K_H\ell^{-1}}}{k^2|\Lambda|^2} (|z_A|^2 \widehat{g}_A^2 +  (|z_A|^2+|z_B|^2) \widehat{g}_{AB}^2  +|z_B|^2 \widehat{g}_B^2) n_+.
\end{align*}
By using term by term the estimate from \cite[Lemma 7.3]{freeEnCPHM}, that we recall here below for the reader's convenience,
\begin{equation}
\frac{1}{|\Lambda|}\sum_{k \in \mathcal{P}_H} \frac{\rho \widehat{g}_{\#}^2(k)}{k^2} \one_{|k_1|\leq K_H \ell^{-1}} \leq C \ell^{-2} (\rho \bar{a}^3)^{1/2} K_H^2, \qquad \# \in \{A,B,AB\},
\end{equation}
and $|z|^2 \leq C K_z \rho \ell^3$, we obtain the bound
\begin{equation}\label{eq:controlTrest}
\mathcal{T}_{\text{rest}} \leq    K_z K_H^2 (\rho \bar{a}^3)^{1/2} \frac{n_+}{\ell^2},
\end{equation}
Now, we rewrite the $\mathcal{T}_{\text{quad}}$ term as
\begin{align*}
&\mathcal{T}_{\text{quad}}  = \sum_{\substack{k \in \mathcal{P}_H\\ p \in \mathcal{P}_L^{\mathbb{Z}}}}\frac{C_H|\sigma(p,k)|^2}{k^2|\Lambda|^2} c^*_p\cdot F_{AB}(z) c_p, \\
&F_{AB}(z) = \begin{pmatrix}
 |z_A|^2 \widehat{g}_A^2 + |z_B|^2 \widehat{g}_{AB}^2 & 0\\
0 &|z_A|^2 \widehat{g}_{AB}^2 +|z_B|^2 \widehat{g}_B^2
\end{pmatrix}.
\end{align*}
We now observe that $\sigma(p,k) = \sigma_{p-k} \sigma_p^{-1} \sigma_k^{-1}$ where $\sigma_k := \sigma_{k_1}\sigma_{k_2} \sigma_{k_3}$ (see \eqref{sigmajdef} to recall the definitions). 
We can write
\begin{equation}\label{eq:fromZton+}
\sum_{p \in \mathcal{P}_L^{\mathbb{Z}}} \frac{1}{\sigma_p^2} c^*_p \cdot F_{AB}(z) c_p = \sum_{p \in \mathcal{P}_L}  c^*_p \cdot F_{AB}(z) c_p,
\end{equation}
and, analogously we can reconstruct the sum of $k \in \mathcal{P}_H^{\mathbb{Z}}$ from $k \in \mathcal{P}_H$ at the price of an additional  coefficient $\sigma_k^2$ in the denominator. By these considerations, from \eqref{eq:fromZton+} and bounding $\sigma_{p-k}^2 \leq 8$, we have that
\begin{equation}
\mathcal{T}_{\text{quad}} \leq \sum_{k \in \mathcal{P}_H^{\mathbb{Z}}} \frac{8C_H}{\sigma_k^4  k^2 |\Lambda|^2} \sum_{p \in \mathcal{P}_L} c^*_p\cdot F_{AB}(z) c_p .
\end{equation}
We now distinguish between two cases: $\sigma_k^4= 64 $ or $\sigma_k^4 \neq 64 $. The latter happens when at least one of the components $k_j=0$, and therefore this case can be treated like in the bound \eqref{eq:controlTrest}. We therefore obtain, also using the relations between the parameters,
\begin{equation}\label{calc:Tquad1}
\mathcal{T}_{\text{quad}} \leq (1+CK_H^{-1})\sum_{k \in \mathcal{P}_H^{\mathbb{Z}}} \frac{1}{8  k^2 |\Lambda|^2} \sum_{p \in \mathcal{P}_L} c^*_p \cdot F_{AB}(z) c_p +C K_z K_H^2 (\rho \bar{a}^3)^{1/2} \frac{n_+}{\ell^2}.
\end{equation}
By Lemma \ref{lem:reconstructiongomega}, we can reconstruct a matrix with $g\omega$ entries:
\begin{equation}\label{calc:Tquad2}
\sum_{k \in \mathcal{P}_H^{\mathbb{Z}}}\frac{1+CK_H^{-1}}{8k^2|\Lambda|^2} \sum_{p \in \mathcal{P}_L}c^*_p \cdot F_{AB}(z) c_p  \leq  \sum_{p \in \Lambda^*}c^*_p \cdot F_{\omega}(z) c_p   + C K_{z}K_{\ell}^2(K_H^{-1} + K_H\bar{a}\ell^{-1} ) \frac{n_+}{\ell^2},
\end{equation}
where $F_{\omega} := L_{\omega} + S_{\omega}$, with
\begin{equation}
L_{\omega} := \begin{pmatrix}
 2\rho_{z_A}\widehat{g\omega}_A(0) & 0\\
0 &2\rho_{z_B} \widehat{g\omega}_B(0)
\end{pmatrix}, \qquad S_{\omega} :=2\widehat{g\omega}_{AB}(0)\begin{pmatrix}
\rho_{z_B} & 0\\
 0 & \rho_{z_A} 
 \end{pmatrix}.
\end{equation}
Turning our attention on $Z_2^{\text{ex}}$, we split it as $Z_2^{\text{ext}} = L_{\omega} + T_{\omega} + \mathcal{E}_2$, where $\mathcal{E}_2$ is the error made substituting $Z_2^{\text{ext}}$ with the sum $L_{\omega} + T_{\omega}$, and 
\begin{equation}
T_{\omega} := \widehat{g\omega}_{AB}(0)  \begin{pmatrix}
\rho_{z_B} & \sqrt{\rho_{z_A}\rho_{z_B}} \\
\sqrt{\rho_{z_A}\rho_{z_B}} & \rho_{z_A}
\end{pmatrix}.
\end{equation}
Now, using that for low momenta $p \in \mathcal{P}_L$ we can bound $|\widehat{g\omega}(p)-\widehat{g\omega}(0)| \leq C p^2 R^2 \bar{a} \leq C K_H^2 \ell^{-2}R^2 \bar{a}$, we can estimate
\begin{align}
|\mathcal{E}_2| &\leq \rho K_z\sum_{p \in \Lambda^*} c^*_p \cdot \begin{pmatrix}
|\widehat{g\omega}_A(p)-\widehat{g\omega}_A(0)| & 0 \\
0 & |\widehat{g\omega}_B(p)-\widehat{g\omega}_B(0)|
\end{pmatrix} c_p\nonumber \\
&\leq C K_z R^2\rho \bar{a}\sum_{p \in \mathcal{P}_L}p^2 c^*_p c_p + C K_z \rho \bar{a} \sum_{p \in \mathcal{P}_H} c^*_p c_p\\
&\leq C K_z K_H^2 R^2 \rho \bar{a} \frac{n_+}{\ell^2}  + K_z K_{\ell}^2 \frac{n_+^H}{\ell^2},\label{calc:Tquad3}
 \end{align}
which, thanks to the assumptions on the parameters, can be reabsorbed in a small fraction of the spectral gap.
We have then that 
\begin{multline}\label{def:Eomega}
Z_2^{\text{ext}} - \mathcal{T}_{\text{quad}} + \frac{1}{100} \mathcal{G}_{\text{gap}} \geq \sum_{p \in \mathcal{P}_L}c^*_p \cdot (T_{\omega} - S_{\omega}) c_p \\
=\widehat{g\omega}_{AB}(0) \sum_{p \in \mathcal{P}_L} c^*_p \cdot\begin{pmatrix}
-\rho_{z_B}  & \sqrt{\rho_{z_A} \rho_{z_B}} \\
\sqrt{\rho_{z_A} \rho_{z_B}}  & -\rho_{z_A}
\end{pmatrix} c_p =: E_{\omega}.
\end{multline}
We observe that $E_{\omega}$ is semidefinite negative. 
\begin{itemize}

\item If $\eta=0$, then using that $\widehat{g\omega}_{AB} (0)\leq \widehat{g}_{AB}(0) \leq 8\pi \bar{a}$,
\begin{equation}\label{eq:whereGetaisneeded}
E_{\omega} + \frac{1}{2}\mathcal{G}_{\eta} \geq - \frac{32 \pi\bar{a}}{|\Lambda|} |z|^2 n_+ + \frac{32\pi \bar{a}}{|\Lambda|} |z|^2 n_+ \geq 0.
\end{equation}

\item If $\eta \neq 0$, recalling the expression of $\delta_{AB}$ from \eqref{eq:defdeltas}, we can bound $E_{\omega}$ as 
\begin{equation}\label{calc:Tquad4}
E_{\omega}  \geq -C \rho K_z  \widehat{v \omega}_{AB}(0) n_+ \geq - C \rho K_z\delta_{AB} n_+ = - C (K_{\ell}^{2}K_z \delta_{AB}\bar{a}^{-1}) \frac{n_+}{\ell^2},
\end{equation}
which, thanks to the assumption \eqref{assmpt:deltaKz}, can be reabsorbed in a small fraction of the spectral gap.
\end{itemize}
We conclude, using assumption \eqref{cond:aRa} and joining together \eqref{eq:controlTrest} , \eqref{calc:Tquad1}, \eqref{calc:Tquad2}, \eqref{calc:Tquad3}, \eqref{calc:Tquad4}, that 
\begin{equation}\label{Tcomm:Z2}
Z_2^{\text{ex}} - \mathcal{T}_{\text{comm}} +\varepsilon \mathcal{G}_{\text{gap}}\geq -C K_z K_{\ell}( K_H^{-1} +   K_H^2 (\rho \bar{a}^3)^{1/2}) )\frac{n_+}{\ell^2} - CK_zK_{\ell}^2 \frac{n_+^H}{\ell^2}.
\end{equation}
The terms on the r.h.s., since $K_{\ell}^4K_z^2\leq C K_H$,  and since $K_{\ell} K_z K_H \ll (\rho \bar{a}^3)^{-1/2}$, can be reabsorbed in a small fraction of the spectral gap.
For what concerns the remaining normal ordered term
\begin{align*}
\mathcal{T}_{\text{norm}} = &\sum_{\substack{k \in \mathcal{P}_H\\ p,s \in \mathcal{P}_L^{\mathbb{Z}}}}\frac{C_H\sigma(p,k) \sigma(s,k)}{k^2|\Lambda|^2} \\  &\times \big( (|z_A|^2 \widehat{g}_A^2 + |z_B|^2 \widehat{g}_{AB}^2 ) a^*_p a^*_{s-k} a_{p-k}  a_s +(|z_A|^2 \widehat{g}_A  + |z_B|^2 \widehat{g}_B)\widehat{g}_{AB} a^*_p b^*_{s-k} a_{p-k}  b_s \\
&+ (|z_A|^2 \widehat{g}_A + |z_B|^2 \widehat{g}_B)\widehat{g}_{AB} b^*_p a^*_{s-k} b_{p-k} a_s +(|z_A|^2 \widehat{g}_{AB}^2 +|z_B|^2 \widehat{g}_B^2) b_p^* b^*_{s-k} b_{p-k}  b_s \big),
\end{align*}
we bound it by a Cauchy-Schwarz
\begin{align*}
\mathcal{T}_{\text{norm}} \leq \frac{C}{|\Lambda|} (\rho_{z_A} + \rho_{z_B}) \bar{a}^2 \sum_{k \in \mathcal{P}_H} \frac{1}{k^2}  \sum_{p,s \in \mathcal{P}_L^{\mathbb{Z}}} &(a^*_p a^*_{s-k} a_{s-k} a_p + a^*_p b^*_{s-k} b_{s-k} a_p\\
&+b^*_p a^*_{s-k} a_{s-k} b_p + b^*_p b^*_{s-k} b_{s-k} b_p),
\end{align*}
and using that $|k| \geq K_H \ell^{-1}$ and $\rho_z \leq K_z\rho$, and that $|\mathcal{P}_L^{\mathbb{Z}}| = K_H^3$, 
\begin{equation}\label{Tnorm:final}
\mathcal{T}_{\text{norm}} \leq \frac{C}{|\Lambda|} \rho \bar{a}^2 K_z K_H \ell^{2} n_+ n_+^L =C K_z K_H K_{\ell}^2 \frac{n_+}{\ell^2} \frac{n_+^L}{\mathcal{M}} \frac{\mathcal{M}}{\rho \ell^3},
\end{equation}
which can be reabsorbed in the spectral gap thanks to the condition $\mathcal{M}\leq \rho \ell^3 K_H^{-2}K_{\ell}^{-4}K_z^{-3}$. The inequalities \eqref{Tcomm:Z2} and \eqref{Tnorm:final} give us \eqref{eq:estimateT1}, which concludes the proof.
\end{proof}

The last proposition together with \eqref{eq:KKdiag+S}, \eqref{eq:SZ0sum} let us then give a bound on the Hamiltonian $\mathcal{R}(z) = \mathcal{K} +Z_0(z)+ Z^{ex}_2(z) + Z_{3,L}(z)$ in the region $|z|^2 \leq K_z (n+m).$ Assuming the conditions on the parameters from Proposition \ref{prop:killQ3}, and that $K_z^{5} \ll K_{\ell}^2$ for $\eta \neq  0$, we have 
\begin{equation}
|\mathcal{E}| \leq C (\rho \bar{a})^{5/2}\ell^3 (\rho \bar{a}^3)^{\eta}+\frac{1}{2}\mathcal{G}_{\eta}.
\end{equation} 

\begin{corollary}\label{cor:lowerboundrightregion}
Under the assumptions of Proposition \ref{prop:killQ3} and assuming $K_z^{5} \ll K_{\ell}^2$ for $\eta \neq  0$, we have that, for $|z|^2 \leq K_z (n+m)$,
\begin{equation}
\mathcal{R}(z)+ \varepsilon\mathcal{G}_{\text{gap}}(z)+ \mathcal{G}_{\eta}(z) \geq |\Lambda| G_{AB}^{5/2} I_{AB}  - C |\Lambda| (\rho \bar{a})^{5/2}K_{\ell}^{-1}.
\end{equation}
where
\begin{align}
G_{AB}(\rho_{z_A},\rho_{z_B}) &= (\rho_{z_A}^2 a_A^2 + 2 \rho_{z_A} \rho_{z_B} a_{AB}^2 + \rho_{z_B}^2 a_B^2)^{5/4},\\
I_{AB}(\rho_{z_A},\rho_{z_B}) &= (8\pi)^{5/2}\frac{2\sqrt{2}}{15\pi^2} (\mu_+^{5/2}(\rho_{z_A},\rho_{z_B}) + \mu_-^{5/2}(\rho_{z_A},\rho_{z_B})).
\end{align}
\end{corollary}

\appendix

\section{Diagonalization of the Bogoliubov Hamiltonian}\label{app:bogint}

In this appendix we diagonalize the Bogoliubov Hamiltonian, a result which is going to be used both in the upper and lower bounds.

\begin{lemma}\label{lem:diagBog}
Let $\tau_k \in \{ k^2, k^2-\frac{\pi}{2\ell^2} -\frac{K_H}{\ell^2}\one_{k \in \mathcal{P}_H}\}$ and $\rho_A,\rho_B >0$ two positive parameters. Let 
\begin{equation}
\mathcal{A}_k = \tau_k \one_2 + \mathcal{B}_k, \qquad \mathcal{B}_k = \begin{pmatrix} 
\rho_A \widehat{g}_A(k) & \sqrt{\rho_A \rho_B}\, \widehat{g}_{AB}(k)\\
\sqrt{\rho_A\rho_B}\, \widehat{g}_{AB}(k) & \rho_B \widehat{g}_B(k)
\end{pmatrix}.
\end{equation}
Then the following equivalence holds, introducing $d_k = c_k + \beta_k c^*_k$,
\begin{equation}
\sum_{k \in {\Lambda}^*_+}\Big(  c^*_k\cdot \mathcal{A}(k) c_k +\frac{1}{2} \big( c_k \cdot\mathcal{B}(k) c_k + c^*_k \cdot\mathcal{B}(k) c^*_k \big) \Big) =  \mathcal{K}^{\text{diag}}+ \mathcal{S}, \qquad \mathcal{K}^{\text{diag}}:=\sum_{k \in \Lambda^*}d^*_k \cdot \mathcal{D}_k d_k \geq 0,
\end{equation}
where $\mathcal{D}_k = U_k^* \mathcal{D}_k^{\text{diag}} U_k$ and $\beta_k = U^*_k \beta^{\text{diag}}_k U_k$, with $U_k \in \mathcal{O}(2)$ for any $k \in \Lambda^*$ defined in \eqref{unit:diagonal} and such that 
\begin{align}
\mathcal{D}^{\text{diag}}_k &= \mathrm{diag}\left(\frac{1}{2}\Big(\tau_k + \lambda_{\pm}(k) + \sqrt{\tau_k^2 + 2\lambda_{\pm}(k)\tau_k}\Big)\right)_{\pm}, \label{eq:Ddiagexpressions}\\
\beta^{\mathrm{diag}}_k &= \mathrm{diag}\left(
\frac{\lambda_{\pm}(k) }{\tau_k + \lambda_{\pm}(k) + \sqrt{\tau_k^2 + 2\lambda_{\pm}(k) \tau_k}}\right )_{\pm},\label{eq:Bdiagexpressions}
\end{align}
and $\mathcal{S}$ defined as
\begin{equation}\label{def:SBog}
\mathcal{S} := \sum_{k \in\Lambda^*} \Big( \frac{1}{2} \Big(\sqrt{\tau^2  +2\lambda_+(k) \tau_k} +\sqrt{\tau_k^2  +2\lambda_-(k) \tau_k}  \Big)- \tau_k - \frac{1}{2}(\lambda_+(k) + \lambda_-(k))\Big), 
\end{equation}
where
\begin{align}
&\lambda_{\pm}(k) = [\lambda_{\pm}(\rho_A,\rho_B)](k) \nonumber\\
&= \frac{1}{2}(\rho_{A}\widehat{g}_A(k) + \rho_{B}\widehat{g}_{B}(k)) \pm \frac{1}{2} \sqrt{ \rho_{A}^2 \widehat{g}_A^2(k) +\rho_{B}^2 \widehat{g}_B^2(k)   + 2 \rho_{A}\rho_{B}(2\widehat{g}_{AB}^2(k) - \widehat{g}_A(k)\widehat{g}_B(k))}.\label{def:lambdapm}
\end{align}
In particular $\mathcal{K} \geq \mathcal{S}$.
\end{lemma}
\begin{remark}\label{rem:extensionMspecies}
The values $\lambda_{\pm}$ represent the eigenvalues of the matrix $\mathcal{B}$, as explained below in the proof. One of the problems related to the extension of the results of this paper to $M>2$ species of bosons is the possible difficulty coming from the solution of the eigenvalues problem for a matrix of $M\times M$ entries. If it is impossible to express the eigenvalues in a closed form, a possible solution may be to simply state their existence and trying to estimate their asymptotic behavior, using it for the approximation of the constant in front of the second order of the energy expansion.
\end{remark}
\begin{proof}
By their definition, we see that by diagonalizing $\mathcal{B}$, we diagonalize $\mathcal{A}$ as well. The matrix $\mathcal{B}$ is real symmetric, therefore there exists an orthogonal matrix $U \in \mathcal{O}(2)$ such that 
\begin{equation}
U \mathcal{B} U^* = \begin{pmatrix}
\lambda_+ & 0 \\
0 & \lambda_-
\end{pmatrix} = \mathcal{B}_{\mathrm{diag}}, \qquad U \mathcal{A} U^* = \begin{pmatrix}
\tau_k + \lambda_+ & 0 \\
0 & \tau_k +\lambda_-
\end{pmatrix} = \mathcal{A}_{\mathrm{diag}},
\end{equation}
where $\lambda_{\pm}$ have been defined in the lemma. From the diagonalization algorithm, we can also find the expression of the $U$
\begin{equation}\label{unit:diagonal}
U = \frac{1}{\sqrt{\lambda_+-\rho_{z_B}\widehat{g}_B + \rho_{z_A}\rho_{z_B} \widehat{g}_{AB}^2}} \begin{pmatrix}
\lambda_+ -\rho_{z_B} \widehat{g}_B & \sqrt{\rho_{z_A}\rho_{z_B}} \widehat{g}_{AB}\\
\sqrt{\rho_{z_A}\rho_{z_B}} \widehat{g}_{AB} & \lambda_- -\rho_{z_A}\widehat{g}_A\end{pmatrix}.
\end{equation}
Recalling that $c_k = (a_k, b_k)=:(c^{(1)}_k, c_k^{(2)})$, we introduce new vectors of creation and annihilation operators and the operators $\mathcal{D}$ and $\beta$ to be determined later such that 
\begin{equation}
d_p = c_p + \beta_p \cdot c^*_p,
\end{equation}
and 
\begin{align}
d^* \cdot \mathcal{D} d  &= c^* \cdot \mathcal{D} c + c^* \cdot \beta \mathcal{D} c^* + c \cdot\mathcal{D} \beta c + c \cdot \beta \mathcal{D} \beta c^* \nonumber \\
&= c^* \cdot(\mathcal{D} + \beta \mathcal{D} \beta) c + c^* \cdot \beta \mathcal{D} c^* + c \cdot \mathcal{D} \beta c + \mathrm{Tr}(\beta \mathcal{D}\beta)\nonumber \\
&= c^*\cdot \mathcal{A} c  + \frac{1}{2} \Big(c^*\cdot \mathcal{B} c^* + c \cdot \mathcal{B} c\Big)  + \mathrm{Tr}(\beta \mathcal{D}\beta), \label{eq:ABDalpha}
\end{align}
where between the first and second line we used the canonical commutation relations of $a,b$ such that $[c^{(j)}_p , c^{(k)*}_q] = \delta_{j,k}\delta_{p,q}$ for $j,k \in \{A,B\}$.

The condition \eqref{eq:ABDalpha} imposes the following equations
\begin{equation}\label{condition:ABDalpha}
2 \mathcal{D}\beta = \mathcal{B}, \qquad \mathcal{D} + \beta \mathcal{D} \beta = \mathcal{A}.  
\end{equation}
From the first equation we can substitute $\beta = \frac{1}{2} \mathcal{B}\mathcal{D}^{-1}$ in the second equation, assuming $\mathcal{D}$ is a positive definite matrix,
\begin{equation}
4 \mathcal{D}^3 - 4 \mathcal{D} \mathcal{A} \mathcal{D} + \mathcal{B} \mathcal{D} \mathcal{B} = 0. 
\end{equation}
Applying the orthogonal transformation $U$ we obtain an equation for diagonal matrices with incognita $\mathcal{D}_{diag} := U \mathcal{D} U^*$, which gives the solution \eqref{eq:Ddiagexpressions}.

Conditions \eqref{condition:ABDalpha} implies also that $\beta$ is diagonalizable by the same transformation: $\beta_{\mathrm{diag}} := U \beta U^*$ which gives expression \eqref{eq:Bdiagexpressions}.
This characterizes the solutions of \eqref{eq:ABDalpha}, and inserting them in the equation we obtain
\begin{equation}
c^*\cdot  \mathcal{A} c  + \frac{1}{2} \Big(c^* \cdot \mathcal{B} c^* + c\cdot \mathcal{B} c\Big)   = d^* \cdot \mathcal{D} d  -\mathrm{Tr}( \mathcal{D}_{\mathrm{diag}}\beta^2_{\mathrm{diag}}),
\end{equation} 
where we used the ciclicity of the trace by the action of $U$.
By the fact that 
\begin{equation}
\mathrm{Tr}( \mathcal{D}_{\mathrm{diag}}\beta^2_{\mathrm{diag}}) = \sum_{k \in \Lambda^*}\Big( \tau_k + \frac{1}{2}(\lambda_+(k) + \lambda_-(k)) - \frac{1}{2} \Big(\sqrt{\tau_k^2  +2\lambda_+(k) \tau_k} +\sqrt{\tau_k^2  +2\lambda_-(k) \tau_k}  \Big)\Big)
\end{equation}
and that $\mathcal{D}$ is positive definite, we obtain the result.
\end{proof}

\section{Calculation of the Bogoliubov integral}\label{app:calcBogint}

In this section we show how the sum obtained in the diagonalization of the Bogoliubov Hamiltonian gives the Lee-Huang-Yang type integral which corresponds to the desired second order correction to the energy. Here $\rho_{z_A},\rho_{z_B}>0$ are positive parameters such that $\rho_z:= \rho_{z_A}+\rho_{z_B} \leq K_z \rho$, and 
\[K_z=(\rho \bar{a}^3)^{-\nu}, \qquad K_{\ell}= (1000C)^{-1}(\rho \bar{a}^3)^{-2\eta},\] 
for $\nu, \eta >0$. We introduce $\tau_k \in \{k^2, k^2-\frac{\pi}{2\ell^2}-\frac{K_H}{\ell^2}\one_{k \in \mathcal{P}_H}\}$ and
\begin{align}
\mathcal{S}_0 &= \mathcal{S} + Z_0,\\
\mathcal{S} &= \frac{1}{2}\sum_{k \in\Lambda^*} \Big(\sqrt{\tau_k^2  +2\lambda_+(k) \tau_k} +\sqrt{\tau_k^2  +2\lambda_-(k) \tau_k}   - 2 \tau_k -\lambda_+(k) - \lambda_-(k) \Big) ,\\
Z_0 &= \sum_{k \in \Lambda^*}\frac{\rho_{z_A}^2 \widehat{g}_A^2(k)+ 2 \rho_{z_A}\rho_{z_B} \widehat{g}_{AB}^2(k) + \rho_{z_B}^2 \widehat{g}_B^2(k)}{4 \tau_k} = \sum_{k \in \Lambda^*}\frac{\lambda_+^2(k)+\lambda_-^2(k)}{4\tau_k},
\end{align}
where $\lambda_{\pm}$ are defined in \eqref{def:lambdapm}.

\begin{lemma}\label{lem:calcBogInt}
There exists a constant $C>0$ such that the following equivalence holds:
\begin{align*}
&\mathcal{S}_0 =  |\Lambda|(\rho_{z_A}^2 a_A^2 + 2 \rho_{z_A} \rho_{z_B} a_{AB}^2 + \rho_{z_B}^2 a_B^2)^{5/4} I_{AB}(\rho_{z_A},\rho_{z_B})+\mathcal{E},\\
&I_{AB}(\rho_{z_A},\rho_{z_B}) = (8\pi)^{5/2}\frac{2\sqrt{2}}{15\pi^2} (\mu_+^{5/2} (\rho_{z_A},\rho_{z_B})+ \mu_-^{5/2}(\rho_{z_A},\rho_{z_B})),
\end{align*}
with $\mu_{\pm} =\mu_{\pm}(\rho_{z_A},\rho_{z_B})$ as defined in \eqref{appdef:mupm}, and we estimate in two different ways the error term $\mathcal{E}$, for the upper bound:
\begin{align*}
|\mathcal{E}| \leq C(\rho_z \bar{a})^{5/2} \ell^3 &\Big(\big(\ell^{-1} (\rho_z \bar{a})^{-1/2}\log(\ell \bar{a}^{-1}) +  (\rho_z \bar{a}^3)^{2\eta} \\
&\quad + R^2(\rho_z \bar{a})(\rho_z \bar{a}^3)^{-6\eta}\big) \one_{\eta\neq 0} +  \one_{\eta = 0}\Big),
\end{align*}
and for the lower bound:
\begin{equation*}
|\mathcal{E}| \leq 
C (\rho \bar{a})^{5/2}|\Lambda| \big( R^2 K_z^{7/2} K_{\ell}^3 \rho \bar{a} + K_z^{5/2}K_{\ell}^{-1}\big)\one_{\eta\neq 0}  + \frac{1}{2}\mathcal{G}_{\eta},
\end{equation*}
the expression for $\mathcal{G}_{\eta}$ being \eqref{eq:Gconv(z)}.
\end{lemma}

\begin{proof}
We introduce the function 
\begin{equation}
G(x,y) = \sqrt{x^2+2xy} -x-y  + \frac{y^2}{2x}, \qquad x>0,\quad y \geq -\frac{1}{2}x,
\end{equation}
and observe that 
\begin{equation}
\mathcal{S}_0 = \frac{1}{2}\sum_{p \in \Lambda^*}  \Big( G(\tau_p, \lambda_+(p)) +  G(\tau_p, \lambda_-(p))\Big),
\end{equation}
and we want to approximate it by 
\begin{equation}
\mathcal{S}_{AB}:= \frac{\ell^3}{2(2\pi)^3}\int_{\mathbb{R}^3}\,\Big(G(k^2, \lambda_+(k)) +  G(k^2, \lambda_-(k))\Big)\mathrm{d}k.
\end{equation}
We distinguish the two cases $\eta =0$ and $\eta >0$. If $\eta = 0$, we observe that both $|\mathcal{S}_0|, |\mathcal{S}_{AB}|\leq C (\rho_z \bar{a})^{5/2}\ell^3$, and therefore
\begin{equation}
|\mathcal{S}_0 - \mathcal{S}_{AB}| \leq |\mathcal{S}_0| + |\mathcal{S}_{AB}| \leq C(\rho_z\bar{a})^{5/2} \ell^3 \leq K_{\ell}^{-1} (\rho_z\bar{a})^{5/2} \ell^3,
\end{equation}
which can be bounded by $ \frac{1}{2} \mathcal{G}_{\eta}$.
For the case $\eta > 0$, we have the bound
\begin{equation}
|\partial_x G(x,y)| \leq \begin{cases}
C \sqrt{\frac{|y|}{x}}, \quad &\text{ if } x \leq 2|y|,\\
C\frac{y^2}{x^2}, \quad &\text{ if } x> 2|y|.
\end{cases}
\end{equation}
We split the integral in boxes $B_p$ centered at $p$ and of size $\frac{\pi}{\ell}$,
\begin{align*}
|\mathcal{S}_0 - \mathcal{S}_{AB}| &\leq C|\Lambda|\sum_{\pm}\sum_{p \in \Lambda^*} \int_{B_p} \mathrm{d}k\, |G(\tau_p,\lambda_{\pm}) - G(k^2,\lambda_{\pm})|\\
&\leq C \sum_{\pm}\sum_{p \in \Lambda^*} |\partial_x G(p^2,\lambda_{\pm}) | \frac{|p|}{\ell},
\end{align*}
where we used that $|\tau_p - k^2| \leq C \frac{|p|}{\ell}$. Splitting the sum in momenta on the intervals $|p| \leq 2 \sqrt{\rho_z \bar{a}}, 2\sqrt{\rho_z \bar{a}} < |p| < \bar{a}^{-1}, |p| \geq a^{-1}$, we have the following bounds
\begin{itemize}
\item $|p|\leq 2\sqrt{\rho_z \bar{a}}$:
\begin{equation}
 C \sum_{\pm}\sum_{|p|\leq 2 \sqrt{\rho_z \bar{a}}} |\partial_x G(p^2,\lambda_{\pm}) | \frac{|p|}{\ell} \leq \frac{C}{\ell}|\Lambda| (\rho_z \bar{a})^2, 
\end{equation}
\item $2\sqrt{\rho_z \bar{a} }< |p| < \bar{a}^{-1}$,
\begin{equation}
 C \sum_{\pm}\sum_{2\sqrt{\rho_z \bar{a}} < |p| < \bar{a}^{-1}} |\partial_x G(p^2,\lambda_{\pm}) | \frac{|p|}{\ell} \leq \frac{C}{\ell}|\Lambda| (\rho_z \bar{a})^2 \log(\ell \bar{a}^{-1}), 
\end{equation}
\item $|p| \geq \bar{a}^{-1}$: we use that $\lambda_+^2(p) + \lambda_-^2(p)=\frac{1}{2}(\rho_{z_A}^2\widehat{g}_A^2(p) + 2\rho_{z_A}\rho_{z_B} \widehat{g}_{AB}^2(p) +\rho_{z_B}^2\widehat{g}_B^2(p)) $ and Lemma \ref{lem:reconstructiongomega} to reconstruct the $\widehat{g\omega}$'s to bound
\begin{equation}
C \sum_{\pm}\sum_{|p| \geq  \bar{a}^{-1}} |\partial_x G(p^2,\lambda_{\pm}) | \frac{|p|}{\ell} \leq  \frac{C}{\ell} |\Lambda|(\rho_z \bar{a})^2.
\end{equation}
\end{itemize}

All together they finally give, for $\eta >0$,
\begin{equation}
|\mathcal{S}_0 - \mathcal{S}_{AB}| \leq \frac{C}{\ell}|\Lambda| (\rho_z \bar{a})^2 \log(\ell \bar{a}^{-1}),
\end{equation}
which, by the definition of $K_z$, is smaller than the error in the statement of the lemma.
Now we want to approximate $\mathcal{S}_{AB}$ with $\mathcal{T}_0$:
\begin{equation}\label{def:T0}
\mathcal{T}_{0} = \frac{|\Lambda|}{2(2\pi)^3} \int_{\mathbb{R}^3} \mathrm{d}k\;\Big(G(k^2,\lambda_+(0))+G(k^2,\lambda_-(0))\Big).
\end{equation}
By a similar strategy we split the integration into boxes and use the bound
\begin{equation}
|\partial_y G(x,y)| \leq \begin{cases}
C \frac{|y|}{x}, \quad &\text{ if } x \leq 2|y|,\\
C\frac{y^2}{x^2}, \quad &\text{ if } x> 2|y|,
\end{cases}
\end{equation}
and $|g_{\#}(k)-g_{\#}(0)|\leq CR^2\widehat{g}_{\#}(0)k^2$ for $\# \in \{A,B,AB\}$, giving
\begin{align*}
|\mathcal{T}_{0}-\mathcal{S}_{AB}| &\leq C|\Lambda|\sum_{\pm}\sum_{p \in \Lambda^*} \int_{B_p} \mathrm{d}k\, |\partial_y G(k^2,\lambda_{\pm}(0))| |\lambda_{\pm}(k) - \lambda_{\pm}(0)| \\
&\leq C\sum_{\substack{|p| \leq K_{\ell}\sqrt{\rho_z\bar{a}}\\\pm}} \frac{\rho_z^2 \bar{a}}{|p|^2} |\lambda_{\pm}(p)- \lambda_{\pm}(0) | + C\sum_{|p| > K_{\ell}\sqrt{\rho_z\bar{a}}} \frac{(\rho_z \bar{a})^2}{|p|^4} \rho_z \bar{a} \\
&\leq C  R^2  K_{\ell}^3 (\rho_z \bar{a})^{7/2}|\Lambda| + C (\rho_z \bar{a})^{5/2} K_{\ell}^{-1} |\Lambda|.
\end{align*}

We are left with 
\begin{align}
\mathcal{T}_{0} = \frac{|\Lambda|}{2(2\pi)^3} \int_{\mathbb{R}^3} \mathrm{d}k\;&\Big(\sqrt{k^4  +2\lambda_+(0) k^2} +\sqrt{k^4  +2\lambda_-(0) k^2}  \nonumber\\
&\quad - 2 k^2 -\lambda_+(0) - \lambda_-(0)  + \frac{\lambda_+^2(0) + \lambda_-^2(0)}{2 k^2}\Big).\label{def:Tab}
\end{align}
We now change variable $k \mapsto  (\lambda_+^2(0)+\lambda_-^2(0))^{-1/2}k$ in the integral, introducing
\begin{align}
\mu_{\pm}&:=\frac{\lambda_{\pm}(0)}{\lambda_+^2(0)+\lambda_-^2(0)} =  \frac{1}{2} \big( \sqrt{1+\xi_{AB}} \pm \sqrt{1-\xi_{AB}}\big), \label{appdef:mupm}\\
\xi_{AB} &:= \frac{2 \rho_A \rho_B (a_A a_B - a_{AB}^2)}{\rho_A^2 a_A^2 + 2\rho_A \rho_B a_{AB}+ \rho_B^2 a_B^2},\label{appdef:xiab}
\end{align}
we get $\mathcal{T}_{0} =(\rho_A^2 a_A^2 + 2 \rho_A \rho_B a_{AB}^2 + \rho_B^2 a_B^2)^{5/4} I_{AB}$, with $I_{AB}$ defined below, and using that $1 = \mu_+^2+\mu_-^2$,
\begin{align}
I_{AB}&:=  \frac{ (8\pi)^{5/2}}{2(2\pi)^3} \int_{\mathbb{R}^3} \mathrm{d} k \,  \Big(\sqrt{k^4  +2\mu_+ k^2} +\sqrt{k^4  +2\mu_- k^2} + \frac{1}{2k^2}- 2k^2 - (\mu_+^2 + \mu_-^2)\Big)\nonumber\\
&= \frac{(8\pi)^{5/2}2\sqrt{2}}{15\pi^2} (\mu_+^{5/2} + \mu_-^{5/2}).\label{eq:intcalcmupm}
\end{align}
\end{proof}

\section{Localization of large matrices}\label{app:largematrix}
We isolate the parts of the Hamiltonian $d_{i,j}$ which change the action of $n_+^{AL}$ and $n_+^{BL}$ on the states by $i,j \in \{1,2\}$, respectively:
\begin{align*}
d_{1,0} &:= \frac{1}{2} \sum_{i \neq j } Q_i^{AL} (1-Q_j^{AL}) v_{A}(x_i-x_j) [Q_i^{AL} Q_j^{AL} + (1-Q_i^{AL}) (1-Q_j^{AL})] + h.c.\\
&\quad  + \sum_{i=1}^{n} \sum_{j=1}^{m} Q_i^{AL} (1-Q_j^{BL})v_{AB}(x_i-y_j) (1-Q_i^{AL}) (1-Q_j^{BL})+ h.c.,\\
d_{2,0} &:= \frac{1}{2}\sum_{i \neq j} Q_i^{AL} Q_j^{AL} v_A(x_i-x_j) (1-Q_i^{AL}) (1-Q_j^{AL}) + h.c.,\\
d_{0,1} &:= d_{1,0}[A \leftrightarrow B]; \quad d_{0,2} := d_{2,0}[A \leftrightarrow B]; \quad d_{1,2} = 0 = d_{2,1};\\
d_{1,1} &:=  \sum_{i=1}^{n} \sum_{j=1}^{m} Q_i^{AL} Q_j^{BL}v_{AB}(x_i-y_j) (1-Q_i^{AL}) (1-Q_j^{BL})+ h.c. \\
&\quad +\sum_{i=1}^{n} \sum_{j=1}^{m} Q_i^{AL} (1-Q_j^{BL}) v_{AB}(x_i-y_j) (1-Q_i^{AL})Q_j^{BL}  + h.c. ,
\end{align*}
where we denoted by $d_{i,j}[A \leftrightarrow B]$ the term $d_{i,j}$, where we swapped the $A$'s with the $B$'s.

The next lemma shows that we can restrict the action of the Hamiltonian to states with a bounded number of low-momenta excitations of type $A$ and $B$, with an error dependent on the $d_{i,j}$.

\begin{lemma}\label{lem:locLargeMatrix}
Let $\theta: \mathbb{R}^2 \rightarrow [0,1]$ be a smooth, compactly supported function such that $\theta(s,t) = 1$, if $|s|,|t|\leq \frac{1}{8}$, and $\theta(s,t) = 0$, if $|s| > \frac{1}{4}$ or $|t| > \frac{1}{4}$, and define, for $c,\mathcal{M} >0$,
\begin{equation}\label{def:thetatilde}
\tilde{\theta}(s,t) = c\, \theta \Big( \frac{s}{\mathcal{M}}, \frac{t}{\mathcal{M}}\Big), \qquad \sum_{(s,t) \in \mathbb{Z}^2} \tilde{\theta}(s,t)^2 = 1.
\end{equation}
Then, there exists a constant $C >0$, depending only on $\theta$, such that, for any normalized $\Psi$,
\begin{equation}
\langle \Psi, H_{n,m} \Psi\rangle  \geq   \sum_{(s,t) \in \mathbb{Z}^2}\langle \Psi_{(s,t)}, H_{n,m} \Psi_{(s,t)} \rangle -  \frac{C}{\mathcal{M}^2}\sum_{h,k = 0,1,2} \langle \Psi, d_{h,k}\Psi\rangle,
\end{equation}
where $\Psi_{(s,t)} = \tilde{\theta}(n_+^{AL}-s, n_+^{BL}-t)\Psi$.
\end{lemma}
\begin{proof}
We follow the strategy in \cite[Lemma 4.3]{freeEnCPHM} adapting it to the mixture of two types of bosons.
We can write the Hamiltonian as $H_{n,m} = \sum_{h,k \in \mathbb{Z},|h|,|k| \leq 2} \mathcal{H}^{(h,k)}$, where $\mathcal{H}^{(h,k)}$ is composed by the terms which change $(n_+^{AL},n_+^{BL})$ by $(h,k)$, that is,
\begin{equation}\label{eq:propn+Hk}
\mathcal{H}^{(h,k)} n_+^{AL} = (n_+^{AL} + h) \mathcal{H}^{(h,k)}, \qquad \mathcal{H}^{(h,k)} n_+^{BL} = (n_+^{BL} + k) \mathcal{H}^{(h,k)},  
\end{equation}
It is easy to see that
\begin{equation}\label{eq:sumd=Hk}
d_{h,k} = \sum_{\sigma, \tau \in \{\pm 1\} } \mathcal{H}^{(\sigma h, \tau k)}.
\end{equation}

Furthermore, by property \eqref{eq:propn+Hk}, we have that
\begin{align}
&\sum_{(s,t) \in \mathbb{Z}^2} \langle \Psi_{(s,t)}, \mathcal{H}^{(h,k)} \Psi_{(s,t)}\rangle  \nonumber \\
&=  \Big\langle \Psi, \sum_{(s,t) \in \mathbb{Z}^2} \tilde{\theta}(n_+^{AL}-s,n_+^{BL}-t) \tilde{\theta}(n_+^{AL}-s+ h,n_+^{BL}-t+ k)\mathcal{H}^{(h,k)}\Psi\Big\rangle \nonumber\\
& =  \sum_{(s,t) \in \mathbb{Z}^2} \tilde{\theta}(s,t) \tilde{\theta}(s - h,t- k) \langle \Psi, \mathcal{H}^{(h,k)} \Psi \rangle, \label{eq:thetaHhk}
\end{align}
where we obtained the last line since $\sum_{(s,t) \in \mathbb{Z}^2} \tilde{\theta}(z-s,w-t) \tilde{\theta}(z-s+h,w-t+k)$ is constant for any $z,w \in \mathbb{Z}$. Now we use the following equivalence 
\begin{align}
\delta_{(h,k)} :&= \sum_{(s,t) \in \mathbb{Z}^2}\big( \tilde{\theta}(s,t) \tilde{\theta}(s-h,t-k) - \tilde{\theta}(s,t)^2 \big)\nonumber\\
&= - \frac{1}{2}\sum_{(s,t) \in \mathbb{Z}^2} \big(\tilde{\theta}(s,t) - \tilde{\theta}(s-h,t-k)\big)^2, \label{def:deltahk}
\end{align}
and that $\sum_{(s,t)\in \mathbb{Z}^2} \tilde{\theta}(s,t)^2 =1$, together with \eqref{eq:thetaHhk} to obtain
\begin{align*}
\langle \Psi, H_{n,m} \Psi\rangle  &= \sum_{(s,t) \in \mathbb{Z}^2}\langle \Psi_{(s,t)}, H_{n,m} \Psi_{(s,t)} \rangle - \sum_{h,k=0,\pm 1,\pm 2} \langle\Psi, \delta_{(h,k)} \mathcal{H}^{(h,k)}  \Psi \rangle\\
&= \sum_{(s,t) \in \mathbb{Z}^2}\langle \Psi_{(s,t)}, H_{n,m} \Psi_{(s,t)} \rangle  -  \sum_{h,k =0,1,2} \delta_{(h,k)} \langle \Psi, d_{h,k} \Psi\rangle,
\end{align*} 
where we used \eqref{eq:sumd=Hk} to obtain the last line. It remains to bound $\delta_{(h,k)}$. We observe that the sum in \eqref{def:deltahk} can be restricted to $|s|,|t| \leq \mathcal{M}/2$ due to the compact support of $\theta$. By the intermediate value theorem, and the construction of $\tilde{\theta}$, we get, since $|h|,|k|\leq 2$,
\begin{equation*}
|\delta_{(h,k)}| \leq \frac{c^2}{2} \sum_{|s|,|t|\leq \mathcal{M}/2} \Big(\frac{|h|}{\mathcal{M}} \|\partial_s\theta\|_{\infty} +\frac{|k|}{\mathcal{M}} \|\partial_t \theta\|_{\infty}\Big)^2 \leq  \frac{C}{\mathcal{M}^2} ,
\end{equation*}
where by \eqref{def:thetatilde}, we used the bound $c \leq C\mathcal{M}^{-1}$.
\end{proof}

The lemma below gives a bound on the $d_{h,k}$ terms.
\begin{lemma}\label{lem:estimdAB}
There exists a constant $C>0$ such that the following inequality holds
\begin{align*}
\sum_{h,k\in \{0,1,2\}}d_{h,k} &\leq \sum_{i \neq j} \Big(v_{A}(x_i-x_j)+v_B(x_i-x_j) + v_{AB}(x_i-x_j)\Big) \\
&\quad + C \frac{K_H^3}{\ell^3} \Big(\lVert v_A\rVert_1 n n_+^A + \lVert v_B\rVert_1 m n_+^B + \lVert v_{AB}\rVert_1 \big(n n_+^B +m n_+^A  \big)\Big).
\end{align*}
\end{lemma}

\begin{proof}
Let us start by observing that the following bound holds, expanding on a Neumann basis,
\begin{align}
&\sum_{j=1}^{n}Q_{x_j}^{AL} v_{\xi}(x_j-y) Q_{x_j}^{AL} \leq C \frac{K_H^3}{\ell^3} \lVert v_{\xi}\rVert_1 n_+^A, \qquad \xi \in \{A,AB\}, \nonumber\\
&\sum_{j=1}^m Q_{x_j}^{BL} v_{\xi}(x_j-y) Q_{x_j}^{BL} \leq C \frac{K_H^3}{\ell^3} \lVert v_{\xi}\rVert_1 n_+^B, \qquad \xi \in \{B,AB\}. \label{eq:boundsQvQ}
\end{align}

By a Cauchy-Schwarz inequality and the previous bounds we can estimate
\begin{align*}
d_{1,0} + d_{2,0} &\leq C \sum_{i \neq j}( v_A(x_i-x_j) + Q_i^{AL} v_{A}(x_i-x_j) Q_i^{AL} + Q_i^{AL}Q_j^{AL} v_A(x_i-x_j) Q_i^{AL}Q_j^{AL})\\
&\quad + C \sum_{i \neq j} ( v_{AB}(x_i-x_j) + Q_i^{AL} v_{AB}(x_i-x_j) Q_i^{AL})\\
&\leq C \sum_{i \neq j} (v_A(x_i-x_j) + 
v_{AB}(x_i-x_j) ) +C \frac{K_H^3}{\ell^3} n_+^{AL} ( n\lVert v_A\rVert_1 +  m \lVert v_{AB}\rVert_1).
\end{align*}
By analogous inequalities we can bound the remaining terms.
\end{proof}

\section{Condensation estimate in the small box}\label{app:BEC}

In this section we prove Proposition \ref{propos:BEC}. The proof is inspired by \cite{LY} and \cite{greenbook}. It is a direct consequence of the following lemma, where we derive a rough lower bound on the Hamiltonian after having singled out the spectral gaps. 

We recall that $n_+ := n_+^A + n_+^B$ and the scattering lengths matrix $\mathscr{A} = \begin{pmatrix}
a_A & a_{AB}\\ a_{AB} & a_B
\end{pmatrix}$.

\begin{lemma}\label{lem:rough_lowbound}
There exists a constant $C>0$ such that the following lower bound holds: if $(\rho \bar{a}^3)^{-\frac{1}{17}} \leq n+m \leq C \rho \ell^3$, and for any state $\Psi \in L^2_s(\Lambda^n) \otimes L^2_s(\Lambda^m)$ satisfying Assumption \ref{cond:psi_low_energ}, then
\begin{equation}
\langle H_{n,m} \rangle_{\Psi}\geq \frac{\langle n_+\rangle_{\Psi}}{\ell^2} + \frac{4 \pi }{\ell^3} \left(\begin{smallmatrix}
n\\m
\end{smallmatrix}\right) \cdot \mathscr{A} \left( \begin{smallmatrix} n \\ m \end{smallmatrix}\right)(1- C (\rho \bar{a}^3)^{\frac{1}{17}}).
\end{equation}
\end{lemma}

\begin{proof}[Proof of Proposition \ref{propos:BEC}]
Low energy states satisfy Assumption \ref{cond:psi_low_energ} and by combining this with Lemma \ref{lem:rough_lowbound} we obtain that 
\begin{equation}
\frac{\langle n_+ \rangle_{\Psi}}{\ell^2} \leq C (n+m) \rho \bar{a} (\rho \bar{a}^3)^{\frac{1}{17}},
\end{equation}
and we conclude by recalling that $\ell = K_{\ell} (\rho \bar{a})^{-1/2}$.
\end{proof}

The rest of the appendix is dedicated to the proof of the rough lower bound in Lemma \ref{lem:rough_lowbound}. 
We start by observing that, for any $\Psi \in L^2_s(\Lambda^n) \otimes L^2_s(\Lambda^m)$, by using the symmetry in exchanging position variables of type $A$ and type $B$ separately, 
\begin{align}
\langle \Psi, H_{n,m} \Psi\rangle &= T +  I \nonumber \\
&= T^{in}_A + T^{in}_B + T^{out}_A + T^{out}_B + I, \label{eq:inout_splitting}
\end{align}
where $\Omega_{\sigma,x} :=\{x \in \Lambda: \min_{2 \leq j \leq n} |x  - x_j| \geq \sigma \}$, $\sigma = (\rho \bar{a})^{-\frac{5}{17}}\bar{a}$ (and analogous definition for $\Omega_{\sigma,y}$) and
\begin{align}
T^{in}_A &= n \int_{\Lambda^{n-1} \times \Lambda^m} \mathrm{d} X_{n\setminus 1}\mathrm{d} Y_m \int_{\Omega^c_{\sigma,x_1}} \mathrm{d} x_1 \, |\nabla_{x_1} \Psi(X_n,Y_m)|^2, \\
 T^{in}_B &= m \int_{\Lambda^{n} \times \Lambda^{m-1}} \mathrm{d} X_{n}  \mathrm{d} Y_{m \setminus 1} \int_{\Omega^c_{\sigma,y_1}} \mathrm{d}y_1 \, |\nabla_{y_1} \Psi(X_n,Y_m)|^2, \\
 I &= \int_{\Lambda^{n} \times \Lambda^{m}} \mathrm{d} X_{n}  \mathrm{d} Y_{m} \Big( \frac{n}{2} \sum_{j=2}^n v_A(x_1-x_j) \nonumber\\
 &\quad + \frac{m}{2} \sum_{k=2}^m v_B(y_1-y_k) + nm v_{AB}(x_1-y_1)\Big) |\Psi(X_n,Y_m)|^2
\end{align}
and $T^{out}_A, T^{out}_B$ having a similar definition but with the integration in $x_1,y_1$, respectively, in $\Omega_{\sigma}$. We indicated by $X_n = (x_1, \ldots x_n )$, and $X_{n \setminus j} = (x_1, \ldots \hat{x}_j, \ldots x_n )$ where the $j-$th variable is missing. Using the Poincarè inequality \cite[Lemma 4.1]{greenbook} and denoting by $\langle \psi \rangle_{x_1} = 
\frac{1}{\ell^3} \int_{\Lambda}\mathrm{d}x_1 \psi(X_n, Y_m)$, 
\begin{align}
T^{out}_A + T^{out}_B &\geq \frac{Cn}{\ell^2} \int_{\Lambda^n\times \Lambda^m}\mathrm{d}X_n \mathrm{d}Y_m \,|\Psi (X_n,Y_m) - \langle \Psi(\,\cdot\,,X_{n\setminus 1},Y_m)\rangle_{x_1}|^2 \nonumber\\
&\quad + \frac{Cm}{\ell^2} \int_{\Lambda^n\times \Lambda^m}\mathrm{d}X_n \mathrm{d}Y_m \, |\Psi (X_n,Y_m)- \langle \Psi(X_{n},\,\cdot\,,Y_{m\setminus1})\rangle_{y_1}|^2 \nonumber\\
&\quad - \frac{|\Omega_{\sigma,x_1}^c|^{2/3}n}{\ell^2} \int_{\Lambda^n\times \Lambda^m} \mathrm{d}X_n \mathrm{d}Y_m\, |\nabla_{x_1}\Psi(X_n,Y_m)|^2 \nonumber\\
&\quad -\frac{|\Omega_{\sigma,y_1}^c|^{2/3}m}{\ell^2} \int_{\Lambda^n\times \Lambda^m} \mathrm{d}X_n \mathrm{d}Y_m \,|\nabla_{y_1}\Psi(X_n,Y_m)|^2 \nonumber \\
&\geq \frac{\langle n_+^A \rangle_{\Psi}}{\ell^2} + \frac{\langle n_+^B \rangle_{\Psi}}{\ell^2} - C \rho\bar{a} (n+m)(\rho \bar{a}^3)^{2/17}. \label{eq:singleout_spgap}
 \end{align}
We used the definitions \eqref{eq:defn+A}, \eqref{eq:defn+B} of $n_+^A,n_+^B$ to treat the elements in the first line, the fact that $|\Omega_{\sigma,x_1}^c| \leq C n \sigma^3, |\Omega_{\sigma,y_1}^c| \leq C m \sigma^3$, $n+m \leq C\rho \ell^3$, that $\Psi$ satisfies Assumption \ref{cond:psi_low_energ}, and therefore the upper bound on the Hamiltonian is inherited by the gradient terms because the potentials are positive. 

Using then the following inequality, choosing $\varepsilon = (\rho \bar{a}^3)^{\frac{1}{17}}$,
\begin{equation}\label{eq:estimate_via_Tinout}
\langle \Psi , H_{n,m} \Psi \rangle \geq \varepsilon T + (1- \varepsilon)(T^{in}_A + T^{in}_B + I) + (1-\varepsilon) (T^{out}_A + T^{out}_B),
\end{equation}
and by \eqref{eq:singleout_spgap}, the proof of Lemma \ref{lem:rough_lowbound} is reduced to prove a lower bound for 
\begin{equation}\label{def:boxenergy_partial}
\mathscr{E}_{n,m,\ell}:= \varepsilon T + (1- \varepsilon)(T^{in}_A + T^{in}_B + I).
\end{equation}
 We have first to further localize the problem in smaller boxes $\widetilde{\Lambda}:=\Lambda_{\tilde{\ell}} = [-\tilde{\ell}/2,\tilde{\ell}/2]^3$ of size $\tilde{\ell} = (\rho \bar{a})^{- \frac{6}{17}} \bar{a}$ and prove a lower bound there. 

\begin{lemma}\label{lem:lowerbound_smallerbox}
If $h,k\in \mathbb{N}$ are such that $h+ k \leq 2p$, where $p := 4 \frac{n+m}{\ell^3}\tilde{\ell}^3 \leq C \rho \widetilde{\ell}^3$, then 
\begin{equation}\label{eq:lowerbound_smallerbox}
\mathcal{E}_{h,k,\tilde{\ell}} \geq \frac{4\pi}{\tilde{\ell}^3} \begin{psmallmatrix} h\\k\end{psmallmatrix} \cdot \mathscr{A} \begin{psmallmatrix} h-1 \\ k-1 \end{psmallmatrix} (1- C (\rho \bar{a}^3)^{\frac{1}{17}}).
\end{equation}
\end{lemma}

\begin{proof}
We introduce the sets $A=\{1,\ldots,h\}, B=\{h+1, \ldots, h+k\}$ and use the notation $z_j = x_j$ for $j = 1, \ldots, h$, and $z_{h+j} = y_j$, for $j= 1, \ldots, k$ and introduce the following potentials
\begin{equation}
U(r) = \begin{cases}
3 (\sigma^3 - \sigma_0^3)^{-1}, \quad &\text{ if } \sigma_0 \leq r \leq \sigma,\\
0, \quad &\text{ otherwise},
\end{cases}
\end{equation}
for some $\sigma_0 < \sigma$, and 
\begin{equation}
W(z_1, \ldots, z_{h+k}) = \sum_{i \neq j}^{h+k} a_{i,j} U(z_i-z_{j}), 
\end{equation}
and $a_{i,j} = a_A$ if $i,j \in A$, $a_{i,j} =a_B$ if $i,j \in B$, and $a_{i,j} = a_{AB}$ if $i\in A$ and $j \in B$ and viceversa, respectively.
We observe that $\Omega_{\sigma,x}^c  = \bigcup_{2\leq j \leq h} B_{\sigma}(x_j)$ and $\Omega_{\sigma,y}^c  = \bigcup_{2\leq j \leq k} B_{\sigma}(y_j)$, with the ball $B_{\sigma}(z) = \{w \in \Lambda\, : \, |z-w|< \sigma\}$. We start by using Dyson's Lemma \cite{dyson} for the variable $z_i$ and use the contribution from the kinetic energy to substitute the potentials by the soft version $U$ in the star-shaped domains $B_{\sigma}(z_j)$, to obtain, for any $\Psi$  differentiable on $\widetilde{\Lambda}^{h+k}$,
\begin{equation}
\int_{B_{\sigma}(z_j)} \mathrm{d}z_i \Big(|\nabla_{z_i}\Psi|^2 +  \frac{1}{2} v_{ij} (z_i-z_j)|\Psi|^2\Big) \geq a_{i,j}\int_{B_{\sigma}(z_j)} \mathrm{d}z_i \, U(z_i-z_j)|\Psi|^2,
\end{equation}
where $v_{ij} = v_A,v_B,v_{AB}$ according to the choice of indices as for $a_{i,j}$. We apply this inequality for $z_1 = x_1$ and $z_{h+1} = y_1$, integrate over the other variables and sum in $j$ to obtain
\begin{align}
&T^{in}_A + T^{in}_B  + I \nonumber \\
&\geq h a_A \sum_{2\leq j \leq h}  \int_{\widetilde{\Lambda}^k}\mathrm{d}Y_k\int_{\widetilde{\Lambda}^{h-1}} \mathrm{d}X_{h \setminus 1} \int_{B_{\sigma}(x_j)} \mathrm{d}x_1\;U(x_1-x_j)|\Psi|^2 \nonumber\\
&\quad + k a_{B} \sum_{2\leq j \leq k} \int_{\widetilde{\Lambda}^h}\mathrm{d}X_h \int_{ \widetilde{\Lambda}^{k-1}} \mathrm{d}Y_{k\setminus 1} \int_{B_{\sigma}(y_{j})} \mathrm{d}y_1\;U(y_1-y_j)|\Psi|^2 \nonumber\\
&\quad + hk a_{AB}\int_{\widetilde{\Lambda}^{h+k-2}}\mathrm{d}X_{h \setminus 1} \mathrm{d} Y_{k \setminus 1} \Big[\int_{B_{\sigma}(y_{1})\times\widetilde{\Lambda } } +\int_{\widetilde{\Lambda} \times B_{\sigma}(x_{1})} \Big] \mathrm{d}x_1\mathrm{d}y_1 U(x_1-y_1) |\Psi|^2 \nonumber\\
&= \langle \Psi, W \Psi \rangle, \label{eq:softpotentialsubstitution}
\end{align}
where, in order to reconstruct $W$, we used separate symmetry in the variables $x_j$ and $y_j$ and the definition of $U$ (its support is smaller than the domains of integration). 

Choosing $\Psi = \tilde{\ell}^{-3h} \otimes \tilde{\ell}^{-3k}$ and calculating explicitly the integrals we obtain
\begin{equation}\label{eq:mainorderW}
\langle \Psi, W \Psi\rangle = \frac{4\pi}{\tilde{\ell}^3} (a_A h(h-1) + a_{AB} (h(k-1) + k(h-1)) + a_B k(k-1)).
\end{equation}
For later purpose, we also estimate
\begin{equation}\label{eq:estimateW2}
\langle W^2\rangle_{\Psi} \leq \frac{3 \bar{a}(h+k)}{ \sigma^3-\sigma_0^3} \langle W\rangle_{\Psi}.
\end{equation}
Plugging \eqref{eq:softpotentialsubstitution} in \eqref{def:boxenergy_partial} and using Temple's inequality \cite{temple}, \cite[(2.51)]{greenbook}, since the chosen $\Psi$ is the ground state of the Laplacian operator in $T$ we get
\begin{align*}
\mathcal{E}_{h,k,\tilde{\ell}} &\geq \varepsilon T + (1-\varepsilon) \langle \Psi , W \Psi\rangle\\
&\geq (1-\varepsilon)\langle W\rangle_{\Psi} - \frac{(1-\varepsilon)^2 \langle W^2\rangle_{\Psi}}{\varepsilon \frac{\pi}{\tilde{\ell}^2} - (1-\varepsilon) \langle W\rangle_{\Psi}}\\
&\geq (1-\varepsilon) \langle W\rangle_{\Psi} \Big( 1- S_{n,m}\Big),
\end{align*}
where we estimated by zero the term $(1-\varepsilon)^2 \langle W\rangle_{\Psi}$ at the numerator and we used \eqref{eq:estimateW2} between the second and third line, with 
\begin{equation}
S_{h,k} = \frac{3 \bar{a}(h+k)}{ \sigma^3-\sigma_0^3} \frac{1}{\varepsilon \frac{\pi}{\tilde{\ell}^2} - \langle W \rangle_{\Psi}}.
\end{equation}
We denote by $Y = \rho \bar{a}^3$ and using that $h+k \leq C \rho \tilde{\ell}^3$, we impose
\begin{equation}
\varepsilon = Y^{\alpha}, \qquad   \frac{\bar{a}}{\tilde{\ell}} = Y^{\beta} , \qquad \frac{\sigma^3-\sigma_0^3}{\tilde{\ell}^3} = Y^{\gamma},
\end{equation}
for some $\alpha, \beta, \gamma >0$. Since $\langle W\rangle_{\Psi} \geq Y^{2-3\beta} \bar{a}^{-2}$, we have then that $S_{h,k} \geq C Y^{1-\gamma-2\beta-\alpha}$, provided that $\alpha + 5\beta <2$. In order for $S_{h,k} = O(Y^{\alpha})$, we further require $2\alpha \leq 1-\gamma-2\beta$. The aforementioned conditions are satisfied with the choices
\begin{equation}
\alpha = \frac{1}{17}, \qquad \beta = \frac{6}{17}, \qquad \gamma = \frac{3}{17},
\end{equation}
which is also coherent with the choice we made for the size of the boxes $\tilde{\ell} = (\rho \bar{a})^{-\frac{6}{17}}\bar{a}$. With this choices and formula \eqref{eq:mainorderW}, we concluded the proof.
\end{proof}

We are ready to derive a lower bound for \eqref{def:boxenergy_partial} from the estimates on the small boxes of size $\widetilde{\ell}$.
For readability purposes, we introduce the following densities 
\begin{equation}
\rho_n := n \ell^{-3}, \qquad \rho_m := m \ell^{-3}, \qquad \rho_{\text{tot}} := \rho_n + \rho_m.
\end{equation}

We follow the same strategy of Appendix \ref{app:localization} with the suitable modifications to obtain the following inequality
\begin{equation}\label{eq:ineq_smallerboxloc_beginning}
\ell^{-3} \mathcal{E}_{n,m,\ell} \geq \tilde{\ell}^{-3} \inf_{\{c_{h,k}\}}\sum_{h=1}^n\sum_{k=1}^m c_{h,k}\mathcal{E}_{h,k,\tilde{\ell}}, 
\end{equation}
with the constraints
\begin{equation}
\sum_{h,k} c_{h,k}  = 1,\qquad \sum_{h,k} h c_{h,k} = \rho_n \tilde{\ell}^{3}, \qquad  \sum_{h,k} k c_{h,k} = \rho_m \tilde{\ell}^3.
\end{equation}
We fix the threshold $p = 8\rho_{\text{tot}} \tilde{\ell}^3$ and then introduce the following quantities:
\begin{equation}
s:= \sum_{h+k \leq 2 p} h \, c_{h,k} \in [1,\rho_n \tilde{\ell}^3], \qquad t:= \sum_{h+k \leq 2p} k \, c_{h,k} \in [1,\rho_m \tilde{\ell}^3].
\end{equation}
We split the analysis in two cases:
\begin{itemize}
\item Case $h+k \leq 2p$: By a Cauchy-Schwarz inequality we have the following bound
\begin{equation}
\begin{psmallmatrix} s\\ t \end{psmallmatrix} \cdot \mathscr{A} \begin{psmallmatrix} s-1\\t-1 \end{psmallmatrix} = \sum_{\substack{h+k \leq 2p\\ \tilde{h}+\tilde{k}\leq 2p}} c_{h,k} c_{\tilde{h},\tilde{k}} \begin{psmallmatrix} h \\ k \end{psmallmatrix} \cdot \mathscr{A} \begin{psmallmatrix} \tilde{h}-1\\ \tilde{k}-1\end{psmallmatrix} \leq \sum_{h+k \leq 2p} c_{h,k} \begin{psmallmatrix} h\\k \end{psmallmatrix}
\cdot \mathscr{A} \begin{psmallmatrix} h-1 \\ k-1 \end{psmallmatrix},
\end{equation}
where we used that $\sum_{h+k \leq 2p}c_{h,k} \leq 1$. We can then apply Lemma \ref{lem:lowerbound_smallerbox} to get inequality \eqref{eq:lowerbound_smallerbox} allowing us to write
\begin{equation}\label{eq:estimate_ts_sum}
\sum_{h+k \leq 2p} c_{h,k} \mathscr{E}_{h,k,\tilde{\ell}} \geq \frac{4\pi}{\tilde{\ell}^3} \begin{psmallmatrix} s\\ t \end{psmallmatrix} \cdot \mathscr{A} \begin{psmallmatrix} s-1\\t-1 \end{psmallmatrix}(1-C (\rho \bar{a}^3)^{\frac{1}{17}}).
\end{equation}

\item Case $h+k >2p$: by super-additivity of the energy we have
\begin{equation}
\mathscr{E}_{h,k,\tilde{\ell}} \geq \frac{h+k}{2p} \mathscr{E}_{p,p,\tilde{\ell}} \geq \frac{h+k}{2} \frac{4\pi}{\tilde{\ell}^3}\begin{psmallmatrix} 1\\1\end{psmallmatrix} \cdot \mathscr{A} \begin{psmallmatrix} p-1\\p-1\end{psmallmatrix} (1-C (\rho \bar{a}^3)^{\frac{1}{17}}),
\end{equation}
where in the last inequality we used Lemma \ref{lem:lowerbound_smallerbox}. Hence we get
\begin{equation}\label{eq:h+k>2p_ineq}
\sum_{h+k > 2p} c_{h,k} \mathscr{E}_{h,k,\tilde{\ell}} \geq \frac{1}{2} (\rho \tilde{\ell}^3 - s-t) \frac{4\pi}{\tilde{\ell}^3}\begin{psmallmatrix}1\\1\end{psmallmatrix} \cdot \mathscr{A} \begin{psmallmatrix} p-1\\p-1\end{psmallmatrix} (1-C (\rho \bar{a}^3)^{\frac{1}{17}}).
\end{equation}
\end{itemize}

We plug inequalities \eqref{eq:estimate_ts_sum} and \eqref{eq:h+k>2p_ineq} in \eqref{eq:ineq_smallerboxloc_beginning}, and we want  therefore to minimize the form
\begin{equation}\label{eq:paraboloid}
q_{t,s}= \begin{psmallmatrix} s\\ t \end{psmallmatrix} \cdot \mathscr{A} \begin{psmallmatrix} s-1\\t-1 \end{psmallmatrix}+ \frac{1}{2} (\rho \tilde{\ell}^3 - s-t)\begin{psmallmatrix} 1\\1\end{psmallmatrix} \cdot \mathscr{A} \begin{psmallmatrix} p-1\\p-1\end{psmallmatrix}
\end{equation}
for $1 \leq s \leq \rho_n \tilde{\ell}^3$ and $1\leq t \leq \rho_m \tilde{\ell}^3$. The matrix $\mathscr{A}$ is semidefinite positive and therefore the expression \eqref{eq:paraboloid} represents a paraboloid in the variables $(s,t)$. Since $p = 8\rho_{\text{tot}} \tilde{\ell}^3$, the minimal point of $q_{t,s}$ can be found in the area $\{s \geq \rho_n \tilde{\ell}^3\} \cap \{t \geq \rho_m \tilde{\ell}^3\}$. Due to the conditions given by the domain of $q_{s,t}$, we have then that $q_{s,t}$ is minimized for $(s,t) = (\rho_n\tilde{\ell}^3, \rho_m \tilde{\ell}^3)$.
As direct consequence, this implies that 
\begin{align}
\sum_{h,k}c_{h,k}\mathcal{E}_{h,k,\tilde{\ell}} &\geq 4\pi \tilde{\ell}^3 \begin{psmallmatrix} \rho_n\\ \rho_m \end{psmallmatrix} \cdot \mathscr{A} \begin{psmallmatrix} \rho_n-\tilde{\ell}^{-3}\\ \rho_m-\tilde{\ell}^{-3} \end{psmallmatrix}(1-C (\rho \bar{a}^3)^{\frac{1}{17}})\nonumber \\
&\geq 4\pi \tilde{\ell}^3\begin{psmallmatrix} \rho_n\\ \rho_m \end{psmallmatrix} \cdot \mathscr{A} \begin{psmallmatrix} \rho_n\\ \rho_m \end{psmallmatrix}(1-C (\rho \bar{a}^3)^{\frac{1}{17}}),
\end{align}
where the $\tilde{\ell}^{-3}$ terms have been reabsorbed in the error. Inserting this final inequality in \eqref{eq:ineq_smallerboxloc_beginning} and joining this with \eqref{eq:singleout_spgap}, \eqref{eq:estimate_via_Tinout} and \eqref{def:boxenergy_partial} concludes the proof of Lemma \ref{lem:rough_lowbound}.

\section{Localization on small boxes}\label{app:localization}

The thermodynamic limit of $ E_{N_A,N_B}/N$ does not depend on the sequence for which $N \rightarrow + \infty$, therefore we may assume that $\Omega$ can be split as $\Omega = \bigcup_{j=1}^M\Lambda_{j}$, with $M$ an integer such that $M= L^3/\ell^3$.
Let us introduce the quantities, for $n,m \in \mathbb{N}$,
\begin{align*}
&E_{n,m}^{\text{Bog}} := \frac{4\pi}{\ell^3} (n^2 a_A + 2  n m a_{AB} + m^2 a_B ) + \ell^3 \big( \big(\frac{n}{\ell^3}\big)^2 a_A^2 + 2 \frac{nm}{\ell^6}a_{AB}^2 + \big(\frac{m}{\ell^3}\big)^2 a^2_B \big)^{5/4} I_{n,m},\\
&E_{\rho_A \ell^3,\rho_B \ell^3}^{\text{Bog}} := 4\pi \ell^3(\rho_A^2 a_A + 2  \rho_A \rho_B a_{AB} + \rho_B^2 a_B ) + \ell^3 \big( \rho_A^2 a_A^2 + 2 \rho_A \rho_B a_{AB}^2 + \rho_B^2  a^2_B \big)^{5/4} I_{A,B}
\end{align*}
where
\begin{align*}
I_{n,m} &:=\frac{2\sqrt{2}(8 \pi)^{5/2}}{15\pi^2}\Big(\mu_+^{5/2}\Big(\frac{n}{\ell^3},\frac{m}{\ell^3}\Big) +\mu_-^{5/2}\Big(\frac{n}{\ell^3},\frac{m}{\ell^3}\Big)\Big), \\
 I_{AB} &:=\frac{2\sqrt{2}(8 \pi)^{5/2}}{15\pi^2}(\mu_+^{5/2}(\rho_A,\rho_B) +\mu_-^{5/2}(\rho_A,\rho_B)),
\end{align*}
and $\mu_{\pm}$ are defined in \eqref{appdef:mupm}.
\begin{proof}[Proof of Theorem \ref{thm:main} by Theorem \ref{thm:lowerloc}]
By using that $v_A,v_B,v_{AB} \geq 0$, we have that the energy is superadditive because we can neglect the interactions between particles belonging to different small boxes. Taking into account all the possible ways the $N_A+N_B$ particles can be distributed in the boxes, we have the following lower bound
\begin{equation}
L^{-3} E_{N_A,N_B} \geq \ell^{-3} \inf_{\{c_{n,m}\}} \sum_{n=1}^{N_A}\sum_{m =1}^{N_B} c_{n,m} E_{n,m} 
\end{equation}
where the infimum is over all the possible choices of $\{c_{n,m}\}$, with $c_{n,m}$ being the fraction of the $M$ boxes containing exactly $n$ particles of type $A$ and $m$ of type $B$. The coefficients satisfy the following constraints
\begin{equation}\label{eq:summation_rule_cnm}
\sum_{n,m} c_{n,m} = 1, \qquad \sum_{n,m} n \,c_{n,m} = \rho_A \ell^3, \qquad \sum_{n,m} m\, c_{n,m} = \rho_B \ell^3.
\end{equation}
To better control the distribution of the number of particles inside the boxes, we introduce two chemical potentials $\mu_A,\mu_B \geq 0$, leading to the bound
\begin{align}
L^{-3} E_{N_A,N_B} &\geq \ell^{-3} \inf_{\{c_{n,m}\}} \sum_{n=1}^{N_A}\sum_{m =1}^{N_B} c_{n,m}\widetilde{E}_{n,m} + 8 \pi (\rho_A^2 a_A + 2 \rho_A \rho_B a_{AB} +\rho_B^2 a_B), \label{eq:estim_cnm}\\
\widetilde{E}_{n,m} &= E_{n,m} -  8 \pi (\rho_A (a_A n+  a_{AB} m )+  \rho_B (a_B m + a_{AB} n)). \nonumber
\end{align}
We introduce a thresholds for the number of particles, $p= 100 C_a^2 ( \rho_A + \rho_B) \ell^3$ (recall the definition of $C_a$ from \eqref{cond:aRa}) and split the sum in two parts. 

\begin{description}
\item[Case $n + m \leq p:$] Since clearly $n,m \leq p$, we are within the assumptions of Theorem \ref{thm:lowerloc}, therefore 
\begin{equation}\label{eq:expansion_n<p_m<q}
E_{n,m} \geq E_{n,m}^{\text{Bog}}  - C \rho^2\ell^3 \bar{a}\big(\rho \bar{a}^3\big)^{1/2}(\rho \bar{a}^3)^{\eta}.
\end{equation}
We want to approximate this quantity with $
E_{\rho_A\ell^3, \rho_B \ell^3}^{\text{Bog}}$
and estimate the error done in this way. 
Introducing the vectors $w= (n,m)$ and $ v = (\rho_A \ell^3,  \rho_B \ell^3)$ and the matrices 
\begin{equation}
\mathscr{A} = \begin{pmatrix}
a_A & a_{AB}\\ a_{AB} & a_B
\end{pmatrix}, \qquad \mathscr{A}_2 = \mathscr{A}^{\odot 2} =\begin{pmatrix}
a_A^2 & a_{AB}^2\\ a_{AB}^2 & a_B^2
\end{pmatrix},  
\end{equation}  
we can write 
\begin{align}
E_{n,m}^{\text{Bog}} &:= \frac{4\pi}{\ell^3}  w \cdot \mathscr{A} w +\frac{1}{\ell^{9/2}}( w \cdot \mathscr{A}_2 w)^{5/4}I_{n,m}, \label{eq:vec_expressBognm}\\
E_{\rho_A \ell^3,\rho_B \ell^3}^{\text{Bog}} &:=  \frac{4\pi}{\ell^3} v \cdot \mathscr{A} v +  \frac{1}{\ell^{9/2}} (v\cdot \mathscr{A}_2 v)^{5/4} I_{AB},\\
 8 \pi (\rho_A (a_A n+  a_{AB} m )&+  \rho_B (a_B m + a_{AB} n)) = 2 \frac{4\pi}{\ell^3} w \cdot \mathscr{A} v. \label{eq:vec_expressBogrho}
\end{align}

Using \eqref{eq:vec_expressBognm} and \eqref{eq:vec_expressBogrho}, we complete the square observing that, since $\mathscr{A}$ is symmetric, 
\begin{equation}\label{eq:completesq_nm}
w \cdot \mathscr{A} w - 2 w \cdot \mathscr{A} v = (w-v) \cdot \mathscr{A} (w-v) - v \cdot \mathscr{A} v.
\end{equation}
Now we estimate the difference of the Bogoliubov integral parts. Using that $I_{n,m} = \mathcal{O}(1)$,
\begin{align*}
&\frac{1}{\ell^{9/2}} |(w\cdot \mathscr{A}_2 w)^{5/4} I_{n,m} - (v \cdot \mathscr{A}_2 v)^{5/4}I_{AB}| \nonumber \\ &\leq \frac{C}{\ell^{9/2}} \Big(|(w\cdot \mathscr{A}_2 w)^{5/4}- (v \cdot \mathscr{A}_2 v)^{5/4}|+ |I_{n,m}-I_{AB}| (v \cdot \mathscr{A}_2 v)^{5/4}\Big) = (I) + (II). 
\end{align*}
We estimate the term $(I)$ by a Cauchy-Schwarz inequality, for a $\varepsilon>0$, 
\begin{align*}
(I) &\leq \frac{C}{\ell^{9/2}} (v \cdot \mathscr{A}_2 v)^{1/4}|(w\cdot \mathscr{A}_2 w)- (v \cdot \mathscr{A}_2 v)|\\
&\leq \frac{C}{\ell^{9/2}} (\rho \ell^3 \bar{a})^{1/2} |\mathscr{A}^{-1/2}\mathscr{A}_2 (v+w)| | \mathscr{A}^{1/2}(v-w)| \\
&\leq \frac{C\varepsilon^{-1}}{\ell^6} \rho \ell^3 \bar{a} (\rho^2 \ell^6 \bar{a}^3) + \frac{\varepsilon}{\ell^3} (v-w)\cdot \mathscr{A}(v-w).
\end{align*}
For the second term, using that $\mu_{\pm} = \mathcal{O}(1)$, again by a Cauchy-Schwarz for $\varepsilon>0$,
\begin{align*}
(II) &\leq \frac{C}{\ell^{9/2}} (\rho \ell^3 \bar{a})^{5/2}|I_{n,m}-I_{AB}| \\
&\leq  \frac{C}{\ell^{9/2}} (\rho \ell^3 \bar{a})^{5/2} \sum_{\pm} |\mathscr{A}^{-1/2}\nabla \mu_{\pm} | \Big|\mathscr{A}^{1/2}\frac{(v-w)}{\ell^3}\Big|\\
&\leq  \frac{C\varepsilon^{-1}}{\ell^{12}\rho^2 \bar{a}}(\rho \ell^3 \bar{a})^{5} +  \frac{\varepsilon}{\ell^3}(v-w)\cdot \mathscr{A}(v-w) .
\end{align*}
Collecting the estimates for $(I)$ and $(II)$ and choosing $\varepsilon = \pi$ we get
\begin{equation}\label{eq:estimateGvGw}
\frac{1}{\ell^{9/2}} |(w\cdot \mathscr{A}_2 w)^{5/4} I_{n,m} - (v \cdot \mathscr{A}_2 v)^{5/4}I_{AB}| \leq \frac{2\pi}{\ell^3}(v-w)\cdot \mathscr{A}(v-w)  - C \ell^3 (\rho \bar{a})^{5/2}(\rho \bar{a}^3)^{\eta}.
\end{equation}

Collecting the estimates \eqref{eq:expansion_n<p_m<q}, \eqref{eq:completesq_nm}, \eqref{eq:estimateGvGw} we obtain, for $n + m \leq p$,
\begin{equation}\label{eq:finalest_nm<p}
\widetilde{E}_{n,m} \geq E_{\rho_a\ell^3, \rho_B \ell^3}^{\text{Bog}} + \frac{2 \pi}{\ell^3} (w-v) \cdot \mathscr{A} (w-v) - \frac{8\pi}{\ell^3}v \cdot \mathscr{A} v - C\ell^3(\rho \bar{a})^{5/2}(\rho \bar{a}^3)^{\eta}.
\end{equation}

\item[Case $n + m > p$:] We split the $n+m$ particles in $\left\lfloor{\frac{n+m}{p}}\right \rfloor$ subgroups of $p=s+t$ particles, with $s,t$ the number of particles of type $A,B$ in the subgroup, respectively. By super-additivity of the energy
\begin{equation}\label{eq:expansion_n+m>q}
\widetilde{E}_{n,m} \geq \left \lfloor{\frac{n+m}{p}}\right \rfloor E_{s,t}-  \frac{8\pi}{\ell^3} w \cdot \mathscr{A} v + E_{\sigma, \tau}, \qquad s+t =p, \quad \sigma + \tau < p.
\end{equation}

For $E_{\sigma,\tau}$, we use again Theorem \ref{thm:lowerloc} and have the following estimate, bounding the positive terms from below by zero:
\begin{equation}\label{eq:estimate_r}
E_{\sigma,\tau} \geq  - C \ell^3 \big(\rho \bar{a}\big)^{5/2}(\rho \bar{a}^3)^{\eta}.
\end{equation}

We introduce now the vector $\widetilde{w} = (s,t) = (s, p-s)$. Since $s,t \leq p$, we can reason like the previous case \eqref{eq:finalest_nm<p} to obtain
\begin{align}
E_{s,t} &\geq E_{\rho_A\ell^3, \rho_B \ell^3}^{\text{Bog}} + \frac{2 \pi}{\ell^3} (\widetilde{w}-v) \cdot \mathscr{A} (\widetilde{w}-v) - \frac{8\pi}{\ell^3}v \cdot \mathscr{A} v \nonumber \\
&\quad + \frac{8 \pi }{\ell^3}\widetilde{w} \cdot \mathscr{A} v - C\ell^3(\rho \bar{a})^{5/2}(\rho \bar{a}^3)^{\eta}. \label{eq:estinitial_nm>p}
\end{align}
Observe this time the presence of the chemical-potential-type term $\widetilde{w} \cdot A v$. 

W.l.o.g. we can assume $s \leq p/2$, which implies $p/2 \leq t \leq p = 100 C_a^2 (v_1 + v_2)$, the opposite case being treated similarly. This lets us derive the following estimate
\begin{align}
(\widetilde{w}-v) \cdot \mathscr{A} (\widetilde{w}-v) &\geq a_B(t-v_2)^2- 2a_{AB} v_1 (t-v_2) \nonumber \\
&\geq C_a^4 \underline{a} (2300 v_1^2 + 4600 v_1v_2 + 2400 v_2^2) \geq 2000 C_a^4 \underline{a} (v_1+v_2)^2, \label{eq:estimatePositivewAw}
\end{align}
where we used that, from condition \eqref{cond:aRa}, $a_{AB} \leq  C_a \underline{a}$. Using again assumption \eqref{cond:aRa}, we can bound
\begin{equation}\label{eq:boundchemical}
- \frac{8 \pi}{\ell^3} w \cdot \mathscr{A} v \geq -C_a \underline{a}\frac{8 \pi }{\ell^3} \max\{n,m\} \max\{v_1,v_2\}.
\end{equation}

Since 
\begin{equation}
\left\lfloor{\frac{n+m}{p}}\right \rfloor \geq \frac{n+m}{2p},
\end{equation}
using \eqref{eq:estinitial_nm>p} we obtain the following estimate for $\widetilde{E}_{n,m}$:
\begin{align*}
\widetilde{E}_{n,m} &\geq \frac{n+m}{2p} \Big(  E_{\rho_A\ell^3, \rho_B \ell^3}^{\text{Bog}} + \frac{ \pi}{\ell^3} (\widetilde{w}-v) \cdot \mathscr{A} (\widetilde{w}-v)\nonumber \\
&\quad - \frac{8\pi}{\ell^3}v \cdot \mathscr{A} v  + \frac{8 \pi }{\ell^3}\widetilde{w} \cdot \mathscr{A} v\Big) - C \ell^3(\rho \bar{a})^{5/2}(\rho \bar{a}^3)^{\eta}\\
&\quad +\frac{n+m}{2p} \frac{ \pi}{\ell^3} (\widetilde{w}-v) \cdot \mathscr{A} (\widetilde{w}-v) - \frac{8 \pi}{\ell^3} w \cdot \mathscr{A} v + E_{\sigma,\tau},
\end{align*}
where we brought half of the $(\widetilde{w}-v)\cdot \mathscr{A}(\widetilde{w}-v)$ term outside the parentheses. The term of the first two lines inside the parentheses is positive, therefore we can estimate the coefficient $(n+m)/2p$ by $1$. We then use \eqref{eq:estimatePositivewAw}, \eqref{eq:boundchemical} and \eqref{eq:estimate_r}, observing that, by \eqref{cond:aRa}, recalling that $C_a >1,$
\begin{align*}
&1000\frac{n+m}{p}\frac{\pi}{\ell^3}C_a^4 \underline{a}  (v_1 + v_2)^2 - 8 C_a \underline{a}\frac{ \pi }{\ell^3} \max\{n,m\} \max\{v_1,v_2\} \\
&\geq  10 C_a^2 \underline{a}\frac{\pi}{\ell^3}(n+m)  (v_1+v_2) -8 C_a \underline{a}\frac{ \pi }{\ell^3} \max\{n,m\} \max\{v_1,v_2\} \geq 0,
\end{align*}
we finally get
\begin{equation}	\label{eq:finalest_nm>p}
\widetilde{E}_{n,m} \geq E_{\rho_A\ell^3, \rho_B \ell^3}^{\text{Bog}}- \frac{8\pi}{\ell^3}v \cdot \mathscr{A} v - C \ell^3 \big(\rho \bar{a}\big)^{5/2}(\rho \bar{a}^3)^{\eta}.
\end{equation}
\end{description}

Collecting the previous inequalities \eqref{eq:finalest_nm<p} and \eqref{eq:finalest_nm>p} into \eqref{eq:estim_cnm}, we observe that the terms of the expansions are independent of $n$ and $m$. We use therefore \eqref{eq:summation_rule_cnm} to obtain
\begin{align*}
&L^{-3} E_{N_A,N_B} \\
&\geq \ell^{-3} \Big( E_{\rho_A\ell^3, \rho_B \ell^3}^{\text{Bog}}- \frac{8\pi}{\ell^3}v \cdot \mathscr{A} v\Big) - C \ell^3 \big(\rho \bar{a}\big)^{5/2}(\rho \bar{a}^3)^{\eta} + 8 \pi (\rho_A^2 a_A + 2 \rho_A \rho_B a_{AB} +\rho_B^2 a_B)\\
&= \ell^{-3}  E_{\rho_A\ell^3, \rho_B \ell^3}^{\text{Bog}} - C  \big(\rho \bar{a}\big)^{5/2}(\rho \bar{a}^3)^{\eta},
\end{align*}
which, by calculating the thermodynamic limit, gives the desired result.

\end{proof}

\section{Convexity of the energy functional}\label{sec:Fconvex}

In this appendix we study the convexity of the energy functional 
\begin{equation}\label{app:Fconvexfunctional}
F(r_A,r_B) = \frac{8 \pi \bar{a}}{|\Lambda|} (\rho \bar{a}^3)^{1/4}(r_A^2 + r_B^2) + |\Lambda|^{-3/2} G^{5/4} I(r_A,r_B) -\mu_A r_A -\mu_B r_B,
\end{equation}
where 
\begin{align}
G (r_A,r_B)&= r_A^2 a_A^2 + 2r_A r_B a_{AB}^2 + r_B^2 a_B^2,\\
I(r_A,r_B) &=  (8\pi)^{5/2}\frac{2\sqrt{2}}{15\pi^2} (\mu_+^{5/2} + \mu_-^{5/2}), \label{eq:intsphericalFconvex}
\end{align}
with 
\begin{equation}\label{eq:muxiABexpressions}
\mu_{\pm} = \sqrt{1+\xi_{AB}} \pm \sqrt{1-\xi_{AB}}, \qquad \xi_{AB} = \frac{2 r_Ar_B (a_A a_B - a_{AB}^2)}{r_A^2 a_A^2 +2 r_Ar_B a_{AB}^2 + r_B a_B^2}.
\end{equation}

\begin{lemma}\label{lem:Fconvexfunctional}
Let $r_A,r_B>0$ be two positive parameters satisfying $r_A + r_B \leq C K_z \rho |\Lambda|$. Then for $\rho \bar{a}^3$ small enough, the functional $F(r_A,r_B)$ is convex in $(r_A,r_B)$.
\end{lemma}
\begin{proof}
We calculate the gradient of $F$ in $r_A,r_B$:
\begin{equation}\label{eq:gradF}
\nabla F =\begin{pmatrix}
 \frac{32 \pi \bar{a}(\rho \bar{a}^3)^{1/4}}{|\Lambda|} r_A + \frac{5G^{1/4}}{2|\Lambda|^{3/2}}  (r_A a_A^2 + r_B a^2_{AB})I + \frac{G^{5/4}}{|\Lambda|^{3/2}} \partial_{r_A}I - \mu_A\\
 \frac{32 \pi \bar{a}(\rho \bar{a}^3)^{1/4}}{|\Lambda|}  r_B + \frac{5G^{1/4}}{2|\Lambda|^{3/2}}  (r_B a_B^2 + r_A a^2_{AB})I + \frac{G^{5/4}}{|\Lambda|^{3/2}} \partial_{r_B}I - \mu_B
 \end{pmatrix}
\end{equation}
and the Hessian
\begin{align*}
\partial_{r_A,r_A}^2 F &= \frac{32 \pi \bar{a}(\rho \bar{a}^3)^{1/4}}{|\Lambda|} + \frac{5I}{4|\Lambda|^{3/2}G^{3/4}} (r_A a_A^2 + r_B a_{AB}^2)^2 + \frac{5G^{1/4}a_A^2I}{2|\Lambda|^{3/2}}\\
&\quad +\frac{5 G^{1/4}}{|\Lambda|^{3/2}} (r_A a_A^2 + r_B a_{AB}^2)\partial_{r_A}I + \frac{G^{5/4}}{|\Lambda|^{3/2}} \partial^2_{r_A}I,\\
\partial^2_{r_A,r_B} F &=  \frac{5I}{4|\Lambda|^{3/2}G^{3/4}} (r_A a_A^2 + r_B a_{AB}^2)(r_B a_B^2 + r_A a_{AB}^2) + \frac{5G^{1/4}a_{AB}^2I}{2|\Lambda|^{3/2}} \\
&\quad + \frac{5 G^{1/4}}{|\Lambda|^{3/2}} (r_A a_A^2 + r_B a_{AB}^2)\partial_{r_B}I + \frac{5 G^{1/4}}{2|\Lambda|^{3/2}} (r_B a_B^2 + r_A a_{AB}^2)\partial_{r_A}I + \frac{G^{5/4}}{|\Lambda|^{3/2}} \partial_{r_B,r_A}^2 I
\end{align*}
with $\partial_{r_B, r_B}^2 F, \partial_{r_A,r_B}^2 F$ begin their symmetric versions. 
By straightforward calculations, using the expressions \eqref{eq:muxiABexpressions}, we can bound 
\begin{equation}
|\partial_{r_A} I|\leq \frac{C\bar{a}^4r_B^3}{G^2}, \qquad |\partial_{r_B} I|\leq \frac{C\bar{a}^4r_A^3}{G^2},
\end{equation}
and 
\begin{equation*}
|\partial_{r_A,r_B}^2 I|\leq \frac{C\bar{a}^2}{G^3} r_B^2(G + r_B(r_A a_{AB}^2 + r_B a_B^2)), \qquad |\partial_{r_A}^2 I| \leq \frac{C \bar{a}^2}{G^3} r_B^3 a_B^3 (r_Aa_A^2 + r_B a_{AB}^2),
\end{equation*}
and symmetric versions for $\partial_{r_B,r_A}^2 I, \partial_{r_B, r_B}^2 I$.
Therefore, in the region $r_A + r_B \leq C K_z \rho |\Lambda|$, we observe that the Hessian elements can be bounded as 
\begin{align*}
\Big|\partial_{r_A,r_A}^2 F -  \frac{32 \pi \bar{a}(\rho \bar{a}^3)^{1/4}}{|\Lambda|} \Big| &\leq \frac{C}{|\Lambda|} K_z^{1/2} \rho^{1/2}a^{5/2}, \qquad |\partial_{r_B,r_A}^2 F| \leq \frac{C}{|\Lambda|} K_z^{1/2} \rho^{1/2}a^{5/2},\\
\Big|\partial_{r_B,r_B}^2 F -  \frac{32 \pi \bar{a}(\rho \bar{a}^3)^{1/4}}{|\Lambda|} \Big| &\leq \frac{C}{|\Lambda|} K_z^{1/2} \rho^{1/2}a^{5/2}, \qquad |\partial_{r_A,r_B}^2 F| \leq \frac{C}{|\Lambda|} K_z^{1/2} \rho^{1/2}a^{5/2},
\end{align*}
and, since $K_z^{1/2}\rho^{1/2}$ is subdominant w.r.t. $\rho^{1/4}$, the following bound holds
\begin{equation}
\mathrm{Hess}F = \frac{32 \pi \bar{a}(\rho \bar{a}^3)^{1/4}}{|\Lambda|} \one_2  +  \mathcal{O}(  (K_z\rho a^5)^{1/2}|\Lambda|^{-1}) \geq 0,
\end{equation}
which gives that $F$ is convex in $(r_A,r_B)$ for $\rho \bar{a}^3 $ small enough.

\end{proof}

\bibliographystyle{plain}
\bibliography{bosmix}

\end{document}